\documentclass[12pt]{article}
\usepackage{graphicx}
\usepackage{pdfpages}
\usepackage{adjustbox}
\usepackage{amsmath,amsfonts,amssymb,amsthm}
\usepackage{microtype}

\allowdisplaybreaks
\usepackage{caption}
\usepackage[ruled,vlined]{algorithm2e}
\usepackage{bbm}
\usepackage{breakcites}
\usepackage{tabularx}
\usepackage{enumitem}
\usepackage{natbib} 
\usepackage{authblk}
\usepackage{graphicx}
\usepackage{subcaption}
\usepackage{hyperref}
\usepackage{ulem}
\usepackage{mathtools}
\usepackage{stix}
\usepackage{color}
\usepackage[title,titletoc]{appendix}

\DeclarePairedDelimiter\ceil{\lceil}{\rceil}
\DeclarePairedDelimiter\floor{\lfloor}{\rfloor}

\theoremstyle{definition}
\newtheorem{definition}{Definition}[section]
\newtheorem{assumption}{Assumption}
\newtheorem{example}{Example}

\newtheorem{theorem}{Theorem}[section]
\newtheorem{corollary}{Corollary}[section]
\newtheorem{lemma}[theorem]{Lemma}

\newtheorem{proposition}[theorem]{Proposition}

\newcommand{\blind}{1}

\newcommand{\saa}{\mathrm{saa}}
\newcommand{\mt}{\mathrm{T}}

\newcommand{\p}{\mathrm{P}}
\newcommand{\mv}{\mathrm{v}}
\newcommand{\g}{\mathrm{G}}
\newcommand{\f}{\mathrm{F}}

\newcommand{\st}{\text{s.t.}}
\newcommand{\I}{\mathbb{I}}
\newcommand{\E}{\mathbb{E}}
\newcommand{\dist}{\text{dist}}
\newcommand{\F}{\mathrm{F}}

\begin{document}

\def\spacingset#1{\renewcommand{\baselinestretch}%
{#1}\small\normalsize} \spacingset{1}

%%%%%%%%%%%%%%%%%%%%%%%%%%%%%%%%%%%%%%%%%%%%%%%%%%%%%%%%%%%%%%%%%%%%%%%%%%%%%%

\if1\blind
{
  \title{\bf A General Framework for Optimal Group Sequential Testing via Mixed-Integer Linear Programming}
  \author{Dae Woong Ham$^*$, Stefanus Jasin$^\dagger$, and Xuejun Zhao$^\diamond$}
  \maketitle
} \fi

\if0\blind
{
  \bigskip
  \bigskip
  \bigskip
  \begin{center}
    {\LARGE\bf A General Framework for Optimal Group Sequential Testing via Mixed-Integer Linear Programming}
\end{center}
  \medskip
} \fi

\bigskip
\begin{abstract}
Group sequential tests (GSTs) allow researchers
to test a hypothesis at $K$ interim analyses as
data accumulates, while controlling type-1 and
type-2 error rates across all stages. The
classical approach specifies the rejection
cutoffs through a deterministic alpha-spending
function---commonly Pocock, O'Brien-Fleming, or
Lan-DeMets---which controls error rates by
construction but is not derived from any
optimality principle. An alternative line of work
derives optimal cutoffs by reformulating GST
design as a Bayes sequential decision problem,
but requires the user to specify a prior, loss,
and sampling cost, and applies only when the
running test statistic is sufficient. We propose
S-MILP, which computes the optimal cutoffs
directly by combining sample average
approximation of the constrained design problem
with a mixed-integer linear program reformulation
solvable to global optimality. We
establish finite-sample convergence guarantees
and show through simulation that S-MILP reduces
expected sample size by 5--10\% relative to
Pocock and Lan-DeMets, and 25--35\% relative to
O'Brien-Fleming. We also find that the optimal solution typically aggressively spends the type-1 error budget early, shedding insight to a long-standing debate in the GST community. Finally, we apply S-MILP to a recent
clinical trial of an acute kidney injury
intervention. 
\end{abstract}

\noindent%
\vfill
\footnotesize
\qquad
\\
$^*$ University of Michigan Ross School of Business, Ann Arbor, MI, U.S.A. daewoong@umich.edu
\\
$^\dagger$ University of Michigan Ross School of Business, Ann Arbor, MI, U.S.A.
\\
$^\diamond$University of North Carolina at Charlotte, U.S.A
\normalsize
\newpage
\spacingset{1.9}

% DON'T change the spacing!
\section{Introduction}

\label{section:intro}

Sequential hypothesis testing is a principled framework for
performing inference on data that arrive over time. It is widely
used in clinical trials, where evidence accumulates as patients
are enrolled, and in online digital experimentation, where firms
such as Spotify and Netflix continuously evaluate product changes
against live user traffic \citep{spotify_blog1, ham_DB,
clinical_trials_GST}. The defining feature of sequential testing
is that it allows a researcher to reach a valid conclusion using
fewer samples than a fixed-sample test would require. The
motivation is both ethical and economic: in clinical trials,
continued enrollment under an inferior treatment exposes patients
to avoidable harm and wastes a limited supply of eligible
participants \citep{clinical_trials_GST, JRSSB_GST}; in online
experimentation, prolonged exposure to a degraded product imposes
direct costs on real users \citep{mypaper, spotify_blog1,
michael_paper, michael_paper2}. In both settings, the value of
stopping early once the evidence is conclusive is substantial.

Two distinct frameworks for sequential testing have emerged in
the recent literature. The first is anytime-valid testing,
which controls the type-1 error uniformly over a potentially
unbounded sequence of analyses, allowing the researcher to
monitor the data continuously and stop at any data-dependent time
\citep{howard_nonasymp, ian_bandit, Ramesh_anytime, michael_paper,
ham_DB}. This framework is well suited to settings in which
sample size is effectively unconstrained and the researcher
wishes to monitor an experiment indefinitely; large-scale online
platforms, where user traffic is plentiful, are a natural
example. The second framework, which is the focus of the present
paper, is group sequential testing (GST). Here the total
sample size $N$ is fixed in advance by budget or feasibility
constraints, and the researcher conducts up to $K$ pre-specified
analyses at cumulative sample sizes $n_1, n_1 + n_2, \ldots,
\sum_{k=1}^K n_k = N$, with the option of stopping at any one of
them \citep{spotify_blog1, obrien_fleming_GST, pocock_GST,
Lan_DeMets_GST}. GST is the dominant framework in clinical
trials, where each additional patient is costly and a maximum
trial size is mandated by the protocol, and the goal is to stop
as early as possible while maintaining valid inference at the
final analysis.

A GST procedure is specified by a sequence of
rejection cutoffs $(\theta_k)_{k=1}^K$. At stage
$k$, the researcher compares a running test
statistic $S_k$ to $\theta_k$ and rejects $H_0$ if
$S_k \geq \theta_k$. The central design question is
how to choose $(\theta_k)_{k=1}^K$ so that the
type-1 error is controlled across all $K$ stages
simultaneously. There are infinitely many valid
choices: with $K=3$ and $\alpha = 0.05$, one could
spend the budget evenly ($0.05/3$ at each stage),
front-load it (allowing more aggressive early
rejection), or back-load it (reserving most of
the budget for the final analysis), all while
maintaining cumulative type-1 error at $0.05$. The
classical literature parameterizes this choice
through an alpha-spending function, with
the proposals of \citet{pocock_GST},
\citet{obrien_fleming_GST}, and
\citet{Lan_DeMets_GST} being the most widely used;
these remain the default in clinical-trial
practice and at large technology firms
\citep{spotify_blog1}.

Alpha-spending functions are heuristics. They
control type-1 error by construction, but they
are not derived from any optimality principle.
A separate strand of work
\citep{eales_jennison_1992, hampson_jennison_2013}
addresses optimality directly by considering a
Bayes sequential decision problem: the researcher
specifies a prior on the parameter, a loss for
incorrect decisions, and a cost of sampling, and
the resulting Bayes problem is solved by dynamic
programming, with a search over the loss
parameter to recover the prescribed error rates.
The approach yields exact optimal designs, but it
imposes two practical requirements. First, the
dynamic program is tractable only when the
running test statistic is a sufficient statistic,
which excludes unknown-variance $t$-tests and
many non-Gaussian settings. Second, different
choices of the prior, loss, and sampling cost
yield different optimal designs, and there is no
canonical default---a barrier that
\citet{eales_jennison_1992} themselves
acknowledge: ``medical researchers are,
understandably, reluctant to specify [decision
theory loss functions] for clinical trials.''

A complementary contribution is
\citet{JRSSB_GST}, who treats the binary-outcome
case. Discreteness reduces the design problem to
a finite linear program, which can be solved
exactly. The approach does not extend to
continuous outcomes, where the underlying
optimization becomes infinite-dimensional.

We propose a different route. Existing
alpha-spending functions specify the rejection
cutoffs $(\theta_k)_{k=1}^K$ through a
deterministic formula, with no claim of
optimality; the Bayes decision-theoretic approach
yields optimal cutoffs but only after specifying
a prior, a loss function, and a sampling cost,
and only in settings where the test statistic is
sufficient. We instead compute the optimal
$(\theta_k)_{k=1}^K$ directly. We work with the
frequentist optimization that defines GST
design to minimize the expected sample size subject to
type-1 and type-2 error constraints. Our approach approximates the error probabilities and expected
sample size by sample average approximation (SAA),
i.e., Monte Carlo averages over simulated
test-statistic paths under the null and the
alternative. We show that the resulting
finite-sample optimization can be reformulated as
a mixed-integer linear program (MILP), which can
be solved to global optimality by standard
commercial solvers. We refer to this two-step
procedure, i.e. SAA followed by MILP reformulation, 
as the S-MILP approach. 

The error constraints
appear directly as linear inequalities in the
MILP; no prior, no loss function, no sampling
cost, and no multiplier search is required. The
user specifies only the alternative $\mu_a$ at
which power is desired---a quantity that any
sample-size calculation already requires---and
S-MILP returns the cutoffs that minimize expected
sample size at that alternative. The framework
applies whenever the test-statistic distribution
can be sampled under the null and the alternative,
a condition that holds for $z$-tests, $t$-tests,
two-sample tests, one- and two-sided tests, and a
wide class of non-normal settings within a single
unified formulation. We establish finite-sample
convergence guarantees for S-MILP. In our
experiments (Section~\ref{sec:sims}), S-MILP
reduces expected sample size by approximately
5--10\% relative to Pocock and Lan-DeMets, and
by approximately 25--35\% relative to
O'Brien-Fleming.

\paragraph{Organization of the paper.}
Section~\ref{sec:z_test} formulates our method
for the base case of known-variance $z$-tests.
Section~\ref{sec:t-tests} extends the formulation
to unknown-variance $t$-tests, and
Section~\ref{sec:convergence-gen} further extends
it to general continuous test statistics that need
not be normally distributed.

% Section~\ref{sec:futility} treats symmetric
% designs, in which the trial may stop early either
% to reject $H_0$ or to accept it when the evidence
% is strongly favorable to the null. 

Throughout
Sections~\ref{sec:z_test}--\ref{sec:convergence-gen},
we establish finite-sample convergence guarantees
showing that the optimal value of
our method converges to that
of the original oracle problem.
Section~\ref{sec:sims} compares our method with
the alpha-spending designs of Pocock,
Lan-DeMets, and O'Brien-Fleming through
simulation. Section~\ref{sec:application} applies
the method to a clinical trial of an acute kidney
injury intervention, and
Section~\ref{sec:conclusion} concludes the paper.

\paragraph{Notation.} Let $[n] = \{1, \ldots, n\}$.
We use $\mathbb{I}(E) = 1$ if the event $E$ holds and $0$
otherwise. We write $\overset{d}{=}$ for equality
in distribution and $\overset{d}{\Rightarrow}$
for convergence in distribution. The cumulative
distribution function of a standard normal random
variable, evaluated at $x \in \mathbb{R}$, is
denoted $\Phi(x)$.

\section{Sequential Analysis for $z$-tests}
\label{sec:z_test}

We begin with the case in which the data are
i.i.d.\ Normal, $X_i\sim \mathcal{N}(\mu, \sigma^2)$,
the variance $\sigma^2$ is known, and the maximum
sample size is $N$; these assumptions are relaxed
in Sections~\ref{sec:t-tests}--\ref{sec:convergence-gen}.
Let $n_k$ denote the number of samples collected
in stage $k\in[K]$, fixed and known prior to data
collection, and let $n_{1:k} = \sum^k_{i=1}n_i$
denote the cumulative sample size through stage
$k$. The total budget is
$\sum_{i=1}^K n_i = n_{1:K} = N$. Let
$\bar X_{1:k} = (1/n_{1:k})\sum^{n_{1:k}}_{i=1}X_i$
denote the sample mean through stage $k$, and let
$S_k$ denote the test statistic at stage $k$
(specified shortly). 

We consider testing the
following null and alternative hypotheses:
\begin{eqnarray}\label{equ:hypotheses}
    H_0: \mu = 0 \quad \mbox{ and } \quad
        H_a: \mu = \mu_a.
\end{eqnarray}
We remark that testing the null at $\mu = 0$ is without loss of
generality, since it can be replaced by any point
null. We focus on the case $\mu_a > 0$; the case
$\mu_a < 0$ is analyzed analogously.

We choose to optimize the expected stopping time,
which is the standard criterion in GST (also used
in \citet{JRSSB_GST}): the goal of GST is to stop
the experiment as early as possible while still
reaching a valid decision. As mentioned above, in
GST one must control the type-1 error across all
$K$ stages simultaneously. Formally, a valid GST
must satisfy
\begin{equation}
\label{eq:type1_control}
\Pr(S_k\geq \theta_k, \,\text{for some } k\in[K]|H_0)\leq \alpha,
\end{equation}
where $\{S_k \geq \theta_k\}$ is the rejection
event at stage $k$, with $\theta_k$ being the
stage-specific rejection cutoff. In words,
the
probability of falsely rejecting the null at any
of the $K$ stages, under $H_0$, is at most
$\alpha$.

As discussed in Section~\ref{section:intro}, there
are infinitely many choices of
$(\theta_k)^K_{k=1}$ that satisfy
\eqref{eq:type1_control}, and the GST literature
has primarily studied different families of
$\theta_k$, equivalently different alpha-spending
functions. Our decision variables are likewise the
rejection cutoffs $(\theta_k)^K_{k=1}$. However,
instead of specifying them through a deterministic
alpha-spending function (examples in
Section~\ref{sec:sims}), we directly optimize a
desired objective subject to
\eqref{eq:type1_control} as a constraint. To the
best of our knowledge, this optimization-based
approach has not been attempted in the GST
literature, except under special cases such as
binary outcomes \citep{JRSSB_GST}.

Many GST procedures additionally come with a
power guarantee \citep{JRSSB_GST, GST_book}. We
therefore add a type-2 error constraint at level
$\beta$ under the alternative $\mu = \mu_a$.
Researchers interested only in type-1 error
control may set $\beta = 1$, which makes the
type-2 constraint vacuous. Let
$\theta = (\theta_k)^K_{k=1}$ and restrict
$\theta$ to a compact box
$C = [\underline{\theta}, \bar\theta]^K$, with
$\underline{\theta}$ chosen sufficiently small
and $\bar\theta$ sufficiently large; this is
without loss of generality. The resulting
optimization problem is given by:
    \begin{equation}\label{equ:main-1}
        \begin{split}
            \min_{\theta\in C} &\,n_1 + \sum^K_{k=2}n_k\Pr(S_i \leq \theta_i, \forall i\in[k-1]|H_a)\\
            \st&\,\Pr(S_k\geq \theta_k, \,\text{for some } k\in[K]|H_0)\leq \alpha,\\
            &\,\Pr(S_k\geq \theta_k, \,\text{for some } k\in[K]|H_a)\geq 1-\beta.
        \end{split}
    \end{equation}
The objective is the expected sample size under
the alternative, the first constraint enforces
type-1 error control as in
\eqref{eq:type1_control}, and the second
constraint requires type-2 error to be at most
$\beta$ (researchers interested only in type-1 error control can ignore the second constraint by setting $\beta = 1$). The formulation \eqref{equ:main-1}
covers one-sample and two-sample tests, in both
one-sided and two-sided forms, by appropriately
specifying $S_k$. For one-sample tests, we set
$S_k = \bar X_{1:k}/(\sigma/\sqrt{n_{1:k}})$
(one-sided) or
$S_k = |\bar X_{1:k}/(\sigma/\sqrt{n_{1:k}})|$
(two-sided). For two-sample tests, let $X_{1,i}$
and $X_{2,i}$ denote i.i.d.\ samples from the
two populations with variances $\sigma^2_1$ and
$\sigma^2_2$. Let $n_{j,k}$ be the number of
samples in stage $k$ from population $j\in\{1,2\}$,
and let $n_{j,1:k} = \sum^k_{i=1}n_{j,i}$. Then
\[
S_k = \frac{\bar X_{1,1:k} - \bar X_{2,1:k}}
{\sqrt{\sigma^2_1/n_{1,1:k} + \sigma^2_2/n_{2,1:k}}}
\]
(for one-sided tests) or
\[
S_k = \frac{|\bar X_{1,1:k} - \bar X_{2,1:k}|}
{\sqrt{\sigma^2_1/n_{1,1:k} + \sigma^2_2/n_{2,1:k}}}
\]
(for two-sided tests).

Problem~\eqref{equ:main-1} is equivalent to:
    \begin{equation}\label{equ:main}
        \begin{split}
            \min_{\theta\in C} &\,n_1 + \sum^K_{k=2}n_k\Pr(S_i < \theta_i, \forall i\in[k-1]|H_a)\\
            \st&\,\Pr(S_k\leq \theta_k, \,\forall k\in[K]|H_0)\geq 1 - \alpha,\\
            &\,\Pr(S_k< \theta_k, \,\forall k\in[K]|H_a)\leq \beta.
        \end{split}
    \end{equation}
We thus focus on problem~\eqref{equ:main} for the remainder of this section\footnote{The objective and second constraint of problem~\eqref{equ:main} use strict inequality $<$ instead of $\leq$, to ensure that the optimal solution for the MILP reformulation of \eqref{equ:main} (i.e. \eqref{equ:known-sigma-milp}), whenever feasible, exists. Since $(S_k)^K_{k=1}$ has a continuous distribution under both $H_0$ and $H_a$, the distinction between $<$ and $\leq$ is immaterial in problems~\eqref{equ:main-1} and \eqref{equ:main}.}.

\subsection{Sample Average Approximation (SAA)}\label{sec:milp-reformulation}

Problem~\eqref{equ:main} is generally
intractable. We address this in two steps. First,
we construct a tractable approximation of
\eqref{equ:main} using the sample average
approximation (SAA) approach. Second, we reformulate the
SAA problem as a mixed-integer linear program
(MILP) that can be solved by standard commercial
solvers \citep{shapiro2021lectures,
pagnoncelli2009sample}. We refer to this two-step
procedure as the ``S-MILP approach.''

% Without loss of generality we let $\sigma^2 = 1$. Then $S_k\sim \set{N}(\sqrt{n_{1:k}}\mu, 1)$. Denote $Z_k \sim N(0,1)$ as a generic standard normal random variable. Then problem~\eqref{equ:main} has the same set of optimal solutions and optimal value as:
%     \begin{subequations}\label{equ:true}
%         \begin{align}
%             \min_{\theta_1,\cdots, \theta_K\geq 0} &\,n_1 + \sum^K_{k=2}n_k\Pr(Z_i \leq \theta_i, \forall i\in[k-1])\\
%             \st&\,\Pr(Z_k\leq \theta_k, \,\forall k\in[K])\geq 1- \alpha,\label{equ:true-c1}\\
%             &\,\Pr(Z_k\leq \theta_k - \sqrt{n_{1:k}}\delta, \,\forall k\in[K])\leq \beta.\label{equ:true-c2}
%         \end{align}
%     \end{subequations}

We replace the probability terms in
\eqref{equ:main} by Monte Carlo averages.
Specifically, we generate $M$ i.i.d.\ sample
paths $\{S^m\}^M_{m=1}$, where each
$S^m = (S^m_1, \dots, S^m_K)$ is drawn from the
distribution of the test statistics
$S = (S_k)^K_{k=1}$ under $H_0$. For instance, in
the one-sided $z$-test with
$S_k = \bar X_{1:k}/(\sigma/\sqrt{n_{1:k}})$,
each $S^m$ is sampled from the multivariate
Gaussian with mean $0$ and covariance matrix
$\Sigma$ with $\Sigma_{k,k} = 1$. Similarly, we
generate $M$ i.i.d.\ sample paths
$\{S^m_a\}^M_{m=1}$ from the distribution of $S$
under $H_a$, where $S^m_a = (S^m_{a,k})_{k\in[K]}$.
For each probability term $\Pr(S\in A\mid H_0)$
and $\Pr(S\in A\mid H_a)$ in \eqref{equ:main}, we
replace it by the SAA surrogate
$(1/M)\sum^M_{m=1}\I(S^m\in A)$ and
$(1/M)\sum^M_{m=1}\I(S^m_a\in A)$, respectively.
This yields the SAA problem:
    \begin{subequations}\label{equ:SAA}
        \begin{align}
            \min_{\theta\in C} &\,n_1 + (1/M)\sum^M_{m=1}\sum^K_{k=2}n_k\I(S^m_{a,i} < \theta_i, \forall i\in[k-1])\label{equ:saa-obj}\\
            \st&\,(1/M)\sum^M_{m=1}\I(S^m_k\leq \theta_k, \,\forall k\in[K])\geq 1- \alpha,\label{equ:saa-c1}
            \\
            &\,(1/M)\sum^M_{m=1}\I(S^m_{a,k}<\theta_k, \,\forall k\in[K])\leq \beta.\label{equ:saa-c2}
        \end{align}
    \end{subequations}
Problem~\eqref{equ:SAA} can then be reformulated as an MILP problem, which can be efficiently solved using available commercial solvers \citep{gurobi, cplex}. The reformulation proceeds by introducing auxiliary binary variables that serve as linear surrogates for the indicator functions appearing in~\eqref{equ:SAA}, namely $\I(S^m_k\leq \theta_k, \forall k\in[K])$ in the type-1 constraint~\eqref{equ:saa-c1} and $\I(S^m_{a,k}<\theta_k, \forall k\in[K])$ together with $\I(S^m_{a,i}<\theta_i, \forall i\in[k-1])$ in the type-2 constraint~\eqref{equ:saa-c2} and objective~\eqref{equ:saa-obj}, respectively. Each such indicator is a discontinuous function of the decision variable $\theta$, which prevents direct solution by continuous optimization methods; replacing each indicator with a binary variable tied to $\theta$ through linear inequalities yields a problem that commercial MILP solvers can handle efficiently.

Concretely, for each sample path $m\in[M]$, we introduce a binary variable $w^m\in\{0,1\}$ that serves as a surrogate for $1-\I(S^m_k\leq \theta_k, \forall k\in[K])$, i.e., the indicator that the null sample path $S^m$ crosses its threshold at some stage. We require $w^m = 1$ if $S^m_k > \theta_k$ for some $k\in[K]$, which will be enforced by constraint~\eqref{equ:known-sigma-milp-c1} in the MILP below. Aggregating $w^m$ across sample paths in constraint~\eqref{equ:known-sigma-milp-c2} then yields a linear expression for the empirical type-1 error, which can be constrained to be at most $\alpha$.

Next, for each $k\in[K]$ and $m\in[M]$, we introduce binary variables $\rho^m_k\in\{0,1\}$ that serve as single-stage surrogates for the events $\{S^m_{a,k}<\theta_k\}$ on the alternative sample path: we require $\rho^m_k = 1$ if $S^m_{a,k} < \theta_k$, enforced by constraint~\eqref{equ:known-sigma-milp-c3}. The variables $\rho^m_k$ encode single-stage events, but the type-2 error constraint and the objective involve $\I(S^m_{a,i}<\theta_i, \forall i\in[k])$, which is the \textit{conjunction} of such events across consecutive stages. To handle this, we introduce a third set of binary variables $\tau^m_k\in\{0,1\}$ that aggregate the $\rho^m_k$ across stages: we require $\tau^m_k = 1$ if $\rho^m_i = 1$ for all $i\in[k]$, enforced by constraint~\eqref{equ:known-sigma-milp-c4} through a standard ``AND'' linearization. Aggregating $\tau^m_k$ across sample paths in constraint~\eqref{equ:known-sigma-milp-c5} and in the objective~\eqref{equ:known-sigma-milp-ob} then yields linear expressions for the empirical type-2 error and the empirical expected sample size, respectively.

Based on the reasoning above, we can reformulate problem~\eqref{equ:SAA} into the following MILP:
    \begin{subequations}\label{equ:known-sigma-milp}
        \begin{align}
            \min&\,n_1 + (1/M)\sum^M_{m=1}\sum^K_{k=2}n_k\tau^m_{k-1}
            \label{equ:known-sigma-milp-ob}\\
            \st
            &\, S_k^m - \theta_k\leq (S_k^m - \underline{\theta})w^m, \forall k\in[K], m\in[M], 
            \label{equ:known-sigma-milp-c1}\\
            &\,(1/M)\sum^M_{m=1}(1-w^m)\geq 1-\alpha,
            \label{equ:known-sigma-milp-c2}\\
            &\,\theta_k-S^m_{a,k}\leq \rho^m_k(\bar\theta- S^m_{a,k}), \,\forall k\in[K], m\in[M],
            \label{equ:known-sigma-milp-c3}\\
            &\, \tau^m_k\geq \sum^k_{i=1}\rho^m_i - k + 1, \forall k\in[K],m\in[M], \label{equ:known-sigma-milp-c4}\\
            &\,(1/M)\sum^M_{m=1}\tau^m_K \leq \beta,
            \label{equ:known-sigma-milp-c5}\\
            &\, w^m\in\{0,1\}, \rho^m_k\in\{0,1\}, \tau^m_k\in\{0,1\}, 
            \forall k\in[K], m\in[M],\\
            &\,\underline{\theta}\leq \theta_k\leq \bar\theta,\,\forall k\in[K]. 
        \end{align}
    \end{subequations}

We now state a proposition that connects the MILP formulation in problem~\eqref{equ:known-sigma-milp} to problem~\eqref{equ:SAA}. 

\begin{proposition}\label{prop:milp-equivalence}
\eqref{equ:known-sigma-milp}, whenever feasible, has the same optimal value as \eqref{equ:SAA}. Formally, if $\hat\theta_k, \hat w^m,  \hat\rho^m_k, \hat\tau^m_k$ for $k\in[K], m\in[M]$ is an optimal solution to \eqref{equ:known-sigma-milp}, then $\{\hat\theta_k\}^K_{k=1}$ is an optimal solution to \eqref{equ:SAA}.
\end{proposition}

% {\color{red}
% According to \eqref{equ:known-sigma-milp-c1}, for each $m\in[M]$, 
% if $w^m = 0$, then we must have $S^m_k\leq \theta_k$ holding for all $k\in[K]$. Otherwise if $S^m_k> \theta_k$ for some $k\in[K]$, we can let $w^m = 1$, in which case  \eqref{equ:known-sigma-milp-c1} implies $\theta_k\geq \underline{\theta}$ and is always feasible. 
% \eqref{equ:known-sigma-milp-c2} ensures the percentage of samples satisfying $S^m_k\leq \theta_k$ for all $k\in[K]$ is no smaller than $1-\alpha$ (see \eqref{equ:saa-c1}). Similarly, according to \eqref{equ:known-sigma-milp-c3}, if $\rho^m_k = 0$, then we must have $S^m_{a,k}\geq \theta_k$; Otherwise we can let $\rho^m_k = 1$, in which case \eqref{equ:known-sigma-milp-c3} implies $\theta_k\leq \bar\theta$ for all $k\in[K]$ and is always feasible. According to \eqref{equ:known-sigma-milp-c4}, if $\tau^m_k = 1$, then we must have $\rho^m_i = 1$ for all $i\in[k]$, implying that $\tau^m_k\geq \I(S^m_{a,i} < \theta_i, \forall i\in[k-1])$. \eqref{equ:known-sigma-milp-c5} ensures the percentage of samples satisfying $S^m_{a,k}\leq \theta_k$ for all $k\in[K]$ is no greater than $\beta$ (see \eqref{equ:saa-c2}). Objective \eqref{equ:known-sigma-milp-ob} minimizes the expected number of samples, calculated as the total sum of percentage of samples satisfying $S^m_{a,i}< \theta_i, \forall i\in[k-1]$ across $k\in[K]$ (see \eqref{equ:saa-obj} for a similar argument). 
% }

\subsection{Finite Sample Convergence Analysis}\label{sec:convergence-main}

\sloppy
We now study the finite-sample convergence of the
SAA problem~\eqref{equ:SAA} to the original
problem~\eqref{equ:main}. We first introduce some
additional notation. Let $v^*(\alpha, \beta)$ and $\theta^*(\alpha, \beta) = (\theta_k^*(\alpha, \beta))^K_{k=1}$ be the optimal value and any optimal solution of the original desired problem~\eqref{equ:main}, respectively. If \eqref{equ:main} is feasible, the existence of $\theta^*(\alpha, \beta)$ is guaranteed since \eqref{equ:main} optimizes a continuous function over a compact space.
Let $\hat v_M(\alpha, \beta)$ and $ \hat\theta_M(\alpha, \beta) = (\hat\theta_{M,k}(\alpha, \beta))^K_{k=1}$ be the optimal value and any optimal solution of the SAA problem~\eqref{equ:SAA}, respectively. 
Let $\Theta(\alpha, \beta) = \{\theta: \Pr(S_k\leq \theta_k, \,\forall k\in[K]|H_0)\geq 1- \alpha, \Pr(S_k< \theta_k, \,\forall k\in[K]|H_a)\leq \beta\}$ be the feasible region of problem \eqref{equ:main}, and let $\Theta_M(\alpha, \beta)$ be the feasible region of the SAA problem \eqref{equ:SAA}. Let $S(\alpha, \beta)$ and $S_M(\alpha, \beta)$ be the set of optimal solutions to the true problem \eqref{equ:main} and the SAA problem \eqref{equ:SAA} respectively. For notational convenience, we also let $f(\theta) = n_1 + \sum^K_{k=2}n_k\Pr(S_i < \theta_i, \forall i\in[k-1]|H_a)$ be the objective, $g_0(\theta) = \Pr(S_k\leq \theta_k, \,\forall k\in[K]|H_0)$ and $g_a(\theta) = \Pr(S_k< \theta_k, \,\forall k\in[K]|H_a)$ be the LHS of the first and second constraint of \eqref{equ:main}, respectively.  

Without loss of generality, we assume that $\alpha$ and $\beta$ are properly chosen so that the SAA problem~\eqref{equ:SAA} is feasible for $M$ sufficiently large. Specifically, we assume that there exists some $\bar\epsilon> 0$ such that  $\Theta(\alpha - \bar\epsilon, \beta -\bar\epsilon)\neq \emptyset$. We begin with the following lemma that shows our feasible set for the SAA problem is non-degenerate:
\begin{lemma}\label{lem:SAA-feasibility}
    Suppose there exists some $\bar\epsilon> 0$ such that  $\Theta(\alpha - \bar\epsilon, \beta -\bar\epsilon)\neq \emptyset$. Then for $M\geq \log{(2/\eta)}/(2\bar\epsilon^2)$, with probability at least $1-\eta$, we have  $\Theta_M(\alpha, \beta)\neq \emptyset$. 
\end{lemma}
The following Proposition~\ref{prop:feasible-set} gives a high probability bound on how well $\Theta_M(\alpha, \beta)$ approximates $\Theta(\alpha, \beta)$, i.e., how well our SAA formulation feasible set approximates the feasible set in the original problem~\eqref{equ:main}.

\begin{proposition}\label{prop:feasible-set}
(i) With probability at least $1-\eta$, we have: 
$$
\Theta\big(\alpha - \epsilon(\eta, M), \beta - \epsilon(\eta, M)\big) \, \subseteq \, \Theta_M(\alpha, \beta) \, \subseteq \, \Theta \big(\alpha + \epsilon(\eta, M), \beta+ \epsilon(\eta, M)\big),
$$ \vspace{-2mm}
with $\epsilon(\eta, M) = \sqrt{\big(c_1(K)\log(1/\eta) + c_2(K)\log(M) + c_3(K)\big)/M}$, where $c_1(K), c_2(K)$ and $c_3(K)$ are some constants that depend on $K$ and not $M$. 

(ii) For all $\epsilon > 0$, we have 
$$
\Pr\big(\Theta(\alpha - \epsilon, \beta - \epsilon)\subseteq \Theta_M(\alpha, \beta)\subseteq \Theta(\alpha + \epsilon, \beta+ \epsilon) \big)\geq 1-\tilde c_1(K)M^{\tilde c_2(K)}\exp{(-\tilde c_3(K)M\epsilon^2)}. 
$$ \vspace{-2mm}
where $\tilde c_1(K), \tilde c_2(K)$ and $\tilde c_3(K)$ are some constants that depend on $K$ and not $M$. 

\end{proposition}
Proposition~\ref{prop:feasible-set} then implies that, $\hat\theta_M(\alpha, \beta)$, any optimal solution to the SAA in problem \eqref{equ:SAA}, is feasible to the original problem \eqref{equ:main} with high probability, as stated in  
Corollary~\ref{prop:prob-feasibility}. 

\begin{corollary}\label{prop:prob-feasibility}
    \sloppy With probability at least $1-\eta$, we have $\hat\theta_M(\alpha, \beta)\in \Theta(\alpha+\epsilon(\eta, M), \beta+\epsilon(\eta, M))$.  
\end{corollary}

It also follows from Proposition~\ref{prop:feasible-set} that we can construct confidence intervals for $v^*(\alpha, \beta)$ using the SAA solution: 
\begin{proposition}\label{prop:confidence-interval}
Let $\eta \in(0,1)$. 
There exists some threshold $\bar M(\bar\epsilon, \eta)> 0$ (see \eqref{proof-equ:ci-M-bar}) such that for all $M\geq \bar M(\bar\epsilon, \eta)$, with probability at least $1-9\eta$, we have $\Theta_M(\alpha - \epsilon(\eta, M), \beta -\epsilon(\eta, M))\neq\emptyset$ and 
\[
v^*(\alpha, \beta)\in\big[f\big(\hat\theta_M(\alpha + \epsilon(\eta, M), \beta+ \epsilon(\eta, M))\big) - \epsilon_f(\eta, M), f\big(\hat\theta_M(\alpha - \epsilon(\eta, M), \beta - \epsilon(\eta, M))\big)\big],
\]
where $\epsilon_f(\eta, M) = (\tilde c_1(K, n) + \tilde c_2(K, n)\log(1/\eta) + \tilde c_3(K, n)\log(M))/\sqrt{M}$ and $\tilde c_1(K, n), \tilde c_2(K, n)$ and $\tilde c_3(K, n)$ are constants that can possibly depend on $K$ and $n$ but not $M$ (see detailed expressions for $\epsilon(\eta, M)$ in Equation~\eqref{equ-appendix:epsilon} and $\epsilon_f(\eta, M)$ in Equation~\eqref{equ-appendix:epsilon-f}). 
\end{proposition}
% Additionally, the SAA optimal solution is also an ``approximately'' optimal solution for the original problem, where the approximately optimality is defined in the following Definition~\ref{def:approx-optimal-solution}.
% \begin{definition}\label{def:approx-optimal-solution}
%     $\hat\theta$ is an $(\upsilon, \epsilon)$-approximately optimal solution if 
%     $$
%     g_0(\hat\theta)\geq 1-\alpha - \upsilon,\,
%     g_a(\hat\theta)\leq \beta + \upsilon, \text{ and } f(\hat\theta) \leq v^*(\alpha, \beta) + \epsilon.
%     $$
% \end{definition}
% \dwh{1-2 sentences that explain the intuition of Definition 2.1?} Let $S(\alpha + \upsilon, \beta+ \upsilon; \epsilon)$ denote the set of all $(\upsilon, \epsilon)$-approximately optimal solutions to the true problem. Corollary~\ref{prop:prob-bound-solution-optimality} gives a high probability upper bound on that the SAA optimal solution being an approximately optimal solution. Corollary~\ref{prop:prob-bound-solution-optimality} follows directly from Propositions~\ref{prop:feasible-set} and \ref{prop:confidence-interval}, and thus we omit its proof. 
% \begin{corollary}\label{prop:prob-bound-solution-optimality}
%     With probability at least $1-\eta$, $\hat\theta_M(\alpha + \epsilon(\eta/2, M), \beta+\epsilon(\eta/2, M)) \in S(\alpha + 2\epsilon(\eta/2, M), \beta+ 2\epsilon(\eta/2, M); \epsilon_f(\eta/2, M))$.  
% \end{corollary}

We now turn to the deviation between the
performance of the SAA optimal solution and the
optimal value of the original
problem~\eqref{equ:main}, i.e.,
$|f(\hat\theta_M(\alpha, \beta)) - v^*(\alpha,
\beta)|$. In Proposition~\ref{prop:objective-bound},
we provide a high-probability upper bound on this
quantity, which can be interpreted as the gap
between the \textit{actual} expected sample size achieved by the
SAA solution and the optimal expected sample size
in~\eqref{equ:main}. The bound requires regularity
on $v^*(\alpha, \beta)$ at $(\alpha, \beta)$, for
which we invoke the Mangasarian-Fromovitz
Constraint Qualification (MFCQ), a standard
condition in parametric programming. To introduce the definition of MFCQ, we consider the following standard nonlinear optimization problem:
\begin{equation}
\label{eq:MFCQ_form}
 \min_{\theta\in\mathbb{R}^n}\{h_0(\theta): h_j(\theta)\leq 0, \,\forall j\in[J]\}.    
\end{equation}
Let $J_0(\theta) = \{j\in[J]: h_j(\theta) = 0\}$ be the set of active constraints at $\theta$. 
\begin{definition}(\citealt[Definition 2.4]{still2018lectures})\label{def:MFCQ}
We say that the Mangasarian-Fromovitz Constraint Qualification (MFCQ) is satisfied for problem~\eqref{eq:MFCQ_form} at $\theta\in\mathbb{R}^n$ if there exists a vector $\xi\in\mathbb{R}^n$ such that $\nabla h_j(\theta)\xi< 0$ for all $j\in J_0(\theta)$. 

\end{definition}
 
With this definition, we are now ready to state our final proposition that characterizes the magnitude of $|f(\hat\theta_M(\alpha, \beta)) - v^*(\alpha, \beta)|$:
\begin{proposition}\label{prop:objective-bound}
    \sloppy 
    Suppose MFCQ holds at any $\theta^*(\alpha, \beta)\in S(\alpha, \beta)$ for problem~\eqref{equ:main}. Let $\eta\in(0,1)$.
    Then, there exists some threshold $\bar M_f(\bar\epsilon, \eta)> 0$ (see $\bar M_f(\bar\epsilon, \eta)$ in \eqref{proof-equ:objective-M-bar}) such that for all $M\geq \bar M_f(\bar\epsilon, \eta)$, with probability at least $1-\eta$, we have $\Theta_M(\alpha, \beta)\neq \emptyset$ and 
    $$
    |f(\hat\theta_M(\alpha, \beta)) - v^*(\alpha, \beta)|\leq \sqrt{\big(c_{\mathrm{v,1}}(K, n) + c_{\mathrm{v},2}(K, n)\log(1/\eta) + c_{\mathrm{v, 3}}(K, n)\log(M)\big)/M}
    $$
    where $c_{\mathrm{v,1}}(K, n), c_{\mathrm{v,2}}(K, n)$ and $c_{\mathrm{v,3}}(K, n)$ are constants that depend on $K$ and $n$ and not $M$.

    % (ii) For any $\epsilon>0$, we have
    % $$
    % \Pr\big(\Theta_M(\alpha, \beta)\neq \emptyset,\text{ and }|f(\hat\theta_M(\alpha, \beta)) - v^*(\alpha, \beta)|\leq \epsilon \big) \geq 1-\tilde c_{\mathrm{v},1}(K,n)M^{\tilde c_{\mathrm{v},2}(K,n)}\exp(-\tilde c_{\mathrm{v},3}(K,n)M\epsilon^2)
    % $$
    % where $\tilde c_{\mathrm{v,1}}(K, n), \tilde c_{\mathrm{v,2}}(K, n)$ and $\tilde c_{\mathrm{v,3}}(K, n)$ are positive constants that depend on $K$ and $n$ and not $M$.
     
\end{proposition}

Proposition~\ref{prop:objective-bound} shows the deviation of the optimal objective value using the SAA approach vanishes at rate $O(1/\sqrt{M})$, up to logarithmic factors. We remark that Proposition~\ref{prop:objective-bound} requires MFCQ to hold at all $\theta^*(\alpha, \beta)\in S(\alpha, \beta)$. We show that this is a very mild assumption. First, it is straightforward to see that if the researcher is interested in only type-1 error control (that is, by setting $\beta = 1$), then this condition is vacuous. Besides, in Lemma~\ref{lem:MFCQ}, we show that the MFCQ holds for all $\theta\in \mathbb{R}^K$, possibly except for a zero Lebesgue measure. 
    
\begin{lemma}\label{lem:MFCQ}
    $\{\theta\in\mathbb{R}^K: \text{MFCQ fails at }\theta \text{ for \eqref{equ:main}}\}$ has a zero Lebesgue measure. 
\end{lemma}
We note that Lemma~\ref{lem:MFCQ} holds for one-sample, two-sample, one-sided, and two-sided z-tests. 
To gain more intuition on the MFCQ condition for Proposition~\ref{prop:objective-bound}, we show in Appendix~\ref{appendix-sec:MFCQ-example} a toy example for which the MFCQ fails to hold at the optimal solution of the true problem, which consequently does not lead to the optimal value of the SAA converging  to the optimal value for the true problem. Lastly, we note that Proposition~\ref{prop:objective-bound} implies that the expected sample size using our SAA solution $\hat\theta_M(\alpha, \beta)$ converges to the optimal expected sample size resulting from solving the original problem \eqref{equ:main}. Nevertheless, the readers may be interested in knowing whether the SAA optimal solution converges to the original optimal solution of \eqref{equ:main}. We can provably show this is true and refer the interested readers to Appendix~\ref{appendix-sec:solution-convergence}.

\section{Sequential Analysis for $t$-tests}\label{sec:t-tests}
In this section, we generalize the previous $z$-test setting. Specifically, we still assume the outcomes are distributed normally, $X_i\sim\mathcal{N}(\mu, \sigma^2)$, but instead do not assume $\sigma^2$ (typically unknown) is known. Under this unknown variance case, we conventionally use the $t$-test test statistic $T_k = \bar X_{1:k}/(\hat \sigma_{1:k}/\sqrt{n_{1:k}})$ for stage $k$, where $\hat \sigma_{1:k} = \sqrt{\sum^{n_{1:k}}_{i=1}(X_i - \bar X_{1:k})^2/(n_{1:k}-1)}$ is the sample standard deviation estimate using samples collected from the first $k$ stages (note $T_k$ can be computed from samples $\{X_i\}^{n_{1:k}}_{i=1}$ directly
without knowledge of the unknown $\sigma$ while the $z$-test statistic $S_k$ required $\sigma$). Then the design problem for the group sequential testing becomes:
\begin{equation}\label{equ:main-unknown-sigma}
        \begin{split}
            \min_{\theta\in C} &\,n_1 + \sum^K_{k=2}n_k\Pr(T_i < \theta_i, \forall i\in[k-1]|H_a)\\
            \st&\,\Pr(T_k\leq \theta_k, \,\forall k\in[K]|H_0)\geq 1 - \alpha,\\
            &\,\Pr(T_k< \theta_k, \,\forall k\in[K]|H_a)\leq \beta.
        \end{split}
    \end{equation}
Let $T(n)$ denote the Student $t$ distribution with $n$ degrees of freedom. Let $T(n, \mu)$ be the noncentral $t$-distribution with $n$ degrees of freedom and noncentrality parameter $\mu$. Then, under $H_0$, we have 
\begin{equation*}
    T_k = \frac{\bar X_{1:k}}{\hat\sigma_{1:k}/\sqrt{n_{1:k}}} = \frac{\sqrt{n_{1:k}}\bar X_{1:k}/\sigma}{\sqrt{\frac{(n_{1:k} - 1)\hat\sigma^2_{1:k}}{\sigma^2}/(n_{1:k} - 1)}} 
    \overset{d}{=} \frac{N(0,1)}{\sqrt{\chi^2_{n_{1:k}-1}/(n_{1:k}-1)}}
    \sim T(n_{1:k} - 1). 
\end{equation*}
Similarly, we also need the distribution of our test statistic under the alternative. Therefore, for any non-degenerate alternative hypothesis ($\mu_a \neq 0$), we have that 
$$
T_k\overset{d}{=}\frac{N(\sqrt{n_{1:k}}\delta,1)}{\sqrt{\chi^2_{n_{1:k}-1}/(n_{1:k}-1)}}
\sim T(n_{1:k} - 1, \sqrt{n_{1:k}}\delta), 
$$
where $\delta := \mu/\sigma$ (assuming $\delta\neq 0$) is the standardized effect size. 

Unfortunately, we are unable to make progress using our SAA approach since the distribution of $T_k$ cannot be sampled under the alternative hypothesis $H_a$ in Equations~\eqref{equ:hypotheses}. As mentioned in Section~\ref{sec:milp-reformulation}, our procedure depends on sample average approximation. One key requirement is to be able to sample our test statistics under the alternative hypothesis (formalized in Section~\ref{sec:convergence-gen}). However, the above equation shows we must sample from a noncentral $t$-distribution with noncentrality parameter $\sqrt{n_{1:k}}\delta$. Unfortunately, under the alternative hypothesis $H_a$ we only know $\mu = \mu_a$ and do not know $\delta = \mu_a/\sigma$. We present two fixes for this. 

\subsection{Standardized Effect Size Hypothesis Testing}
The first is a simple fix, in which we change our null and alternative hypotheses to directly specify the standardized effect size $\delta$. Formally, we consider testing the following hypotheses:
\begin{eqnarray}\label{equ:hypotheses-t}
        H_0': \delta = 0 \quad \mbox{ and } \quad
        H_a': \delta = \delta_a.
\end{eqnarray}
Under this alternative, we can simulate i.i.d. samples from $T = (T_k)^K_{k=1}$ under $H_0'$ and $H_a'$ respectively. Furthermore, this ``workaround'' is standard in the literature as $t$-test based power calculations run into a similar problem \citep{Lachin1981_ttest, cohen1988_ttest}. The interpretation of the hypothesis test, however, is now about the standardized effect size, not the absolute effect size. For later reference, we present the version of \eqref{equ:main-unknown-sigma} with $(H_0, H_a)$ replaced by $(H'_0, H'_a)$:
\begin{equation}\label{equ:main-unknown-sigma-1}
        \begin{split}
            \min_{\theta\in C} &\,n_1 + \sum^K_{k=2}n_k\Pr(T_i < \theta_i, \forall i\in[k-1]|H'_a)\\
            \st&\,\Pr(T_k\leq \theta_k, \,\forall k\in[K]|H'_0)\geq 1 - \alpha,\\
            &\,\Pr(T_k< \theta_k, \,\forall k\in[K]|H'_a)\leq \beta.
        \end{split}
    \end{equation}

With this fix, for each sample path $m\in[M]$, we simulate $X^m_i\sim \mathcal{N}(0,1)$ for each $i\in[n_{1:K}]$ under $H_0'$, and we simulate $X^m_i(\delta_a)\sim \mathcal{N}(\delta_a,1)$ for each $i\in[n_{1:K}]$ under $H_a'$. Let $T^m_k$ and $T^m_k(\delta_a)$ be the $t$-statistics computed using $(X^m_i)^{n_{1:k}}_{i=1}$ and $(X^m_i(\delta_a))^{n_{1:k}}_{i=1}$ respectively. Then $(T^m_k)^K_{k=1}\overset{d}{=} (T_k)^K_{k=1}$ under $H_0'$ and $(T^m_k(\delta_a))^K_{k=1}\overset{d}{=} (T_k)^K_{k=1}$ under $H_a'$, respectively, and the SAA problem for \eqref{equ:main-unknown-sigma-1} is
\begin{subequations}\label{equ:SAA-unknown-sigma}
        \begin{align}
            \min_{\theta\in C} &\,n_1 + (1/M)\sum^M_{m=1}\sum^K_{k=2}n_k\I(T^m_i(\delta_a) < \theta_i, \forall i\in[k-1])\\
            \st&\,(1/M)\sum^M_{m=1}\I(T^m_k\leq \theta_k, \,\forall k\in[K])\geq 1- \alpha,\label{equ:saa-t-c1}
            \\
            &\,(1/M)\sum^M_{m=1}\I(T^m_k(\delta_a)<\theta_k, \,\forall k\in[K])\leq \beta.\label{equ:saa-t-c2}
        \end{align}
    \end{subequations}
Following the same reasoning in problem~\eqref{equ:known-sigma-milp}, we see the MILP reformulation for \eqref{equ:SAA-unknown-sigma} follows directly by replacing $S^m_k$ with $T^m_k$ and $S^m_{a,k}$ with $T^m_k(\delta_a)$, respectively. Thus we omit the exact MILP reformulation here for brevity.

Let $v^*(\alpha, \beta; \delta_a)$ and $\theta^*(\alpha, \beta; \delta_a) = (\theta_k^*(\alpha, \beta; \delta_a))^K_{k=1}$ be the optimal value and any optimal solution of problem~\eqref{equ:main-unknown-sigma-1}, respectively.
Let $\hat v_M(\alpha, \beta; \delta_a)$ and $ \hat\theta_{M}(\alpha, \beta; \delta_a)$ be the optimal value and any optimal solution of the SAA problem~\eqref{equ:SAA-unknown-sigma}, respectively. 
Let $\Theta(\alpha, \beta; \delta_a)$ be the feasible region of true problem \eqref{equ:main-unknown-sigma-1}, and let $\Theta_M(\alpha, \beta; \delta_a)$ be the feasible region of the SAA problem \eqref{equ:SAA-unknown-sigma}. Without loss of generality, we continue to suppose that there exists some $\bar\epsilon> 0$ such that  $\Theta(\alpha - \bar\epsilon, \beta -\bar\epsilon; \delta_a)\neq \emptyset$.
We also let $f(\theta; \delta_a) = n_1 + \sum^K_{k=2}n_k\Pr(T_i < \theta_i, \forall i\in[k-1]|H'_a)$ be the objective. We are now ready to state the main Proposition~\ref{prop:convergence-t} that shows the SAA feasible region approximates the true feasible region and that the objective value of the SAA solution converges to that of the original $t$-test problem.

\begin{proposition}\label{prop:convergence-t}
    
(i) With probability at least $1-\eta$, we have 
\[
\Theta\big(\alpha - \epsilon_\mt(\eta, M), \beta - \epsilon_\mt(\eta, M); \delta_a \big) \subseteq \Theta_M(\alpha, \beta; \delta_a) \subseteq \Theta\big(\alpha + \epsilon_\mt(\eta, M), \beta+ \epsilon_\mt(\eta, M); \delta_a \big),
\]
with $\epsilon_\mt(\eta, M) = \sqrt{\big(c^\mt_1(K)\log(1/\eta) + c^\mt_2(K)\log(M) + c^\mt_3(K)\big)/M}$, where $c^\mt_1(K), c^\mt_2(K)$ and $c^\mt_3(K)$ are some constants that depend on $K$ and not $M$. 

% Alternatively, for all $\epsilon > 0$, we have $
% \Pr(\Theta(\alpha - \epsilon, \beta - \epsilon; \delta_a)\subseteq \Theta_M(\alpha, \beta; \delta_a)\subseteq \Theta(\alpha + \epsilon, \beta+ \epsilon; \delta_a))\geq 1-\tilde c^\mt_1(K)M^{\tilde c^\mt_2(K)}\exp{(-\tilde c^\mt_3(K)M\epsilon^2)} 
% $,
% where $\tilde c^\mt_1(K), \tilde c^\mt_2(K)$ and $\tilde c^\mt_3(K)$ are some constants that depend on $K$ and not $M$. z

(ii) Suppose MFCQ holds at any optimal solution of problem~\eqref{equ:main-unknown-sigma-1}. Then there exists $\bar M_\mt > 0$ such that for all $M\geq \bar M_\mt$, 
with probability at least $1-\eta$, we have $\Theta_M(\alpha, \beta; \delta_a)\neq \emptyset$ and 
\[
|f(\hat\theta_M(\alpha, \beta; \delta_a); \delta_a) - v^*(\alpha, \beta; \delta_a)|\leq \sqrt{\big(c^\mt_{\mv,1}(K) + c^\mt_{\mv,2}(K)\log(1/\eta) + c^\mt_{\mv,3}(K)\log(M)\big)/M}, 
\]
where $c^\mt_{\mathrm{v,1}}(K), c^\mt_{\mathrm{v,2}}(K)$ and $c^\mt_{\mathrm{v,3}}(K)$ are positive constants that depend on $K$ and not $M$. 

% Alternatively, for all $\epsilon > 0$, we have $
% \Pr(|f(\hat\theta_M(\alpha, \beta; \delta_a); \delta_a) - v^*(\alpha, \beta; \delta_a)|\leq\epsilon)\geq 1-\tilde c^\mt_{\mv,1}(K)M^{\tilde c^\mt_{\mv,2}(K)}\exp{(-\tilde c^\mt_{\mv,3}(K)M\epsilon^2)} 
% $,
% where $\tilde c^\mt_{\mv,1}(K), \tilde c^\mt_{\mv,2}(K)$ and $\tilde c^\mt_{\mv,3}(K)$ are some constants that depend on $K$ and not $M$. 

\end{proposition}
    
We leave the proof of Proposition~\ref{prop:convergence-t} to Appendix~\ref{appendix:t-test-effect size}. 
We remind the readers that our first fix for the unknown $\sigma$ case is to simply specify $H_0',H_a'$ in terms of the standardized effect size $\delta$. In the 
scenario where the statistician/practitioner is interested in estimating and constructing alternatives based on the mean difference $\mu_a$ instead of the standardized effect size $\delta_a$, the distribution of the $t$-statistics under $H_a$
cannot be fully determined. However, with this simple fix we recover all the main results from Section~\ref{sec:convergence-main}, i.e., our optimal objective value using the S-MILP approach converges to the optimal value of the original problem~\eqref{equ:main-unknown-sigma-1}. In the next subsection, we study an alternative approach: using the pilot study to collect samples and estimate the unknown $\sigma$.

\subsection{Pilot Study Approach}\label{sec:pilot study}
In this section, we return to solving the original null and alternative hypothesis presented in Equations~\eqref{equ:hypotheses} when $\sigma$ is unknown.

In this case, we do not use a ``simple'' fix like the aforementioned section, i.e., we can no longer sample the $t$-test statistic under the alternative (since $\delta$ is unknown). Formally, in this case the researcher still wants to test $\mu$ directly using $H_0$ and $H_a$ (as opposed to the standardized effect size). To make progress, we require additional ``pilot study'' samples to estimate an initial $\hat\sigma_0$ to use as a ``plug-in''. Formally, before the group sequential test begins, we conduct a pilot study to collect $n_0$ i.i.d. data $X_{-1},\cdots, X_{-n_0}$ (i.i.d. with respect to $X_1, \cdots, X_{n_{1:K}}$ as well) and estimate $\sigma$ using the typical sample standard deviation: 
$$
\hat\sigma_0 = \sqrt{\frac{\sum^{n_0}_{i=1}(X_{-i} - \bar X_{-1:n_0})^2}{n_0 - 1}}, 
$$
where $\bar X_{-1:n_0} = \sum^{n_0}_{i=1}X_{-i}/n_0$ is the sample mean for the pilot data. 
Then we can use the estimated $\hat\sigma_0$ as the surrogate for the true $\sigma$ to calculate the type-1 and type-2 errors under the sequential $t$-tests. We thus formulate the design problem for sequential $t$-tests as:
\begin{subequations}\label{equ:main-unknown-sigma-pilot}
        \begin{align}
            \min_{\theta\in C} &\,n_1 + \sum^K_{k=2}n_k\Pr(T_i < \theta_i, \forall i\in[k-1]|H_a = \mu_a, \sigma = \hat\sigma_0)\label{equ:equ:main-unknown-sigma-pilot-obj}\\
            \st&\,\Pr(T_k\leq \theta_k, \,\forall k\in[K]|H_0)\geq 1 - \alpha,\label{equ:equ:main-unknown-sigma-pilot-c1}\\
            &\,\Pr(T_k< \theta_k, \,\forall k\in[K]|H_a, \sigma = \hat\sigma_0)\leq \beta,\label{equ:equ:main-unknown-sigma-pilot-c2}
        \end{align}
    \end{subequations}
where we recall $T_k = \bar X_{1:k}/(\hat\sigma_{1:k}/\sqrt{n_{1:k}})$ with $\hat \sigma_{1:k} = \sqrt{\sum^{n_{1:k}}_{i=1}(X_i - \bar X_{1:k})^2/(n_{1:k}-1)}$. In Equation~\eqref{equ:equ:main-unknown-sigma-pilot-c1}, we calculate the type-1 error 
conditional on that 
$X_i\sim  \mathcal{N}(0, 1)$ under $H_0$; In Equation~\eqref{equ:equ:main-unknown-sigma-pilot-c2}, we calculate the type-2 error 
conditional on that $X_i\sim \mathcal{N}(\mu_a, \hat\sigma^2_0)$ under $H_a$, where we have used $\hat\sigma_0$ as a ``plug-in'' for $\sigma$. 

Then we can again simulate $M$ sample paths as follows: for each $m\in[M]$, we simulate i.i.d. $X^m_i\sim \mathcal{N}(0,1)$ under $H_0$, and $X^m_i(\mu_a/\hat\sigma_0)\sim \mathcal{N}(\mu_a,\hat\sigma^2_0)$ under $H_a$ for each $i\in[n_{1:K}]$. Let $T^m_k$ and $T^m_k(\mu_a/\hat\sigma_0)$ be the $t$-statistics computed using $(X^m_i)^{n_{1:k}}_{i=1}$ and $(X^m_i(\mu_a/\hat\sigma_0))^{n_{1:k}}_{i=1}$ respectively. Then we have that $(T^m_k)^K_{k=1}\overset{d}{=} (T_k)^K_{k=1}$ under $H_0$ and $(T^m_k(\mu_a/\hat\sigma_0))^K_{k=1}\overset{d}{=} (T_k)^K_{k=1}$ under $H_a$ and $\sigma = \hat\sigma_0$, respectively, so the SAA analog for \eqref{equ:main-unknown-sigma-pilot} is
\begin{subequations}\label{equ:SAA-unknown-sigma-pilot}
        \begin{align}
            \min_{\theta\in C} &\,n_1 + (1/M)\sum^M_{m=1}\sum^K_{k=2}n_k\I(T^m_i(\mu_a/\hat\sigma_0) < \theta_i, \forall i\in[k-1])\\
            \st&\,(1/M)\sum^M_{m=1}\I(T^m_k\leq \theta_k, \,\forall k\in[K])\geq 1- \alpha,\label{equ:saa-t-pilot-c1}
            \\
            &\,(1/M)\sum^M_{m=1}\I(T^m_k(\mu_a/\hat\sigma_0)<\theta_k, \,\forall k\in[K])\leq \beta.\label{equ:saa-t-pilot-c2}
        \end{align}
    \end{subequations}
We can easily get the MILP reformulation for the above problem by replacing $S^m_k$ with $T^m_k$ and $S^m_{a,k}$ with $T^m_k(\mu_a/\hat\sigma_0)$ in \eqref{equ:known-sigma-milp}, respectively. 

Keen readers may have noticed that the optimization problem in \eqref{equ:main-unknown-sigma-pilot} (specifically \eqref{equ:equ:main-unknown-sigma-pilot-obj} and \eqref{equ:equ:main-unknown-sigma-pilot-c2}) conditions on $\sigma = \hat\sigma_0$, which is an approximation to the original problem~\eqref{equ:main-unknown-sigma} (where we do not have this condition). Formally, since $\hat\sigma_0\neq \sigma$ with probability one, the optimal design derived from solving \eqref{equ:main-unknown-sigma-pilot} differs from the ``true'' optimal design if we had been able to know the standardized effect size (see problem~\eqref{equ:main-unknown-sigma}). 
% \begin{equation}\label{equ:main-unknown-sigma-nominal}
%         \begin{split}
%             \min_{\theta\in C} &\,n_1 + \sum^K_{k=2}n_k\Pr(T_i \leq \theta_i, \forall i\in[k-1]|H_a = \mu_a, \sigma = \sigma_0)\\
%             \st&\,\Pr(T_k\leq \theta_k, \,\forall k\in[K]|H_0)\geq 1 - \alpha,\\
%             &\,\Pr(T_k\leq \theta_k, \,\forall k\in[K]|H_a = \mu_a, \sigma = \sigma_0)\leq \beta,
%         \end{split}
%     \end{equation}
Ex ante, the optimal value and optimal solutions of \eqref{equ:main-unknown-sigma-pilot} are random variables which depend on the samples $\{X_{-i}\}^{n_0}_{i=1}$ collected from the pilot study. Therefore, the decision quality of \eqref{equ:main-unknown-sigma-pilot} and also the SAA analog \eqref{equ:SAA-unknown-sigma-pilot} depends on the number of pilot samples as well. However, we show in the following Propositions that this plug-in approach has desirable theoretical properties to approximate the original problem.

Before stating these results, we define a few more notations. Let $\hat\delta_{n_0} = \mu_a/\hat\sigma_0$. Then $\Theta_M(\alpha, \beta; \hat\delta_{n_0})$ is the set of feasible solutions to problem~\eqref{equ:SAA-unknown-sigma-pilot} and $\hat\theta_M(\alpha, \beta; \hat\delta_{n_0})\in \Theta_M(\alpha, \beta; \hat\delta_{n_0})$ is any optimal solution of problem~\eqref{equ:SAA-unknown-sigma-pilot}. With some abuse of notation, we let $\Theta(\alpha, \beta; \delta)$ be the feasible region, and $f(\theta; \delta)$ be the objective, and $v^*(\alpha, \beta; \delta)$ be the optimal value of true problem \eqref{equ:main-unknown-sigma}, respectively.

\begin{proposition}\label{prop:pilot-convergence}
    
(i) Let $\eta\in(0,1)$. There exists $\tilde n_\p$ such that for all $n_0\geq \tilde n_\p$, 
with probability at least $1-\eta$, we have
\[
\Theta\big(\alpha - \epsilon_\p(\eta, M, n_0), \beta - \epsilon_\p(\eta, M, n_0); \delta\big) \subseteq \Theta_M(\alpha, \beta; \hat\delta_{n_0}) \subseteq \Theta\big(\alpha + \epsilon_\p(\eta, M, n_0), \beta+ \epsilon_\p(\eta, M, n_0); \delta\big),
\]
with $\epsilon_\p(\eta, M, n_0) = \sqrt{\big(c^\p_1(K)\log(1/\eta) + c^\p_2(K)\log(Mn_0) + c^\p_3(K)\big)\cdot(1/n_0 + 1/M)}$, where $c^\p_1(K), c^\p_2(K)$ and $c^\p_3(K)$ are some constants that depend on $K$ and not $M$ or $n_0$. 

% Alternatively, for all $\epsilon > 0$, we have $
% \Pr(\Theta(\alpha - \epsilon, \beta - \epsilon; \delta_a)\subseteq \Theta_M(\alpha, \beta; \hat\delta_{n_0})\subseteq \Theta(\alpha + \epsilon, \beta+ \epsilon; \delta_a))\geq 1-\tilde c^\p_1(K)(Mn_0)^{\tilde c^\p_2(K)}\exp{(-\tilde c^\p_3(K)Mn_0\epsilon^2/(M+n_0))} 
% $,
% where $\tilde c^\p_1(K), \tilde c^\p_2(K)$ and $\tilde c^\p_3(K)$ are some constants that depend on $K$ and not $M$ nor $n_0$. 

(ii) Suppose MFCQ holds at any optimal solution of problem~\eqref{equ:main-unknown-sigma}. There exist $\bar M_\p$ and $\bar n_\p$ such that for all $M\geq \bar M_\p$ and $n_0\geq \bar n_\p$, 
with probability at least $1-\eta$, we have $\Theta_M(\alpha, \beta; \hat\delta_{n_0})\neq\emptyset$ and
\begin{align*}
&|f(\hat\theta_M(\alpha, \beta; \hat\delta_{n_0}); \delta) - v^*(\alpha, \beta; \delta)|\\
\leq &\sqrt{\big(c^\p_{\mv,1}(K) + c^\p_{\mv,2}(K)\log(1/\eta) + c^\p_{\mv,3}(K)\log(Mn_0)\big)\cdot(1/n_0 + 1/M)}, 
\end{align*}
where $c^\p_{\mathrm{v,1}}(K), c^\p_{\mathrm{v,2}}(K)$ and $c^\p_{\mathrm{v,3}}(K)$ are positive constants that depend on $K$ and not $M$ or $n_0$. 

% Alternatively, for all $\epsilon > 0$, we have $
% \Pr(|f(\hat\theta_M(\alpha, \beta; \hat\delta_{n_0}); \delta_a) - v^*(\alpha, \beta; \delta_a)|\leq \epsilon)\geq 1-\tilde c^\p_{1, \mv}(K)(Mn_0)^{\tilde c^\p_{\mv,2}(K)}\exp{(-\tilde c^\p_{\mv,3}(K)Mn_0\epsilon^2/(M+n_0))} 
% $,
% where $\tilde c^\p_{\mv,1}(K), \tilde c^\p_{\mv,2}(K)$ and $\tilde c^\p_{\mv,3}(K)$ are some constants that depend on $K$ and not $M$ or $n_0$. 

\end{proposition}

We leave the detailed statements for Proposition~\ref{prop:pilot-convergence} and all the corresponding proofs to Appendix~\ref{sec-appendix:t-test-pilot}. To summarize, we have similar convergence guarantees as above. Specifically, Proposition~\ref{prop:pilot-convergence} implies that the expected sample size corresponding to the SAA optimal solution of \eqref{equ:SAA-unknown-sigma-pilot} converges to the counterpart of the optimal design \eqref{equ:main-unknown-sigma} for sequential $t$-tests as $M\to \infty$ and additionally as our pilot sample size  $n_0\to\infty$.

\section{Convergence Analysis Generalization}
% \label{sec:generalization_extension}
% \subsection{Convergence Analysis Generalization}
\label{sec:convergence-gen}
In this section, we generalize our results and provide sufficient conditions for similar results to hold. We still consider the optimization framework listed in problem~\eqref{equ:main}. However, we aim to relax any normality assumption of our outcomes and further talk about general test statistics (with or without knowledge of $\sigma$). Specifically, we use any general continuous test statistics $Q = (Q_k)^K_{k=1}$ (which are not necessarily normally distributed). Let $G_0$ and $G_a$ denote the distributions of any general continuous test statistic $Q$ under $H_0$ and $H_a$, respectively. With a slight abuse of notation we still assume each outcome is distributed i.i.d. with finite first two moments ($|\mu|, \sigma^2 < \infty$), and thus $H_0, H_a$ refer to the first moment of these i.i.d outcomes. Thus in the remainder of this subsection, we refer to problem~\eqref{equ:main} as an analog with summary statistics $S$ replaced by $Q$. 

As shown in the above sections, we require generating i.i.d. samples from $G_0$ and $G_a$ for our SAA approach to approximate the type-1 and type-2 error constraints. We further require that $G_0$ and $G_a$ have density functions that are upper bounded almost everywhere on $C$. Formally,
\begin{assumption}\label{ass:convergence-generalization}
   (i) The densities of $G_0$ and $G_a$ are upper bounded almost everywhere on $C$. (ii) It is possible to generate i.i.d. samples from $G_0$ and $G_a$ under the null and alternative hypothesis, respectively. (iii) MFCQ holds at any optimal solution of problem~\eqref{equ:main} with $Q\sim G_0$ under $H_0$ and $Q\sim G_a$ under $H_a$. 
\end{assumption}
Under Assumption~\ref{ass:convergence-generalization}, we can generalize the results in Section ~\ref{sec:convergence-main} to other tests. With $Q\sim G_0$ under $H_0$ and $Q\sim G_a$ under $H_a$, 
let $\Theta_\g(\alpha, \beta)$ be the feasible set of problem~\eqref{equ:main} and $\Theta^\g_M(\alpha, \beta)$ be the feasible set of problem~\eqref{equ:SAA}, respectively; let $f_\g(\theta) = n_1 + \sum^K_{k=2}n_k\Pr(Q_i < \theta_i, \forall i\in[k-1]|H_a)$ and $v^*_\g(\alpha, \beta)$ be the optimal value of \eqref{equ:main}. Finally, let $\hat\theta^\g_M(\alpha, \beta)$ be any optimal solution to SAA problem~\eqref{equ:SAA}.

\begin{proposition}\label{prop:convergence-generalization}

Suppose Assumptions~\ref{ass:convergence-generalization}(i)-(ii) hold. 
(i) With probability at least $1-\eta$, we have 
\[
\Theta_\g(\alpha - \epsilon_\g(\eta, M), \beta - \epsilon_\g(\eta, M))\subseteq \Theta^\g_M(\alpha, \beta)\subseteq \Theta_\g(\alpha + \epsilon_\g(\eta, M), \beta+ \epsilon_\g(\eta, M)),
\]
with $\epsilon_\g(\eta, M) = \sqrt{\big(c^\g_1(K)\log(1/\eta) + c^\g_2(K)\log(M) + c^\g_3(K)\big)/M}$, where $c^\g_1(K), c^\g_2(K)$ and $c^\g_3(K)$ are some constants that depend on $K$ and not $M$. 

% Alternatively, for all $\epsilon > 0$, we have $
% \Pr(\Theta_\g(\alpha - \epsilon, \beta - \epsilon)\subseteq \Theta^\g_M(\alpha, \beta)\subseteq \Theta_\g(\alpha + \epsilon, \beta+ \epsilon))\geq 1-\tilde c^\g_1(K)M^{\tilde c^\g_2(K)}\exp{(-\tilde c^\g_3(K)M\epsilon^2)} 
% $,
% where $\tilde c^\g_1(K), \tilde c^\g_2(K)$ and $\tilde c^\g_3(K)$ are some constants that depend on $K$ and not $M$. 

(ii) Suppose Assumption~\ref{ass:convergence-generalization}(iii) holds as well. There exists some threshold $\bar M_\g> 0$ such that for all $M\geq \bar M_\g$, with probability at least $1-\eta$, we have $\Theta^\g_M(\alpha, \beta)\neq \emptyset$ and
\[
|f_\g(\hat\theta^\g_M(\alpha, \beta)) - v^*_\g(\alpha, \beta)|\leq \sqrt{\big(c^\g_{\mv,1}(K) + c^\g_{\mv,2}(K)\log(1/\eta) + c^\g_{\mv,3}(K)\log(M)\big)/M}, 
\]
where $c^\g_{\mathrm{v,1}}(K), c^\g_{\mathrm{v,2}}(K)$ and $c^\g_{\mathrm{v,3}}(K)$ are positive constants that depend on $K$ and not $M$. 

% Alternatively, for all $M\geq \bar M$ and any $\epsilon>0$,
% $
% \Pr(|f_\g(\hat\theta^\g_M(\alpha, \beta)) - v_\g^*(\alpha, \beta)|\leq \epsilon) \geq 1-\tilde c^\g_{\mathrm{v},1}(K,n)M^{\tilde c^\g_{\mathrm{v},2}(K,n)}\exp(-\tilde c_{\mathrm{v},3}(K,n)M\epsilon^2)
% $
% where $\tilde c^\g_{\mathrm{v,1}}(K, n), \tilde c^\g_{\mathrm{v,2}}(K, n)$ and $\tilde c^\g_{\mathrm{v,3}}(K, n)$ are positive constants that depend on $K$ and $n$ and not $M$.

\end{proposition}
While Proposition~\ref{prop:convergence-generalization} is a strong guarantee under general conditions, we remark that Assumption~\ref{ass:convergence-generalization} (ii) is not trivial. Knowing the distributions of $G_0$ and $G_a$ often requires knowing the distribution of the outcome. Furthermore, just knowing the outcome distribution may still not be enough to sample from the test statistic distribution under the null and alternative hypothesis (as shown in Section~\ref{sec:t-tests}). However, in large samples, $G_0$ and $G_a$ typically converge to a normal distribution. In this case, all of our aforementioned results will hold asymptotically (with respect to $n$), but we omit this technicality. 

Lastly, we have one additional generalization, where we allow analysts to stop early not only when rejecting $H_0$ but also when accepting $H_0$. In this case, the optimization problem requires another vector of cutoff values $\theta_k$ to accept the null hypothesis. This new setting is often referred to as symmetric designs and allows early stopping in each of the $K$ group tests in support of $H_0$ \citep{GST_book}. Although this is a more complicated setting, we use the same SAA and MILP approach and have the same convergence guarantees as those listed in Sections~\ref{sec:z_test}-~\ref{sec:t-tests}. Therefore, for brevity we move the full formalization and results of this extension to Appendix~\ref{sec:futility}.

\section{Simulations}
\label{sec:sims}
\subsection{Optimal Sample Size Reduction}
\label{subsec:sim1}
To showcase our optimal S-MILP approach to group sequential hypothesis testing, we compare our method with three popularly used group sequential hypothesis tests: 1) Lan-DeMets, 2) Pocock, and 3) O'Brien-Fleming tests \citep{Lan_DeMets_GST, pocock_GST, obrien_fleming_GST}. All three methods uniformly control type-1 error at level $\alpha$ but each with different alpha-spending functions. As mentioned in Section~\ref{section:intro}, much of the group sequential hypothesis testing literature focuses on specifying a deterministic alpha-spending function that determines the $\theta$ cutoff. Formally, an alpha-spending function $\alpha(k)$ specifies at each time $k = 1, 2, \dots, K$ the cumulative type-1 error under the null hypothesis. For example, suppose $\alpha = 0.05$ and $K = 5$. A common Lan-DeMets alpha-spending function is to equally space out $\alpha$ into $K$ intervals, i.e., $\alpha(k= 1)_{\text{Lan-DeMets}} = 0.05/5 = 0.01$, i.e., by the first check the analyst has a cumulative type-1 error of 0.01, $ \alpha(k= 2)_{\text{Lan-DeMets}}  = 0.02, \dots, \alpha(k= 5)_{\text{Lan-DeMets}} = 0.05$, where by the final check $k = K$ one has spent all the ``alpha-budget'' $\alpha(k=K) = 0.05$. The three alpha-spending functions we compare are the following: 
\begin{equation}
\label{eq:alpha_spending}
\begin{aligned}
   \text{Lan-DeMets: } & \alpha(k)_{\text{Lan-DeMets}} = \alpha k/K, \\
   \text{Pocock: } & \alpha(k)_{\text{Pocock}} = \alpha\ln(1 + (e-1)k/K),\\
      \text{O'Brien-Fleming: } & \alpha(k)_{\text{O'Brien-Fleming}} =1 - \Phi(z_{1-\alpha}/\sqrt{k/K}). \\
\end{aligned}    
\end{equation}

For simplicity, we showcase our method in the base case under the $z$-test scenario analyzed in Section~\ref{sec:z_test}. To facilitate a fair comparison, we use the same $z$-test based test statistic $S_k$ for all methods and numerically compute $\theta_k$ values for each GST test given either by the alpha-spending function above or by our optimal S-MILP approach. We test the (one-sided) hypotheses in Equation~\eqref{equ:hypotheses}, where the true signal is $\mu_a = 0.5$ and $\sigma = 1$. We set the final sample $N = 100$ with equally divided buckets, e.g., for $K = 4$, $n_{1:1} = 25$ and $n_{1:2} = 50$. Given the $\theta$ cutoffs, for each method we calculate the expected sample size under the respective cutoffs. 
% Formally, we record the sample size when we first reject the null for each method. If no rejection is made by the final check $k = K$, then the final sample size $N$ is recorded. 
% The Monte-Carlo parameter is set at 1000 while the SAA parameter $M = 3000$. 
We set $\alpha = 0.05$ and ignore the type-2 error constraint (as competitive GST methods only control type-1 error) by setting a degenerate $\beta = 0.95$. For SAA, we set the number of sample paths to be $M = 2000$, and use $\alpha = 0.045$ as a conservative buffer against SAA error (see Corollary~\ref{prop:prob-feasibility}) and then validate the resulting boundaries for the true type-1 error. We also show the results for $K = 3, 5, 7$. 

% To further ensure the compliance to the error constraints, we replicate the above procedure for $100$ times, and then get the optimal $\hat\theta_M(0.045, 0.95)$ that meets the error constraints and gives the smallest expected sample size. \xz{I add the full details for the experiment setting here. Let me know if you think some details are not necessarily needed David. }

\begin{table}[!t]
\begin{center}
\begin{adjustbox}{max width=\textwidth}
\begin{tabular}{l|ccc|} 
& \multicolumn{3}{c|}{}  \\
Method & $K = 3$ & $K = 5$ & $K = 7$ \\ 
\hline
S-MILP: True ESS & 38.18 & 30.20 & 28.47 \\
S-MILP: Cumulative Type-1 Error & 0.05 & 0.05 & 0.05 \\
\hline
Lan-DeMets ESS & 41.36 & 34.75 & 32.46 \\
S-MILP Improvement Relative to Lan-DeMets & 7.69\% & 13.09\% & 12.29\% \\
\hline
Pocock ESS & 40.14 & 33.12 & 30.80 \\
S-MILP Improvement Relative to Pocock & 4.89\% & 8.80\% & 7.57\% \\
\hline
O'Brien-Fleming ESS & 50.16 & 45.22 & 42.19 \\
S-MILP Improvement Relative to O'Brien-Fleming & 23.89\% & 33.21\% & 32.52\% \\
\end{tabular}
\end{adjustbox}
\caption{Comparison of expected sample size (ESS) between the S-MILP model and other popular GST methods outlined in Equation~\eqref{eq:alpha_spending}. We set the desired type-1 error at $\alpha = 0.05$ and use $M = 2000$ for our S-MILP approach. Each $n_i$ is spaced out equally with final sample size $N = 100$. We set the signal $\mu_a = 0.5$ and use the $z$-test based statistic (Section~\ref{sec:z_test}) for all tests. Relative improvement is computed as $(\text{Benchmark ESS} - \text{S-MILP ESS}) / \text{Benchmark ESS}$.}
\label{tab:ess_comparison}
\end{center}
\end{table}

Table~\ref{tab:ess_comparison} shows the expected sample size of our approach (labeled as ``S-MILP'') compared to the three commonly used GST methods outlined in Equation~\eqref{eq:alpha_spending}. First, the second row shows that our S-MILP approach achieves the desired actual type-1 error of $\alpha = 0.05$ (which is evaluated as $\Pr(S_k\geq \hat\theta_M(\alpha, \beta), \text{ for some }k\in[K]|H_0)$ through analytic multivariate normal integration), allowing a fair comparison with the other methods. Our S-MILP approach also, as guaranteed by our optimization procedure, achieves the lowest expected sample size to rejection, i.e., our S-MILP is optimal, thus more efficient than other methods, in rejecting the null hypothesis as early as possible. For example, the last row of Table~\ref{tab:ess_comparison} shows our S-MILP approach is 33\% more efficient than the commonly used O'Brien-Fleming GST method when $K = 5$ or $K = 7$. The improvements relative to Lan-DeMets and Pocock are more modest but still enjoy an approximately 5-10\% reduction in expected sample size. 

\subsection{Alpha-Spending Budget Comparison}
Although Section~\ref{subsec:sim1} shows our S-MILP optimization approach to GST more efficiently rejects the null hypothesis or equivalently saves samples, it still leaves readers wondering how the optimal alpha-spending function obtained by our S-MILP approach compares to the popularly used ones outlined in Equation~\eqref{eq:alpha_spending}. As mentioned throughout this paper, the GST literature has focused on specifying different alpha-spending functions, each spending $\alpha$ differently over the $K$ intervals. We show how the optimization-based approach answers this long-standing debate.

\begin{figure}[t]
\begin{center}
\includegraphics[width=14cm]{"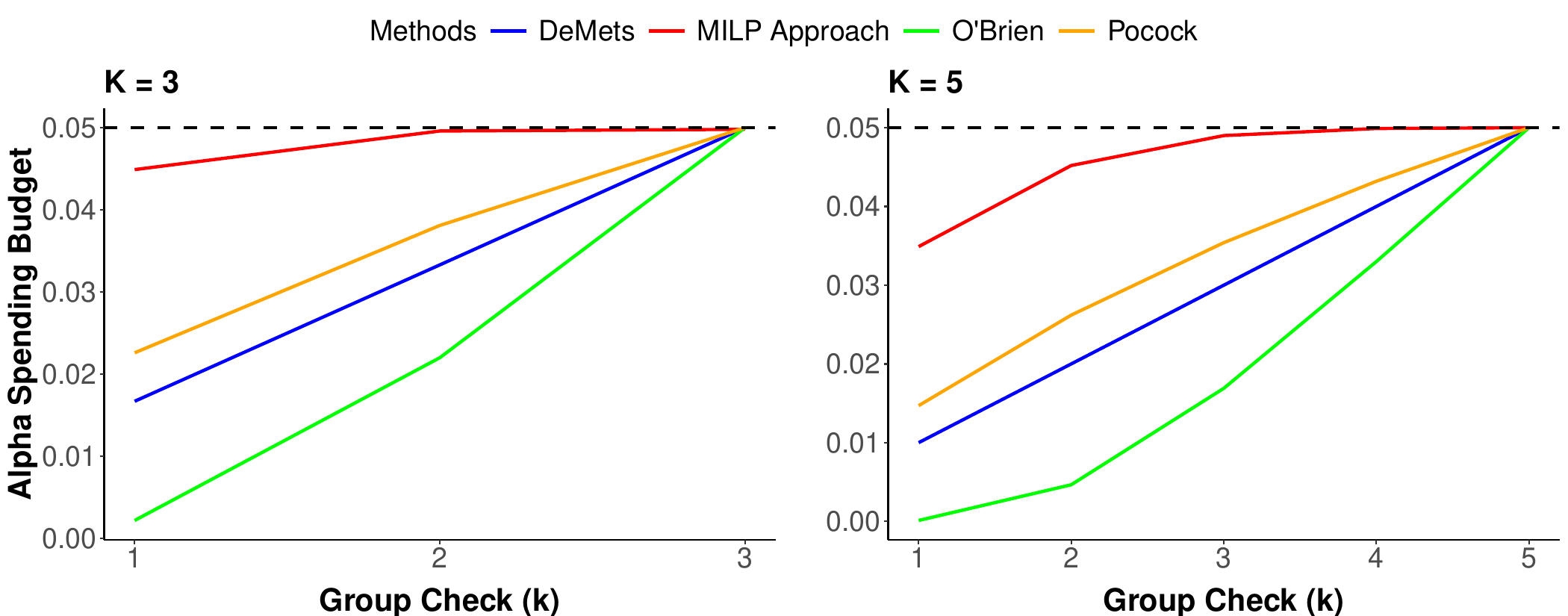"}
\caption{Plot of alpha-spending budgets for all methods. The simulation setting is the same as that in Table~\ref{tab:ess_comparison}. The red line denotes the alpha-spending budget of our S-MILP approach computed by the optimal $\theta$ cutoff values obtained from solving problem~\eqref{equ:known-sigma-milp}. The blue, orange, and green lines are deterministically determined by Equations~\eqref{eq:alpha_spending}, respectively.}
\label{fig:alpha_spending}
\end{center}
\end{figure}

% \begin{table}[!t]
% \begin{center}
% \begin{adjustbox}{max width=\textwidth}
% \begin{tabular}{l|ll|} 
% & \multicolumn{2}{c|}{Alpha Spending Budget ($\alpha = 0.05$)}  \\
% Different Methods & $K = 3$ & $K = 5$  \\ 
%  \hline
% S-MILP Optimization Approach & (0.0449, 0.0496, 0.0498) & (0.0349, 0.0452, 0.0490, 0.0499,	0.05) \\
% Lan-DeMets  & (0.0167, 0.0333, 0.0500) & (0.01, 0.02, 0.03, 0.04, 0.05)\\
% Pocock & (0.0226, 0.0381, 0.05) & (0.0147, 0.0262, 0.0354, 0.0432, 0.05)\\
% O'Brien-Fleming & (0.00219,	0.0220,	0.05) & (0.000118, 0.00465, 0.0169, 0.0330,	0.05)
% \end{tabular}
% \end{adjustbox}
% \caption{Same simulation setting as \xz{Do we put $K = 7$ here as well?}}
% \label{tab:sim2}
% \end{center}
% \end{table}

Figure~\ref{fig:alpha_spending} shows the alpha-spending budget of our S-MILP Optimization Approach compared to the three (deterministic) alpha-spending functions outlined in Equation~\eqref{eq:alpha_spending} under the same simulation setting as that in Section~\ref{subsec:sim1}. For brevity, we omit the $K = 7$ results. As a reminder, all typical GST approaches deterministically specify an alpha-spending function, while our approach is the only one that computes the alpha-spending function (equivalently the $\theta$ cutoffs) through an optimization procedure. In other words, the blue, orange, and green lines are deterministically (not computed or optimized over) determined by Equations~\eqref{eq:alpha_spending}, respectively. The red line is our S-MILP approach induced alpha spending budget computed through the optimal $\theta$ cutoffs from solving problem~\eqref{equ:known-sigma-milp}.

First, we note that Lan-DeMets (blue line) and Pocock's (orange line) alpha-spending budgets differ than that of O'Brien-Fleming's (green line) by more aggressively spending early than later, i.e., Lan-DeMets and Pocock are more likely to reject earlier than O'Brien-Fleming (while O'Brien-Fleming has a higher likelihood to reject later). On the other hand, our S-MILP optimization approach clearly spends more of the alpha-budget even earlier than Pocock and Lan-Demets. For example, at $K = 3$ our approach already spends $0.0449$ of the $\alpha = 0.05$ budget by only $k = 1$. On the other hand, Pocock (the more aggressive early spender) only spends half $0.0226$ of the $\alpha = 0.05$ budget. This pattern is repeated even at $K = 5$ as evidenced by the red line being visibly higher than all other methods. To summarize, our S-MILP optimization approach optimizes the expected sample size by trying to increase the likelihood of rejecting early as much as possible with a relatively (relative to other popularly used GST methods) aggressive early spending behavior. We believe this sheds critical insights into the GST literature and what alpha-spending functions are ``optimal'' to minimize expected sample size.

\section{Application to Medication-Targeted Alerts for Acute Kidney Injury (AKI)}
\label{sec:application}
Acute kidney injury (AKI) is common among hospitalized individuals and is associated with risky health outcomes such as morbidity and in some cases mortality when left untreated. 
Therefore, it becomes even more critical to determine if an AKI intervention is effective or ineffective as soon as possible for patient health. 

\citet{wilson2023randomized} is one such study from the National Institutes of Health that utilizes a randomized controlled clinical trial to study the effects of an AKI intervention. Specifically, the study aims to determine whether an automated clinical decision support system (equivalent to a heavily trained AI agent) can help reduce the progression of AKI among patients. In this study, patients are prescribed one or more potentially nephrotoxic medications. The intervention is a randomized real-time electronic health record pop-up alert that prompts the providers to consider discontinuing potentially nephrotoxic medications. The control (non-intervention) is the usual care group who were randomized to a silent alert. 

We focus on the subpopulation of patients with a proton pump inhibitor (PPI) as this was the subgroup of interest in the original paper that found some statistically significant effect \citep{wilson2023randomized}. Specifically, there are 3298 patients who received a PPI, of which 1644 patients are in the control group and 1654 are in the intervention group. The primary outcome is a composite health indicator that captures whether the patient has a progression of acute kidney injury, dialysis, or even death within 14 days. The authors found a 4\% reduction in negative health outcomes in the treatment group (27\%) versus the control group (31\%) with a $p$-value of 0.02 using all $N = 3298$ samples. We show the benefits of our method in this clinical trial by hypothetically rerunning the experiment ``sequentially'' and seeing if we can get the same statistically significant conclusion with fewer samples. 

For simplicity, when running the GST, we take the first $1600$ patients in control and the first $1600$ patients in treatment. We run the GST for $K = 8$ stages, with equal sample sizes. That is, $n_{1:1} = 400$ (200 from the treatment and 200 from control), $n_{1:2} = 800$ (400 from the treatment and 400 from the control), $\dots$, and $n_{1:8} = 3200$. We use our optimal S-MILP approach as well as the O'Brien-Fleming GST approach for comparison (For a two-sided test, the alpha-spending function of O'Brien-Fleming is $\alpha(k)_{\text{O'Brien-Fleming}} =2 - 2\Phi(z_{1-\alpha/2}/\sqrt{k/K})$). Since the outcome is binary, we use a two-proportion $z$-test with some modifications. First, we consider a Normal approximation of the binary outcomes, with a two-proportion z-test. Specifically, we test the following null and alternative, where $p_x$ and $p_y$ are the population average of the primary negative health outcome for the control and treatment, respectively: 
\begin{eqnarray*}
        H_0: p_x = p_y \quad \mbox{ and } \quad
        H_a: p_x \neq p_y. 
\end{eqnarray*}
Then, we use the classical two sample proportion $z$-test: 
$$
S_k = \frac{\hat p_{x,k} - \hat p_{y,k}}{\sqrt{2\hat p_{x,k}(1-\hat p_{x,k})/n_{1:k} + 2\hat p_{y,k}(1-\hat p_{y,k})/n_{1:k}}},
$$
where $\hat p_{x,k}$ and $\hat p_{y,k}$ are the sample mean for the control and treatment group calculated from samples up until stage $k$, respectively. 
Under $H_0$, $S = (S_k)^K_{k=1}$ is asymptotically distributed as $\mathcal{N}(0, \Sigma)$, where $\Sigma_{k,k} = 1$ for all $k\in[K]$ and $\Sigma_{i,j} = \sqrt{\min\{i,j\}/\max\{i,j\}}$ (since we split sample sizes equally among all stages). Under $H_a$, $S_k$ is asymptotically distributed as $\mathcal{N}(\mu_a(p_x, p_y), \Sigma)$, where the $k$-th element of $\mu_a(p_x, p_y)$ is $\sqrt{n_{1:k}}(p_x - p_y)/\sqrt{2p_x(1-p_x) + 2p_y(1-p_y)}$. As our SAA approach requires sampling $S_k$ under $H_a$, we specify both the values of $p_x$ and $p_y$ under our alternative. We do this by using the sample means estimated from all samples in the dataset. Finally, we reject $H_0$ in stage $k\in[K]$ if $|S_k|\geq \theta_k$, where $\theta_k$ is obtained from our S-MILP formulation or the O'Brien-Fleming $\alpha$-spending budget function. We set $\alpha = 0.05$ and also ignore the type-2 error constraint.

We perform two analyses that showcase our method. In the first approach, we assume the order in which the dataset was labeled by the initial authors was the sequence in which the patients arrived or were tested (\citet{wilson2023randomized} publicly shares the data through a clean CSV file). For this approach, we replicate the real experiment assuming the fixed ordering, i.e., at $K = 1$ we test our hypothesis on the first 200 patients that appear on rows 1-200 on both the treatment and control dataset. With this approach, our S-MILP approach stops by $K = 2$ while O'Brien-Fleming stops at $K = 3$, showing our method would have saved 400 samples over O'Brien-Fleming while reaching the same statistically significant conclusion. 

In our second approach, we aim to be robust to the sequence of the units. To do this, we randomly permute the patient order. Then, for each permutation we repeat our previous procedure. With $2000$ permutations, we calculate the average number of patients needed before making a statistically significant conclusion. Our S-MILP approach leads to an average of $2491$ samples before rejecting the null while O'Brien-Fleming's method requires an average of $2666$ samples before rejection. We remark that our approach still reaches the same statistically significant conclusion earlier than O'Brien-Fleming's approach by $175$ patients. Furthermore, the actual sample size of the study was 3298 patients. This shows how our method can reach the same statistically significant conclusion $807$ patients earlier than the original study, hopefully allowing these $807$ patients to all receive the desirable intervention as opposed to being kept in the experiment. 

\section{Concluding remarks}
\label{sec:conclusion}

This paper provides a general-purpose framework
for computing optimal group sequential designs.
Existing methods either specify rejection cutoffs
through deterministic alpha-spending functions
\citep{obrien_fleming_GST, Lan_DeMets_GST,
pocock_GST}, with no optimality guarantees, or
derive optimal cutoffs by reformulating the
design problem as a Bayes sequential decision
problem and solving it by dynamic programming
\citep{eales_jennison_1992,
hampson_jennison_2013}, which requires the user
to specify a prior, a loss function, and a
sampling cost, and limits applicability to
settings in which the running test statistic is
sufficient. Our S-MILP approach computes the
optimal cutoffs directly by solving the
constrained design problem, requiring only the
alternative at which power is desired---a
quantity any sample-size calculation already
needs. The framework accommodates $z$-tests,
$t$-tests, two-sample tests, one- and two-sided
tests, symmetric designs with futility, and a
wide class of non-normal settings within a single
unified formulation, with finite-sample
convergence guarantees for all variants.

The simulations in Section~\ref{sec:sims} reveal
a substantive empirical finding beyond the
efficiency comparisons. Across all settings we
examined, the S-MILP optimal design spends the
type-1 budget aggressively early---more so than
Pocock or Lan-DeMets, and far more so than
O'Brien-Fleming. This pattern speaks to a
long-standing debate in the GST literature on
which alpha-spending profiles are most efficient.
Our results suggest that, when the goal is
minimizing expected sample size to rejection,
front-loading the alpha-budget is preferable to
back-loading it. We view a structural
characterization of this finding---identifying
conditions under which aggressive early spending
provably dominates---as an important direction
for future theoretical work.

Two practical limitations deserve mention.
First, the method is computationally heavier than
classical alpha-spending designs. Our SAA step
requires $M$ simulated paths from the
test-statistic distribution under both the null
and the alternative, and the resulting MILP must
be solved over these $M$ samples. For the
clinical-trial application of
Section~\ref{sec:application}, with $K = 5$ and
$M \approx 2000$, the optimization took
approximately 30 minutes on a standard computer
using a commercial solver. This is a one-time
design-stage cost, paid before the trial begins,
and we expect it to decrease substantially with
parallelization and recent advances in MILP
solvers \citep{luedtke2010integer}.

Second, the method computes the design that is
optimal at a specific alternative $\mu_a$. If the
true effect differs substantially from $\mu_a$,
the optimality guarantee no longer holds; this is
a feature shared with all existing
optimal-design methods, including the Bayes-DP
approaches of \citet{eales_jennison_1992} and
successors. A formal robustness theory is left
to future work.

Several extensions are natural. The framework
extends in principle to adaptive designs in which
group sizes are chosen based on observed data, to
multi-arm trials with multiple treatment
comparisons, and to settings with delayed
responses where pipeline subjects must be
incorporated into the final analysis
\citep{hampson_jennison_2013}. Each extension
introduces additional constraint structure but
remains within the MILP paradigm. Computational
scaling to large $K$ is a related practical
question; current commercial solvers handle the
$K \leq 10$ regime well, but $K = 50$ or beyond
will require either algorithmic improvements or a
move to specialized decomposition methods
\citep{clautiaux2025last}.

\section*{Data Availability Statement}
All simulation code to generate figures and tables for Section~\ref{sec:sims} can be publicly shared. The application dataset is publicly available as open access in \citep{wilson2023randomized}.

\clearpage

\bibliographystyle{ormsv080}
\bibliography{Ref}

@article{eales_jennison_1992,
  author  = {Eales, J. D. and Jennison, C.},
  title   = {An improved method for deriving optimal
             one-sided group sequential tests},
  journal = {Biometrika},
  volume  = {79},
  number  = {1},
  pages   = {13--24},
  year    = {1992},
  doi     = {10.1093/biomet/79.1.13}
}

@article{hampson_jennison_2013,
  author  = {Hampson, L. V. and Jennison, C.},
  title   = {Group sequential tests for delayed
             responses (with discussion)},
  journal = {Journal of the Royal Statistical
             Society, Series B},
  volume  = {75},
  number  = {1},
  pages   = {3--54},
  year    = {2013},
  doi     = {10.1111/j.1467-9868.2012.01030.x}
}

@book{shapiro2021lectures,
  title={Lectures on stochastic programming: modeling and theory},
  author={Shapiro, Alexander and Dentcheva, Darinka and Ruszczynski, Andrzej},
  year={2021},
  publisher={SIAM}
}

@article{Lachin1981_ttest,
  title={Introduction to sample size determination and power analysis for clinical trials.},
  author={John M. Lachin},
  journal={Controlled clinical trials},
  year={1981},
  volume={2 2},
  pages={
          93-113
        }
}

@book{cohen1988_ttest,
  author = {Cohen, J.},
  publisher = {Lawrence Erlbaum Associates},
  title = {{Statistical Power Analysis for the Behavioral Sciences}},
  year = 1988
}

@article{wang2008sample,
  title={Sample average approximation of expected value constrained stochastic programs},
  author={Wang, Wei and Ahmed, Shabbir},
  journal={Operations Research Letters},
  volume={36},
  number={5},
  pages={515--519},
  year={2008},
  publisher={Elsevier}
}

@article{still2018lectures,
  title={Lectures on parametric optimization: An introduction},
  author={Still, Georg},
  journal={Optimization Online},
  pages={2},
  year={2018}
}

@article{mityagin2015zero,
  title={The zero set of a real analytic function},
  author={Mityagin, Boris},
  journal={arXiv preprint arXiv:1512.07276},
  year={2015}
}

@article{pagnoncelli2009sample,
  title={Sample average approximation method for chance constrained programming: theory and applications},
  author={Pagnoncelli, Bernardo K and Ahmed, Shabbir and Shapiro, Alexander},
  journal={Journal of optimization theory and applications},
  volume={142},
  number={2},
  pages={399--416},
  year={2009},
  publisher={Springer}
}

@inproceedings{michael_paper2,
author = {Lindon, Michael and Sanden, Chris and Shirikian, Vach\'{e}},
title = {Rapid Regression Detection in Software Deployments through Sequential Testing},
year = {2022},
isbn = {9781450393850},
publisher = {Association for Computing Machinery},
address = {New York, NY, USA},
url = {https://doi.org/10.1145/3534678.3539099},
doi = {10.1145/3534678.3539099},
pages = {3336–3346},
numpages = {11},
location = {Washington DC, USA},
series = {KDD '22}
}

@article{mypaper,
  title={Using Machine Learning to Test Causal Hypotheses in Conjoint Analysis},
  author={Ham, Dae Woong and Imai, Kosuke and Janson, Lucas},
  doi = {10.48550/arXiv.2201.08343},
  year={2022}
}

@misc{ian_bandit,
  doi = {10.48550/ARXIV.2210.10768},
  
  url = {https://arxiv.org/abs/2210.10768},
  
  author = {Waudby-Smith, Ian and Wu, Lili and Ramdas, Aaditya and Karampatziakis, Nikos and Mineiro, Paul},
  
  title = {Anytime-valid off-policy inference for contextual bandits},
  
  publisher = {arXiv},
  
  year = {2022},
  
  copyright = {arXiv.org perpetual, non-exclusive license}
}

@inproceedings{
michael_paper,
title={Anytime-Valid Inference For Multinomial Count Data},
author={Michael Lindon and Alan Malek},
booktitle={Advances in Neural Information Processing Systems},
editor={Alice H. Oh and Alekh Agarwal and Danielle Belgrave and Kyunghyun Cho},
year={2022},
url={https://openreview.net/forum?id=a4zg0jiuVi}
}

@inproceedings{GST_book,
  title={Group Sequential and Confirmatory Adaptive Designs in Clinical Trials},
  author={Gernot Wassmer and Werner Brannath},
  year={2016}
}

@article{Ramesh_anytime,
author = {Johari, Ramesh and Pekelis, Leonid and Walsh, David},
year = {2015},
month = {12},
pages = {},
title = {Always Valid Inference: Bringing Sequential Analysis to A/B Testing}
}

@article{howard_nonasymp,
  title={Time-uniform, nonparametric, nonasymptotic confidence sequences},
  author={Steven R. Howard and Aaditya Ramdas and Jon D. McAuliffe and Jasjeet S. Sekhon},
  journal={The Annals of Statistics},
  year={2020}
}

@article{Lan_DeMets_GST,
 ISSN = {00063444},
 URL = {http://www.jstor.org/stable/2336502},
 author = {K. K. Gordon Lan and David L. DeMets},
 journal = {Biometrika},
 number = {3},
 pages = {659--663},
 publisher = {[Oxford University Press, Biometrika Trust]},
 title = {Discrete Sequential Boundaries for Clinical Trials},
 urldate = {2026-01-22},
 volume = {70},
 year = {1983}
}

@article{obrien_fleming_GST,
 ISSN = {0006341X, 15410420},
 URL = {http://www.jstor.org/stable/2530245},
 author = {Peter C. O'Brien and Thomas R. Fleming},
 journal = {Biometrics},
 number = {3},
 pages = {549--556},
 publisher = {[Wiley, International Biometric Society]},
 title = {A Multiple Testing Procedure for Clinical Trials},
 urldate = {2026-01-22},
 volume = {35},
 year = {1979}
}

@article{pocock_GST,
 ISSN = {00063444, 14643510},
 URL = {http://www.jstor.org/stable/2335684},
 author = {Stuart J. Pocock},
 journal = {Biometrika},
 number = {2},
 pages = {191--199},
 publisher = {[Oxford University Press, Biometrika Trust]},
 title = {Group Sequential Methods in the Design and Analysis of Clinical Trials},
 urldate = {2026-01-22},
 volume = {64},
 year = {1977}
}

@article{JRSSB_GST,
    author = {Silva, Ivair R. and Kulldorff, Martin and Katherine Yih, W.},
    title = {Optimal Alpha Spending for Sequential Analysis with Binomial Data},
    journal = {Journal of the Royal Statistical Society Series B: Statistical Methodology},
    volume = {82},
    number = {4},
    pages = {1141-1164},
    year = {2020},
    month = {06},
    issn = {1369-7412},
    doi = {10.1111/rssb.12379},
    url = {https://doi.org/10.1111/rssb.12379},
    eprint = {https://academic.oup.com/jrsssb/article-pdf/82/4/1141/49323656/jrsssb_82_4_1141.pdf},
}

@misc{ham_DB,
      title={Design-Based Confidence Sequences: A General Approach to Risk Mitigation in Online Experimentation}, 
      author={Dae Woong Ham and Iavor Bojinov and Michael Lindon and Martin Tingley},
      year={2023},
      eprint={2210.08639},
      archivePrefix={arXiv},
      primaryClass={stat.ME},
      url={https://arxiv.org/abs/2210.08639}, 
}

@misc{spotify_blog1,
      title={Choosing a Sequential Testing Framework — Comparisons and Discussions}, 
      author={Marten Schultzberg and Sebastian Ankargren},
      year={2023},
      url={https://engineering.atspotify.com/2023/03/choosing-sequential-testing-framework-comparisons-and-discussions}, 
}

@book{clinical_trials_GST,
title = "Group sequential tests with applications to clinical trials",
author = "Christopher Jennison and Turnbull, \{Bruce W.\}",
year = "1999",
month = sep,
day = "15",
language = "English",
isbn = "9780849303166",
series = "Chapman \&amp; Hall/CRC Interdisciplinary Statistics",
publisher = "Chapman \& Hall",
address = "UK United Kingdom",
}

@misc{gurobi,
  author = {{Gurobi Optimization, LLC}},
  title = {{Gurobi Optimizer Reference Manual}},
  year = 2026,
  url = "https://www.gurobi.com"
}

@misc{cplex,
  author       = {{IBM}},
  title        = {{IBM ILOG CPLEX Optimization Studio}},
  year         = {2025},
  howpublished = {\url{https://www.ibm.com/products/ilog-cplex-optimization-studio}},
  note         = {Version 22.1}
}

@article{laurent2000adaptive,
  title={Adaptive estimation of a quadratic functional by model selection},
  author={Laurent, Beatrice and Massart, Pascal},
  journal={Annals of statistics},
  pages={1302--1338},
  year={2000},
  publisher={JSTOR}
}

@book{billingsley2013convergence,
  title={Convergence of probability measures},
  author={Billingsley, Patrick},
  year={2013},
  publisher={John Wiley \& Sons}
}

@article{walker1968note,
  title={A note on the asymptotic distribution of sample quantiles},
  author={Walker, AM},
  journal={Journal of the Royal Statistical Society Series B: Statistical Methodology},
  volume={30},
  number={3},
  pages={570--575},
  year={1968},
  publisher={Oxford University Press}
}

@article{wilson2023randomized,
  title={A randomized clinical trial assessing the effect of automated medication-targeted alerts on acute kidney injury outcomes},
  author={Wilson, F Perry and Yamamoto, Yu and Martin, Melissa and Coronel-Moreno, Claudia and Li, Fan and Cheng, Chao and Aklilu, Abinet and Ghazi, Lama and Greenberg, Jason H and Latham, Stephen and others},
  journal={Nature communications},
  volume={14},
  number={1},
  pages={2826},
  year={2023},
  publisher={Nature Publishing Group UK London}
}

@article{clautiaux2025last,
  title={Last fifty years of integer linear programming: A focus on recent practical advances},
  author={Clautiaux, Fran{\c{c}}ois and Ljubi{\'c}, Ivana},
  journal={European Journal of Operational Research},
  volume={324},
  number={3},
  pages={707--731},
  year={2025},
  publisher={Elsevier}
}

@article{luedtke2010integer,
  title={An integer programming approach for linear programs with probabilistic constraints},
  author={Luedtke, James and Ahmed, Shabbir and Nemhauser, George L},
  journal={Mathematical programming},
  volume={122},
  number={2},
  pages={247--272},
  year={2010},
  publisher={Springer}
}

\clearpage 

\begin{center}
{\large\bf SUPPLEMENTARY MATERIAL}
\end{center}

\begin{appendices}

\section{Additional Materials for Section~\ref{sec:milp-reformulation}}\label{sec-appendix:MILP-reformulation}
\subsection{Proof of Proposition~\ref{prop:milp-equivalence}}
\begin{proof}[Proof of Proposition~\ref{prop:milp-equivalence}]
Let $v_{\mathrm{ip}}$ be the optimal value of \eqref{equ:known-sigma-milp} and let $v_{\mathrm{saa}}$ be the optimal value of \eqref{equ:SAA}. We first show $v_{\mathrm{ip}}\geq v_{\mathrm{saa}}$ and then $v_{\mathrm{saa}}\geq v_{\mathrm{ip}}$.

    Let $\hat\theta_k, \hat w^m, \hat\rho^m_k, \hat\tau^m_k$ for $k\in[K], m\in[M]$ be an optimal solution to \eqref{equ:known-sigma-milp}. We show that $\hat\theta_k$ for $k\in[K]$ is a feasible solution to \eqref{equ:SAA} and thus $v_{\mathrm{ip}}\geq v_{\mathrm{saa}}$.
    According to \eqref{equ:known-sigma-milp-c1}, we have $\hat w^m\geq 1-\I(S_k^m - \hat\theta_k\leq 0, \forall k\in[K])$. According to \eqref{equ:known-sigma-milp-c2}, $1-\alpha\leq (1/M)\sum_{m\in[M]}(1-\hat w^m)\leq (1/M)\sum_{m\in[M]}\I(S_k^m - \hat \theta_k\leq 0, \forall k\in[K])$. Thus $(\hat\theta_k)^K_{k=1}$ satisfies \eqref{equ:saa-c1}. \eqref{equ:known-sigma-milp-c3} implies $\hat\rho^m_k\geq \I(S^m_{a,k}<\hat\theta_k)$ for all $k\in[K]$ and \eqref{equ:known-sigma-milp-c4} implies $\hat\tau^m_k\geq \prod^k_{i=1}\I(S^m_{a,i}<\hat\theta_i) = \I(S^m_{a,i}<\hat\theta_i,\forall i\in[k])$. According to \eqref{equ:known-sigma-milp-c5}, $\beta\geq (1/M)\sum^M_{m=1}\hat\tau^m_K\geq (1/M)\sum^M_{m=1}\I(S^m_{a,k}< \hat\theta_k, \forall k\in[K])$. Thus $(\hat\theta_k)^K_{k=1}$ is feasible to \eqref{equ:saa-c2}.  
    Now according to \eqref{equ:known-sigma-milp-ob}, $v_{\mathrm{ip}} = n_1 + (1/M)\sum^M_{m=1}\sum^K_{k=2}n_k\hat\tau^m_{k-1}\geq n_1+(1/M)\sum^{M}_{m=1}\sum^K_{k=2}n_k\I(S^m_{a,i}<\hat\theta_i,\forall i\in[k-1])\geq v_{\mathrm{saa}}$, where the last inequality follows since $(\hat\theta_k)^K_{k=1}$ is feasible to \eqref{equ:SAA}.

    On the other hand, let $\tilde\theta_k$ for $k\in[K]$ be an optimal solution to \eqref{equ:SAA}. Then, let $\tilde w^m = 1-\I(S_k^m - \tilde\theta_k\leq 0, \forall k\in[K])$, $\tilde\rho^m_k = \I(S^m_{a,k}<\tilde\theta_k)$ and $\tilde\tau^m_k = \I(S^m_{a,i}<\tilde\theta_i ,\forall i\in[k])$ for all $m\in[M]$ and $k\in[K]$. Then \eqref{equ:known-sigma-milp-c1}, \eqref{equ:known-sigma-milp-c3} and \eqref{equ:known-sigma-milp-c4} are feasible. \eqref{equ:known-sigma-milp-c2} is feasible according to \eqref{equ:saa-c1}, and \eqref{equ:known-sigma-milp-c5} is feasible according to \eqref{equ:saa-c2}. As a result, $\tilde\theta_k, \tilde w^m, \tilde\rho^m_k, \tilde\tau^m_k$ for $k\in[K], m\in[M]$ is feasible to \eqref{equ:known-sigma-milp}. Then $v_{\mathrm{ip}} \leq n_1+(1/M)\sum^M_{m=1}\sum^K_{k=2}n_k\tilde\tau^m_{k-1} = n_1+(1/M)\sum^M_{m=1}\sum^K_{k=2}n_k\I(S^m_{a,i}<\tilde\theta_i,\forall i\in[k-1]) = v_{\mathrm{saa}}$. 

\end{proof}
% \subsection{Proof of Lemma~\ref{lem:compactness}}
% \begin{proof}[Proof of Lemma~\ref{lem:compactness}]
%     It suffices to show that the feasible set $\Theta$ of problem~\eqref{equ:true} is bounded.  
%     Let $\hat\theta = (\hat\theta_k)^K_{k=1}\in\Theta$. 
%     Let $z_{1-\alpha}$ be the $1-\alpha$ quantile of a standard Normal distribution. 
%     According to \eqref{equ:true-c1}, we must have $\hat\theta_k\geq z_{1-\alpha}$ for all $k\in[K]$, since 
%     $$
%     1-\alpha\leq \Pr(Z_k\leq \hat\theta_k, \,\forall k\in[K])\leq \Pr(Z_k\leq \hat\theta_k), \,\forall k\in[K]. 
%     $$
%     On the other hand, letting $\bar\theta_k = ax\{\theta_k: \beta\geq \Pr(Z_{k'}\leq z_{1-\alpha} - \sqrt{n_{1:k'}}\delta, \,\forall k'\neq k, \text{ and } Z_k\leq \hat\theta_k - \sqrt{n_{1:k}}\delta)\}$ for all $k\in[K]$, we have 
%     $$
%     \beta\geq \Pr(Z_k\leq \theta^*_k - \sqrt{n_{1:k}}\delta, \,\forall k\in[K]) \geq \Pr(Z_{k'}\leq z_{1-\alpha} - \sqrt{n_{1:k'}}\delta, \,\forall k'\neq k, \text{ and } Z_k\leq \theta^*_k - \sqrt{n_{1:k}}\delta), \forall k\in[K],  
%     $$
%     where the second inequality follows since we have shown that $\theta^*_k\geq z_{1-\alpha}$ for each $k\in[K]$. Thus we must also have $\theta^*_k\leq \bar\theta_k$. 

%     Arguments in (ii) follow straightforwardly, as we must have $|f(\theta)|\leq n_{1:K}$ for all $\theta\inathbb{R}^K$, and $Z = (Z_k)^K_{k=1}$ has a multivariate normal distribution.
% \end{proof}

\section{Additional Materials for Section~\ref{sec:convergence-main}}\label{appendix-sec:main-sample-convergence}
Recall that $v^*(\alpha, \beta)$ and $\theta^*(\alpha, \beta) = (\theta_k^*(\alpha, \beta))^K_{k=1}$ are the optimal value and any optimal solution of problem~\eqref{equ:main}, respectively. $\hat v_M(\alpha, \beta)$ and $ \hat\theta_M(\alpha, \beta) = (\hat\theta_{M,k}(\alpha, \beta))^K_{k=1}$ are the optimal value and any optimal solution of the SAA problem~\eqref{equ:SAA}, respectively. $\Theta(\alpha, \beta)$ is the feasible region of true problem \eqref{equ:main}, and $\Theta_M(\alpha, \beta)$ is the feasible region of the SAA problem \eqref{equ:SAA}. $S(\alpha, \beta)$ and $S_M(\alpha, \beta)$ are the sets of optimal solutions to the true problem \eqref{equ:main} and the SAA problem \eqref{equ:SAA} respectively. $f(\theta) = n_1 + \sum^K_{k=2}n_k\Pr(S_i \leq \theta_i, \forall i\in[k-1]|H_a)$ is the objective, $g_0(\theta) = \Pr(S_k\leq \theta_k, \,\forall k\in[K]|H_0)$ is the LHS of the first constraint and $g_a(\theta) = \Pr(S_k\leq \theta_k, \,\forall k\in[K]|H_a)$ is the LHS of the second constraint of \eqref{equ:main}, respectively. 

Additionally, in this section we let $\hat f_M(\theta) = n_1 + (1/M)\sum^M_{m=1}\sum^K_{k=2}n_k\I(S^m_{a,i} < \theta_i, \forall i\in[k-1])$ be the objective of the SAA problem \eqref{equ:SAA}, $\hat g_{0,M}(\theta) = (1/M)\sum^M_{m=1}\I(S^m_k\leq \theta_k, \,\forall k\in[K])$ and $\hat g_{a,M}(\theta) = (1/M)\sum^M_{m=1}\I(S^m_{a,k}<\theta_k, \,\forall k\in[K])$ be the LHS of the first constraint and second constraint of SAA problem \eqref{equ:SAA}, respectively.

\subsection{Proof of Lemma~\ref{lem:SAA-feasibility}}

\begin{proof}[Proof of Lemma~\ref{lem:SAA-feasibility}]
    As $\Theta(\alpha - \bar\epsilon, \beta -\bar\epsilon)\neq\emptyset$, we let $\hat\theta$ be arbitrary such that
    $\hat\theta\in \Theta(\alpha - \bar\epsilon, \beta -\bar\epsilon)$. Then,
    \begin{equation*}
        \begin{split}
    &\Pr\big(\Theta_M(\alpha, \beta)\neq\emptyset\big)\\
    \geq &\Pr\big(\hat\theta\in \Theta_M(\alpha, \beta)\big)\\
    = & \Pr\big(\hat g_{0,M}(\hat\theta)\geq 1-\alpha \text{ and }\hat g_{a, M}(\hat\theta)\leq \beta\big)\\
    \geq & \Pr\big(\hat g_{0,M}(\hat\theta)\geq g_0(\hat\theta)-\bar\epsilon \text{ and }\hat g_{a,M}(\hat\theta)\leq g_a(\hat\theta)+\bar\epsilon\big)\\
    \geq & 1- \Pr\big(\hat g_{0,M}(\hat\theta)\leq g_0(\hat\theta)-\bar\epsilon\big) - \Pr\big(\hat g_{a, M}(\hat\theta)\geq g_a(\hat\theta)+\bar\epsilon\big)\\
    \geq & 1-2\exp(-2M\bar\epsilon^2), 
        \end{split}
    \end{equation*}
\sloppy 
where the last inequality follows from Hoeffding's inequality. The argument then follows by letting $2\exp(-2M\bar\epsilon^2) = \eta$. 
\end{proof}
\subsection{Proofs of Proposition~\ref{prop:feasible-set} and Corollary~\ref{prop:prob-feasibility}}

Let $\hat C(\Delta) = \{\underline{\theta}, \underline{\theta} + \Delta, \underline{\theta} + 2\Delta, \cdots, \bar\theta\}^K$ be a discretization of $C$. Then we have $|\hat C(\Delta)|\leq \ceil{(\bar\theta - \underline{\theta})/\Delta}^K$. By construction of $\hat C(\Delta)$, for any $\hat\theta\in C$, there exists some $\tilde\theta\in \hat C(\Delta)$ such that $\Vert \tilde\theta - \hat\theta\Vert_\infty\leq \Delta$.

Lemma~\ref{lem:finite-feasibility} derives a high probability bound on the set of all feasible solutions of \eqref{equ:main} and \eqref{equ:SAA} restricted to $\hat C(\Delta)$. Let $\epsilon_1(\eta, M) = \sqrt{\log(4|\hat C(\Delta)|/\eta)/(2M)}$. 
\begin{lemma}\label{lem:finite-feasibility}
    With probability at least $1 - \eta$, we have $\Theta(\alpha - \epsilon_1(\eta, M), \beta -\epsilon_1(\eta, M))\cap \hat C(\Delta)\subseteq \Theta_M(\alpha, \beta)\cap \hat C(\Delta)\subseteq \Theta(\alpha + \epsilon_1(\eta, M), \beta+ \epsilon_1(\eta, M))\cap \hat C(\Delta)$. 
\end{lemma}
\begin{proof}[Proof of Lemma~\ref{lem:finite-feasibility}]
    According to Proposition 1 of \citet{wang2008sample}, 
    \begin{align*}
        &\Pr\bigg(\Theta(\alpha - \epsilon, \beta -\epsilon)\cap \hat C(\Delta)\subseteq \Theta_M(\alpha, \beta)\cap \hat C(\Delta)\subseteq \Theta(\alpha + \epsilon, \beta+ \epsilon)\cap \hat C(\Delta)\bigg)\\
    \geq& 1-2|\hat C(\Delta)|\exp(-M\epsilon^2/2\sigma_{g_0}^2) - 2|\hat C(\Delta)|\exp(-M\epsilon^2/2\sigma_{g_a}^2),
    \end{align*}
    where $\sigma_{g_0}^2 = \max_{\theta\in C}\text{Var}(\I(S^m_k\leq \theta_k, \forall k\in[K]) - g_0(\theta))\leq 1/4$ and $\sigma_{g_a}^2 = \max_{\theta\in C}\text{Var}(\I(S^m_{a,k}< \theta_k, \forall k\in[K]) - g_a(\theta))\leq 1/4$. Thus we have
    $$
    \Pr\bigg(\Theta(\alpha - \epsilon, \beta -\epsilon)\cap \hat C(\Delta)\subseteq \Theta_M(\alpha, \beta)\cap \hat C(\Delta)\subseteq \Theta(\alpha + \epsilon, \beta+ \epsilon)\cap \hat C(\Delta)\bigg)
    \geq 1- 4|\hat C(\Delta)|\exp(-2M\epsilon^2).
    $$
By letting $4|\hat C(\Delta)|\exp(-2M\epsilon_1(\eta, M)^2) = \eta$, we have $\epsilon_1(\eta, M) = \sqrt{\log(4|\hat C(\Delta)|/\eta)/(2M)}$. 
\end{proof}

\begin{lemma}\label{lem:lipshitz-continuity}
$f(\theta), g_0(\theta)$ and $g_a(\theta)$ are increasing in $\theta\in C$. In addition, 
    there exist positive constants $l_f, l_{g_0}$ and $l_{g_a}$ such that the following holds:
    \begin{align*}
        |f(\hat\theta) - f(\tilde\theta)|\leq l_f\Vert\hat\theta - \tilde\theta\Vert_1, \forall \hat\theta, \tilde\theta\in C,\\
        |g_0(\hat\theta) - g_0(\tilde\theta)|\leq l_{g_0}\Vert\hat\theta - \tilde\theta\Vert_1, \forall \hat\theta, \tilde\theta\in C,\\
        |g_{a}(\hat\theta) - g_a(\tilde\theta)|\leq l_{g_a}\Vert\hat\theta - \tilde\theta\Vert_1, \forall \hat\theta, \tilde\theta\in C. 
    \end{align*}
\end{lemma}
\begin{proof}[Proof of Lemma~\ref{lem:lipshitz-continuity}]
    Let $l_f = \max_{k\in[K], \hat\theta\in C}|\partial f(\hat\theta)/\partial \theta_k|$, $l_{g_0} = \max_{k\in[K], \hat\theta\in C}|\partial g_0(\hat\theta)/\partial \theta_k|$, and $l_{g_a} = \max_{k\in[K], \hat\theta\in C}|\partial g_a(\hat\theta)/\partial \theta_k|$. Then we have $l_f, l_{g_0}, l_{g_a}<\infty$. The result follows by applying the multivariate mean value theorem, which says that for a continuous and differentiable function $h: C \to \mathbb{R}$, there exists a point $\theta_0\in C$ such that $h(\hat\theta) - h(\tilde\theta) = \nabla h(\theta_0)\cdot (\hat\theta - \tilde\theta)$, and so 
    $|h(\hat\theta) - h(\tilde\theta)| \leq \max_{k\in[K], \theta_0\in C}|\partial h(\theta_0)/\partial \theta_k| \Vert\hat\theta - \tilde\theta\Vert_1$. 
    
\end{proof}

\begin{lemma}\label{lem:SAA-constraint-bound}
There exist constants $c_{g_0}, c_{g_a}, c_f <\infty$ such that the following arguments hold:

(i) With probability at least $1-\eta$, 
\begin{align*}
    \max\big\{\hat g_{0,M}(\tilde\theta) - \min_{\theta: \Vert\theta - \tilde\theta\Vert_\infty\leq\Delta}\hat g_{0,M}(\theta), 
    \max_{\theta: \Vert\theta - \tilde\theta\Vert_\infty\leq\Delta}\hat g_{0,M}(\theta) - \hat g_{0,M}(\tilde\theta)\big\}\\
    \leq c_{g_0}K\Delta + \epsilon_1(4\eta, M), \,\forall \tilde\theta\in \hat C(\Delta).
\end{align*}

(ii) With probability at least $1-\eta$, 
\begin{align*}
    \max\big\{\hat g_{a,M}(\tilde\theta) - \min_{\theta: \Vert\theta - \tilde\theta\Vert_\infty\leq\Delta}\hat g_{a,M}(\theta), \max_{\theta: \Vert\theta - \tilde\theta\Vert_\infty\leq\Delta}\hat g_{a,M}(\theta) - \hat g_{a,M}(\tilde\theta)\big\}\\
\leq c_{g_a}K\Delta + \epsilon_1(4\eta, M),\, \forall \tilde\theta\in \hat C(\Delta). 
\end{align*}

(iii) With probability at least $1-\eta$,  
\begin{align*}
    \max\big\{\hat f_M(\tilde\theta) - \min_{\theta: \Vert\theta - \tilde\theta\Vert_\infty\leq\Delta}\hat f_M(\theta), \max_{\theta: \Vert\theta - \tilde\theta\Vert_\infty\leq\Delta}\hat f_M(\theta) - \hat f_M(\tilde\theta)\big\}\\
    \leq c_fK(K-1)\Delta/2 + \sum^K_{k=2}n_k\epsilon_1(4\eta/(K-1), M), \,\forall \tilde\theta\in \hat C(\Delta).
\end{align*}

\end{lemma}
\begin{proof}[Proof of Lemma~\ref{lem:SAA-constraint-bound}]
(i) For all $\theta, \tilde\theta\in C$ and $\Vert\theta - \tilde\theta\Vert_\infty\leq \Delta$, we have that 
\begin{equation*}
    \begin{split}
        &\hat g_{0,M}(\tilde\theta) - \hat g_{0,M}(\theta)\\
    \leq &(1/M)\sum^M_{m=1}\I(S^m_k\leq \max\{\tilde\theta_k, \theta_k\},\forall k\in[K]\text{ and }S^m_k > \min\{\tilde\theta_k, \theta_k\} \text{ for some }k\in[K])\\
    \leq &(1/M)\sum^M_{m=1}\I(\min\{\tilde\theta_k, \theta_k\}<S^m_k\leq \max\{\tilde\theta_k, \theta_k\} \text{ for some }k\in[K])\\
    \leq &
    (1/M)\sum^M_{m=1}\I(\tilde\theta_k - \Delta<S^m_k\leq \tilde\theta_k + \Delta \text{ for some }k\in[K]).
    \end{split}
\end{equation*}
Since the above inequality holds for all $\theta, \tilde\theta\in C$ with $\Vert\theta-\tilde\theta\Vert_\infty\leq\Delta$, it holds for $\hat g_{0,M}(\tilde\theta) - \min_{\theta: \Vert\theta-\tilde\theta\Vert_\infty\leq \Delta}\hat g_{0,M}(\theta)$ and $\max_{\theta: \Vert\theta-\tilde\theta\Vert_\infty\leq \Delta}\hat g_{0,M}(\theta) - \hat g_{0,M}(\tilde\theta)$ as well. Thus we have
\begin{equation*}
    \begin{split}
\max\{\hat g_{0,M}(\tilde\theta) - \min_{\theta: \Vert\theta-\tilde\theta\Vert_\infty\leq \Delta}\hat g_{0,M}(\theta), \max_{\theta: \Vert\theta-\tilde\theta\Vert_\infty\leq \Delta}\hat g_{0,M}(\theta) - \hat g_{0,M}(\tilde\theta)\}\\
\leq (1/M)\sum^M_{m=1}\I(\tilde\theta_k - \Delta<S^m_k\leq \tilde\theta_k + \Delta \text{ for some }k\in[K]).
    \end{split}
\end{equation*}

Let $p_{0,k}(\theta_k)$ and $p_{a,k}(\theta_k)$ be the densities of the marginal distributions of $S_k$ at $\theta_k$ under $H_0$ and $H_a$ respectively, and let $c_{g_0} = 2\max_{\theta_k\in[\underline{\theta}, \overline{\theta}]}p_{0,k}(\theta_k)$, then $\Pr(\tilde\theta_k - \Delta <S^m_k\leq \tilde\theta_k+\Delta)\leq c_{g_0}\Delta$. 
For all $\epsilon > 0$, 
    \begin{align*}
    &\Pr(\max\{\hat g_{0,M}(\tilde\theta) - \min_{\theta: \Vert\theta - \tilde\theta\Vert_\infty\leq\Delta}\hat g_{0,M}(\theta), 
    \max_{\theta: \Vert\theta - \tilde\theta\Vert_\infty\leq\Delta}\hat g_{0,M}(\theta) - \hat g_{0,M}(\tilde\theta)\} \leq c_{g_0}K\Delta + \epsilon, \forall \tilde\theta\in \hat C(\Delta))\\
    \geq &\Pr((1/M)\sum^M_{m=1}\I(\tilde\theta_k - \Delta<S^m_k\leq \tilde\theta_k+\Delta \text{ for some }k\in[K]) \leq c_{g_0}K\Delta + \epsilon, \forall \tilde\theta\in \hat C(\Delta))\\
    \geq& \Pr((1/M)\sum^M_{m=1}\I(\tilde\theta_k-\Delta<S^m_k\leq \tilde\theta_k+\Delta \text{ for some }k\in[K]) \\
    &\quad\quad\quad\quad\quad\quad
    \leq \Pr(\tilde\theta_k-\Delta<S^m_k\leq \tilde\theta_k+\Delta \text{ for some }k\in[K]) + \epsilon,\,\forall\tilde\theta\in\hat C(\Delta))\\
    \geq &1-|\hat C(\Delta)|\exp(-2M\epsilon^2),
    \end{align*}
    where the second inequality follows since $\Pr(\tilde\theta_k-\Delta<S^m_k\leq \tilde\theta_k+\Delta \text{ for some }k\in[K])\leq \sum^K_{k=1}\Pr(\tilde\theta_k - \Delta <S^m_k\leq \tilde\theta_k+\Delta)\leq c_{g_0}K\Delta$. The third inequality follows from Hoeffding's inequality. Letting $\epsilon = \epsilon_1(4\eta, M)$, we have the desired result. 

(ii) can be proved similarly as (i). We thus omit its proof.

(iii) Recall that $\hat f_M(\theta) = n_1 + (1/M)\sum^M_{m=1}\sum^K_{k=2}n_k\I(S^m_{a, i}< \theta_i, \forall i\in[k-1])$. Then, for all $\theta, \tilde\theta\in C$ and $\Vert\theta - \tilde\theta\Vert_\infty\leq \Delta$, we have that
\begin{align*}
&\hat f_M(\tilde\theta) - \hat f_M(\theta)\\
=&(1/M)\sum^M_{m=1}\sum^K_{k=2}n_k
(\I(S^m_{a,i} < \tilde\theta_i, \forall i\in[k-1]) - \I(S^m_{a,i} < \theta_i, \forall i\in[k-1])
)\\
\leq &
(1/M)\sum^M_{m=1}\sum^K_{k=2}n_k
(\I(S^m_{a,i} < \max\{\tilde\theta_i, \theta_i\}, \forall i\in[k-1] \text{ and }S^m_{a,i} \geq \min\{\tilde\theta_i, \theta_i\} \text{ for some } i\in[k-1])
)\\
\leq & (1/M)\sum^M_{m=1}\sum^K_{k=2}n_k\I(\min\{\tilde\theta_i, \theta_i\}\leq S^m_{a,i}\leq \max\{\tilde\theta_i, \theta_i\} \text{ for some }i\in[k-1])
\\
\leq & (1/M)\sum^M_{m=1}\sum^K_{k=2}n_k\I(\tilde\theta_i - \Delta \leq S^m_{a,i}\leq \tilde\theta_i + \Delta\text{ for some }i\in[k-1]). 
\end{align*}
It thus follows that $\hat f_M(\tilde\theta) - \min_{\theta: \Vert\theta - \tilde\theta\Vert_\infty\leq \Delta}\hat f_M(\theta)\leq (1/M)\sum^M_{m=1}\sum^K_{k=2}n_k\I(\tilde\theta_i - \Delta \leq S^m_{a,i}\leq \tilde\theta_i + \Delta \text{ for some }i\in[k])$ and $\max_{\theta: \Vert\theta - \tilde\theta\Vert_\infty\leq \Delta}\hat f_M(\theta) - \hat f_M(\tilde\theta)\leq (1/M)\sum^M_{m=1}\sum^K_{k=2}n_k\I(\tilde\theta_i - \Delta \leq S^m_{a,i}\leq \tilde\theta_i + \Delta \text{ for some }i\in[k-1])$. Let $c_f = 2(\max_{k\in[K]}n_k)\max_{\theta_k\in[\underline{\theta}, \overline{\theta}], k\in[K]}p_{a,k}(\theta_k)$, so $\Pr(\tilde\theta_i - \Delta <S^m_{a,i}\leq \tilde\theta_i+\Delta)\leq c_f\Delta$. Then, for any $\epsilon_2, \cdots, \epsilon_K>0$, 
 \begin{align*}
    &\Pr(\max\{\hat f_M(\tilde\theta) - \min_{\theta: \Vert\theta - \tilde\theta\Vert_\infty\leq\Delta}\hat f_M(\theta), 
    \max_{\theta: \Vert\theta - \tilde\theta\Vert_\infty\leq\Delta}\hat f_M(\theta) - \hat f_M(\tilde\theta)\} \\
    &\quad\quad\quad\quad\quad\quad\quad\leq c_fK(K-1)\Delta/2 + \sum^K_{k=2}\epsilon_k, \forall \tilde\theta\in \hat C(\Delta))\\
    \geq 
    &\Pr((1/M)\sum^M_{m=1}\sum^K_{k=2}n_k\I(\tilde\theta_i - \Delta \leq S^m_{a,i}\leq \tilde\theta_i + \Delta \text{ for some }i\in[k-1]) \\
    &\quad\quad\quad\quad\quad\quad\quad\leq c_f K(K-1)\Delta/2 + \sum^K_{k=2}\epsilon_k, \forall \tilde\theta\in \hat C(\Delta))\\
    \geq 
    &\Pr(\forall 2\leq k\leq K, \, (1/M)\sum^M_{m=1}n_k\I(\tilde\theta_i - \Delta \leq S^m_{a,i}\leq \tilde\theta_i + \Delta \text{ for some }i\in[k-1]) \\
    &\quad\quad\quad\quad\quad\quad\quad
    \leq c_f (k-1)\Delta + \epsilon_k, \forall \tilde\theta\in \hat C(\Delta))\\
    \geq&
    \Pr(\forall 2\leq k\leq K, (1/M)\sum^M_{m=1}n_k\I(\tilde\theta_i-\Delta\leq S^m_{a,i}\leq \tilde\theta_i+\Delta \text{ for some }i\in[k-1]) \\
    &\quad\quad\quad\quad\quad\quad
    \leq n_k\Pr(\tilde\theta_i-\Delta\leq S^m_{a, i}\leq \tilde\theta_i+\Delta \text{ for some }i\in[k-1]) + \epsilon_k,\,\forall\tilde\theta\in\hat C(\Delta))\\
    \geq &1-\sum^K_{k=2}|\hat C(\Delta)|\exp(-2M\epsilon_k^2/n^2_k), 
    \end{align*}
    where the third inequality follows since $n_k\Pr(\tilde\theta_i-\Delta\leq S^m_{a,i}\leq \tilde\theta_i+\Delta \text{ for some }i\in[k-1])\leq n_k\sum^{k-1}_{i=1}\Pr(\tilde\theta_i - \Delta \leq S^m_{a,i}\leq \tilde\theta_i+\Delta)\leq c_f (k-1)\Delta$. The last inequality follows from Hoeffding's inequality. Letting $\epsilon_k = n_k\epsilon_1(4\eta/(K-1), M)$, we have the desired result. 

\end{proof}

\begin{lemma}\label{lem:Delta-net}
    (i) With probability at least $1-\eta - \Pr(\Theta_M(\alpha, \beta)=\emptyset)$, $\Theta_M(\alpha, \beta)\neq\emptyset$, and for all $\hat\theta\in \Theta_M(\alpha, \beta)$, we can find $\tilde\theta\in \Theta_M(\alpha + c_{g_0}K\Delta + \epsilon_1(2\eta, M), \beta+ c_{g_a}K\Delta + \epsilon_1(2\eta, M))\cap \hat C(\Delta)$ such that $\Vert \hat\theta - \tilde\theta\Vert_\infty\leq \Delta$. 

    (ii) Suppose $\Theta(\alpha, \beta)\neq\emptyset$. Then for all $\hat\theta\in \Theta(\alpha, \beta)$, we can find $\tilde\theta\in \Theta(\alpha + l_{g_0}K\Delta, \beta+ l_{g_a}K\Delta)\cap \hat C(\Delta)$ such that $\Vert \hat\theta - \tilde\theta\Vert_\infty\leq \Delta$.
\end{lemma}
\begin{proof}[Proof of Lemma~\ref{lem:Delta-net}]
(i) By construction of $\hat C(\Delta)$, for any $\hat\theta\in C$, we have that $\{\theta\in \hat C(\Delta): \Vert\theta - \hat\theta\Vert_\infty\leq \Delta\}\neq \emptyset$, and we let $\tilde\theta(\hat\theta)\in\arg\min_{\theta\in \hat C(\Delta), \Vert\theta - \hat\theta\Vert_\infty\leq \Delta}\Vert\theta - \hat\theta\Vert_1$. It thus follows that for any $\epsilon > 0$, 
\begin{align*}
&\Pr(\Theta_M(\alpha, \beta)\neq\emptyset, \text{ and } \forall \theta\in \Theta_M(\alpha, \beta), \text{exists }\tilde\theta\in \Theta_M(\alpha + c_{g_0}K\Delta + \epsilon, \beta+ c_{g_a}K\Delta + \epsilon)\cap \hat C(\Delta)\\
&\quad\quad\quad\quad\quad
\text{ such that }\Vert \theta - \tilde\theta\Vert_\infty\leq \Delta)\\
\geq & \Pr(\Theta_M(\alpha, \beta)\neq\emptyset, \text{ and }\forall \theta\in \Theta_M(\alpha, \beta), \tilde\theta(\theta)\in \Theta_M(\alpha + c_{g_0}K\Delta + \epsilon, \beta+ c_{g_a}K\Delta + \epsilon)
)\\
\geq & \Pr(\Theta_M(\alpha, \beta)\neq\emptyset, \text{ and }\forall \theta\in \Theta_M(\alpha, \beta), 
\hat g_{0,M}(\tilde\theta(\theta)) - \hat g_{0,M}(\theta) \geq - c_{g_0}K\Delta - \epsilon, \, \\
&\quad\quad\quad\quad\quad
\hat g_{a,M}(\tilde\theta(\theta)) - \hat g_{a,M}(\theta) \leq c_{g_a}K\Delta + \epsilon
)\\
\geq & \Pr(\Theta_M(\alpha, \beta)\neq\emptyset, \text{ and }\forall \tilde\theta\in \hat C(\Delta), 
\hat g_{0,M}(\tilde\theta) - \max_{\theta: \Vert\theta - \tilde\theta\Vert_\infty\leq \Delta}\hat g_{0,M}(\theta) \geq - c_{g_0}K\Delta - \epsilon, \, \\
&\quad\quad\quad\quad\quad\quad\quad\quad\quad\quad\quad\quad\quad\quad\quad\quad\hat g_{a,M}(\tilde\theta) - \min_{\theta: \Vert\theta - \tilde\theta\Vert_\infty\leq \Delta}\hat g_{a,M}(\theta) \leq c_{g_a}K\Delta + \epsilon
)\\
\geq & 1 - \Pr(\Theta_M(\alpha, \beta)=\emptyset) - \Pr(\exists \tilde\theta\in \hat C(\Delta), 
\hat g_{0,M}(\tilde\theta) - \max_{\theta: \Vert\theta - \tilde\theta\Vert_\infty\leq \Delta}\hat g_{0,M}(\theta) \leq - c_{g_0}K\Delta - \epsilon, \, \\
&\quad\quad\quad\quad\quad\quad\quad\quad\quad\quad\quad\quad\quad\quad\quad\quad \text{ or }\hat g_{a,M}(\tilde\theta) - \min_{\theta: \Vert\theta - \tilde\theta\Vert_\infty\leq \Delta}\hat g_{a,M}(\theta) \geq c_{g_a}K\Delta + \epsilon
)\\
\geq & 1-2|\hat C(\Delta)|\exp(-2M\epsilon^2) - \Pr(\Theta_M(\alpha, \beta)=\emptyset). 
\end{align*}
In the above equations, the second inequality follows since if $\hat g_{0,M}(\tilde\theta(\theta)) - \hat g_{0,M}(\theta) \geq - c_{g_0}K\Delta - \epsilon$, we then have $\hat g_{0,M}(\tilde\theta(\theta)) \geq \hat g_{0,M}(\theta)- c_{g_0}K\Delta - \epsilon\geq 1-\alpha - c_{g_0}K\Delta - \epsilon$, and similarly if $\hat g_{a,M}(\tilde\theta(\theta)) - \hat g_{a,M}(\theta) \leq c_{g_a}K\Delta + \epsilon$, we have $\hat g_{a,M}(\tilde\theta(\theta))\leq  \hat g_{a,M}(\theta) + c_{g_a}K\Delta + \epsilon\leq \beta + c_{g_a}K\Delta + \epsilon$. The third inequality follows since otherwise if there exists $\hat\theta\in \Theta_M(\alpha, \beta)$ with $\hat g_{0,M}(\tilde\theta(\hat\theta)) - \hat g_{0,M}(\hat\theta) < - c_{g_0}K\Delta - \epsilon$ or $\hat g_{a,M}(\tilde\theta(\hat\theta)) - \hat g_{a,M}(\hat\theta) > c_{g_a}K\Delta + \epsilon$, then we can find $\tilde \theta( = \tilde\theta(\hat\theta))\in \hat C(\Delta)$ such that $\hat g_{0,M}(\tilde\theta) - \max_{\theta: \Vert\theta - \tilde\theta\Vert_\infty\leq \Delta}\hat g_{0,M}(\theta) < - c_{g_0}K\Delta - \epsilon$ or $\hat g_{a,M}(\tilde\theta) - \min_{\theta: \Vert\theta - \tilde\theta\Vert_\infty\leq \Delta}\hat g_{a,M}(\theta) > c_{g_a}K\Delta + \epsilon$. The last inequality follows from Lemmas~\ref{lem:SAA-constraint-bound} and \ref{lem:SAA-feasibility}.

    (ii) Consider an arbitrary $\hat\theta\in \Theta(\alpha, \beta)$. By construction of $\hat C(\Delta)$, we have that $\{\theta\in \hat C(\Delta): \Vert\theta - \hat\theta\Vert_\infty\leq \Delta\}\neq \emptyset$. Then we let $\tilde\theta(\hat\theta)\in\arg\min_{\theta\in \hat C(\Delta), \Vert\theta - \hat\theta\Vert_\infty\leq \Delta}\Vert\theta - \hat\theta\Vert_1$. According to Lemma~\ref{lem:lipshitz-continuity}, we have that $g_0(\tilde\theta(\hat\theta))\geq g_0(\hat\theta) - l_{g_0}\Vert \tilde\theta(\hat\theta) - \hat\theta\Vert_1\geq 1-\alpha - l_{g_0}K\Delta$, and similarly $g_a(\tilde\theta(\hat\theta))\leq g_a(\hat\theta) + l_{g_a}\Vert \tilde\theta(\hat\theta) - \hat\theta\Vert_1\leq \beta + l_{g_a}K\Delta$. It thus follows that $\tilde\theta\in \Theta(\alpha + l_{g_0}K\Delta, \beta+ l_{g_a}K\Delta)$, so $\tilde\theta(\hat\theta)\in \Theta(\alpha + l_{g_0}K\Delta, \beta+ l_{g_a}K\Delta)\cap \hat C(\Delta)$.

\end{proof}

% Proposition~\ref{prop:feasible-set} gives a high probability bound on the feasibility of $\Theta_M(\alpha, \beta)$ for the SAA problem to the true problem. 
% \begin{proposition}\label{prop:appendix-feasible-set}
% With probability at least $1-\eta$, 
% $$
% \Theta(\alpha - \hat\epsilon_\alpha(\eta, M), \beta -\hat\epsilon_\beta(\eta, M)\subseteq \Theta_M(\alpha, \beta)\subseteq \Theta(\alpha + \hat\epsilon(\eta, M), \beta+ \hat\epsilon(\eta, M)),
% $$
% where $\hat\epsilon_\alpha(\eta, M) = 2\epsilon_1(\eta, M) + (l_{g_0} + c_{g_0})K\Delta$, $\hat\epsilon_\beta(\eta, M) = 2\epsilon_1(\eta, M) + (l_{g_a} + c_{g_a})K\Delta$, and $c_{g_0}$ and $c_{g_a}$ are the same constants given in Lemma~\ref{lem:SAA-constraint-bound}, and $l_{g_0}$ and $l_{g_a}$ are the same constants given in Lemma~\ref{lem:lipshitz-continuity}. 

% \end{proposition}

\begin{proof}[Proof of Proposition~\ref{prop:feasible-set}]
We let $\hat\epsilon_\alpha(\eta, M) = 2\epsilon_1(\eta, M) + (l_{g_0} + c_{g_0})K\Delta$, $\hat\epsilon_\beta(\eta, M) = 2\epsilon_1(\eta, M) + (l_{g_a} + c_{g_a})K\Delta$, $\varepsilon_\alpha(\eta, M) = c_{g_0}K\Delta + \epsilon_1(\eta, M)$, and $\varepsilon_\beta(\eta, M) = c_{g_a}K\Delta + \epsilon_1(\eta, M)$. 

We first derive an upper bound on $\Pr(\Theta_M(\alpha, \beta)\nsubseteq \Theta(\alpha + \hat\epsilon_\alpha(\eta, M), \beta+ \hat\epsilon_\beta(\eta, M)))$, and then an upper bound on $\Pr(\Theta(\alpha - \hat\epsilon_\alpha(\eta, M), \beta -\hat\epsilon_\beta(\eta, M))\nsubseteq \Theta_M(\alpha, \beta))$. 

\textbf{Part A: Upper bound on $\Pr(\Theta_M(\alpha, \beta)\nsubseteq \Theta(\alpha + \hat\epsilon_\alpha(\eta, M), \beta+ \hat\epsilon_\beta(\eta, M)))$. }
We first have that 
\begin{align*}
    &\Pr(\Theta_M(\alpha, \beta)\nsubseteq \Theta(\alpha + \hat\epsilon_\alpha(\eta, M), \beta+ \hat\epsilon_\beta(\eta, M))) \\
     = &\, 1- \Pr(\Theta_M(\alpha, \beta)\subseteq\Theta(\alpha + \hat\epsilon_\alpha(\eta, M), \beta+ \hat\epsilon_\beta(\eta, M)))\\
= &\, 1-\Pr(\Theta_M(\alpha, \beta) = \emptyset) - \Pr(\Theta_M(\alpha, \beta) \neq \emptyset, \text{ and }\Theta_M(\alpha, \beta)\subseteq\Theta(\alpha + \hat\epsilon_\alpha(\eta, M), \beta+ \hat\epsilon_\beta(\eta, M))). 
\end{align*}
Additionally, 
\begin{align*}
    &\Pr(\Theta_M(\alpha, \beta) \neq \emptyset, \text{ and }\Theta_M(\alpha, \beta)\subseteq\Theta(\alpha + \hat\epsilon_\alpha(\eta, M), \beta+ \hat\epsilon_\beta(\eta, M))) \\
\geq &\,\Pr(\Theta_M(\alpha, \beta)\neq \emptyset, \text{ and }\Theta_M(\alpha, \beta)\subseteq \Theta(\alpha + \hat\epsilon_\alpha(\eta, M), \beta+ \hat\epsilon_\beta(\eta, M)),\\
&\quad\quad\quad\quad\text{ and }
\Theta_M(\alpha + \varepsilon_\alpha(\eta, M), \beta+ \varepsilon_\beta(\eta, M))\cap \hat C(\Delta)\\
&\quad\quad\quad\quad\quad\quad\quad
\subseteq \Theta(\alpha + \varepsilon_\alpha(\eta, M) + \epsilon_1(\eta, M), \beta+\varepsilon_\beta(\eta, M) + \epsilon_1(\eta, M))). 
\end{align*}
Now we show that if for all $\hat\theta\in \Theta_M(\alpha, \beta)$, we have $\Theta_M(\alpha + \epsilon_\alpha(\eta, M), \beta + \epsilon_\beta(\eta, M))\cap \{\theta\in \hat C(\Delta): \Vert \theta - \hat\theta\Vert_\infty\leq \Delta\}\neq\emptyset$ and $\Theta_M(\alpha + \varepsilon_\alpha(\eta, M), \beta+ \varepsilon_\beta(\eta, M))\cap \hat C(\Delta)
\subseteq \Theta(\alpha + \varepsilon_\alpha(\eta, M) + \epsilon_1(\eta, M), \beta+\varepsilon_\beta(\eta, M) + \epsilon_1(\eta, M)))$, then $\Theta_M(\alpha, \beta)\subseteq \Theta(\alpha + \hat\epsilon_\alpha(\eta, M), \beta + \hat\epsilon_\beta(\eta, M))$. Consider any $\hat\theta\in \Theta_M(\alpha, \beta)$, since $\Theta_M(\alpha + \epsilon_\alpha(\eta, M), \beta + \epsilon_\beta(\eta, M))\cap \{\theta\in \hat C(\Delta): \Vert \theta - \hat\theta\Vert_\infty\leq \Delta\}\neq\emptyset$, we can find $\tilde\theta(\hat\theta) \in \Theta_M(\alpha + \epsilon_\alpha(\eta, M), \beta + \epsilon_\beta(\eta, M))\cap \{\theta\in \hat C(\Delta): \Vert \theta - \hat\theta\Vert_\infty\leq \Delta\}$. Then $g_0(\hat\theta) \geq g_0(\tilde\theta(\hat\theta)) - l_{g_0}\Vert\hat\theta - \tilde\theta(\hat\theta)\Vert_1\geq 1-\alpha - \varepsilon_\alpha(\eta, M)  - \epsilon_1(\eta, M) - l_{g_0}K\Delta = 1-\alpha - \hat\epsilon_\alpha(\eta, M)$, where the first inequality follows from Lemma~\ref{lem:lipshitz-continuity} and that $\Vert\tilde\theta(\hat\theta) - \hat\theta\Vert_\infty\leq \Delta$. The second inequality follows since $\tilde\theta(\hat\theta)\in \Theta_M(\alpha + \varepsilon_\alpha(\eta, M), \beta+ \varepsilon_\beta(\eta, M))\cap \hat C(\Delta)\subseteq \Theta(\alpha + \varepsilon_\alpha(\eta, M) + \epsilon_1(\eta, M), \beta+\varepsilon_\beta(\eta, M) + \epsilon_1(\eta, M))$. Similarly, we can show that $g_a(\hat\theta) \leq g_a(\tilde\theta(\hat\theta)) + l_{g_a}\Vert\hat\theta - \tilde\theta(\hat\theta)\Vert_1\leq \beta + \varepsilon_\beta(\eta, M) + \epsilon_1(\eta, M) + l_{g_a}K\Delta = \beta + \hat\epsilon_\beta(\eta, M)$. Based on the above arguments, we have that 
\begin{align*}
    &\Pr(\Theta_M(\alpha, \beta)\neq \emptyset, \text{ and }\Theta_M(\alpha, \beta)\subseteq \Theta(\alpha + \hat\epsilon_\alpha(\eta, M), \beta+ \hat\epsilon_\beta(\eta, M)),\\
&\quad\quad\quad\quad
\Theta_M(\alpha + \varepsilon_\alpha(\eta, M), \beta+ \varepsilon_\beta(\eta, M))\cap \hat C(\Delta)\\
&\quad\quad\quad\quad\quad\quad\quad
\subseteq \Theta(\alpha + \varepsilon_\alpha(\eta, M) + \epsilon_1(\eta, M), \beta+\varepsilon_\beta(\eta, M) + \epsilon_1(\eta, M)))\\
\geq & \Pr(\Theta_M(\alpha, \beta)\neq \emptyset\text{ and }\forall \hat\theta\in \Theta_M(\alpha, \beta), \Theta_M(\alpha + \varepsilon_\alpha(\eta, M), \beta+ \varepsilon_\beta(\eta, M))\\
&\quad\quad\quad\quad
\cap \{\theta\in \hat C(\Delta): \Vert\theta - \hat\theta\Vert_\infty\leq \Delta\}\neq\emptyset, \\
&\quad\quad\quad\quad\quad\quad\quad
\text{and }\Theta_M(\alpha + \varepsilon_\alpha(\eta, M), \beta+ \varepsilon_\beta(\eta, M))\cap \hat C(\Delta)\\
&\quad\quad\quad\quad\quad\quad\quad\quad\quad\quad
\subseteq \Theta(\alpha + \varepsilon_\alpha(\eta, M) + \epsilon_1(\eta, M), \beta+\varepsilon_\beta(\eta, M) + \epsilon_1(\eta, M))
)\\
\geq & 1-\Pr(\Theta_M(\alpha, \beta) = \emptyset, \text{ or }\Theta_M(\alpha, \beta)\neq \emptyset \text{ and }\exists \hat\theta\in \Theta_M(\alpha, \beta), \\
&\quad\quad\quad\quad
\Theta_M(\alpha + \varepsilon_\alpha(\eta, M), \beta+ \varepsilon_\beta(\eta, M))\cap \{\theta\in \hat C(\Delta): \Vert\theta - \hat\theta\Vert_\infty\leq \Delta\} = \emptyset) \\
&\quad\quad\quad\quad\quad\quad\quad
- \Pr(\Theta_M(\alpha + \varepsilon_\alpha(\eta, M), \beta+ \varepsilon_\beta(\eta, M))\cap \hat C(\Delta)\\
&\quad\quad\quad\quad\quad\quad\quad\quad\quad\quad
\nsubseteq \Theta(\alpha + \varepsilon_\alpha(\eta, M) + \epsilon_1(\eta, M), \beta+\varepsilon_\beta(\eta, M) + \epsilon_1(\eta, M)
)\\
\geq & 1 - \Pr(\Theta_M(\alpha, \beta) = \emptyset) - 2\eta. 
\end{align*}
In the above equations, the last inequality follows from Lemmas~\ref{lem:finite-feasibility} and \ref{lem:Delta-net}(i). Combining the above arguments, we have that
\begin{align*}
    \Pr(\Theta_M(\alpha, \beta)\nsubseteq \Theta(\alpha + \hat\epsilon_\alpha(\eta, M), \beta+ \hat\epsilon_\beta(\eta, M)))\leq 2\eta. 
\end{align*}

\textbf{Part B: Upper bound on $\Pr(\Theta(\alpha - \hat\epsilon_\alpha(\eta, M), \beta -\hat\epsilon_\beta(\eta, M))\nsubseteq \Theta_M(\alpha, \beta)) $. }
If $\Theta(\alpha - \hat\epsilon_\alpha(\eta, M), \beta -\hat\epsilon_\beta(\eta, M))=\emptyset$, then we must have $\Theta(\alpha - \hat\epsilon_\alpha(\eta, M), \beta -\hat\epsilon_\beta(\eta, M))\subseteq \Theta_M(\alpha, \beta)$. Thus we focus on the case when $\Theta(\alpha - \hat\epsilon_\alpha(\eta, M), \beta -\hat\epsilon_\beta(\eta, M))\neq \emptyset$. We first have that
    \begin{equation*}
        \begin{split}
&\,\Pr(\Theta(\alpha - \hat\epsilon_\alpha(\eta, M), \beta -\hat\epsilon_\beta(\eta, M))\subseteq \Theta_M(\alpha, \beta)) \\
\geq &\,\Pr(\Theta(\alpha - \hat\epsilon_\alpha(\eta, M), \beta -\hat\epsilon_\beta(\eta, M))\subseteq \Theta_M(\alpha, \beta),\\
&\quad\quad\quad\quad
\Theta(\alpha - \hat\epsilon_\alpha(\eta, M) + l_{g_0}K\Delta, \beta -\hat\epsilon_\beta(\eta, M) + l_{g_a}K\Delta)\\
&\quad\quad\quad\quad\quad\quad\quad
\cap \hat C(\Delta)\subseteq \Theta_M(\alpha - \varepsilon_\alpha(\eta, M), \beta-\varepsilon_\beta(\eta, M) )).
        \end{split}
    \end{equation*}

Now we show that if for all $\tilde\theta\in \hat C(\Delta)$, we have $\min_{\theta: \Vert\theta-\tilde\theta\Vert_\infty\leq \Delta}\hat g_{0,M}(\theta) - \hat g_{0,M}(\tilde\theta)\geq -\epsilon_1(\eta, M) - c_{g_0}K\Delta$ and $\max_{\theta: \Vert\theta-\tilde\theta\Vert_\infty\leq \Delta}\hat g_{a,M}(\theta) - \hat g_{a,M}(\tilde\theta)\leq \epsilon_1(\eta, M) + c_{g_a}K\Delta$, and $\Theta(\alpha - \hat\epsilon_\alpha(\eta, M) + l_{g_0}K\Delta, \beta -\hat\epsilon_\beta(\eta, M) + l_{g_a}K\Delta)
\cap \hat C(\Delta)\subseteq \Theta_M(\alpha - \varepsilon_\alpha(\eta, M), \beta-\varepsilon_\beta(\eta, M) )$, then $\Theta(\alpha - \hat\epsilon_\alpha(\eta, M), \beta -\hat\epsilon_\beta(\eta, M))\subseteq \Theta_M(\alpha, \beta)$. For any $\hat\theta\in \Theta(\alpha - \hat\epsilon_\alpha(\eta, M), \beta -\hat\epsilon_\beta(\eta, M))$, with some abuse of notation we let $\tilde\theta(\hat\theta)$ be such that $\tilde\theta(\hat\theta)\in \Theta(\alpha - \hat\epsilon_\alpha(\eta, M)+ l_{g_0}K\Delta, \beta -\hat\epsilon_\beta(\eta, M)+ l_{g_a}K\Delta)\cap \hat C(\Delta)$ and $\Vert \hat\theta - \tilde\theta(\hat\theta)\Vert_\infty\leq \Delta$. Existence of $\tilde\theta(\hat\theta)$ is guaranteed by Lemma~\ref{lem:Delta-net}(ii). Then, by construction of $\tilde\theta(\hat\theta)$, we have that $\hat g_{0,M}(\hat\theta) - \hat g_{0,M}(\tilde\theta(\hat\theta))\geq -\epsilon_1(\eta, M) - c_{g_0}K\Delta$ and $\hat g_{a,M}(\hat\theta) - \hat g_{a,M}(\tilde\theta(\hat\theta))\leq \epsilon_1(\eta, M) + c_{g_a}K\Delta$. Additionally, since $\Theta(\alpha - \hat\epsilon_\alpha(\eta, M) + l_{g_0}K\Delta, \beta -\hat\epsilon_\beta(\eta, M) + l_{g_a}K\Delta)
\cap \hat C(\Delta)\subseteq \Theta_M(\alpha - \varepsilon_\alpha(\eta, M), \beta-\varepsilon_\beta(\eta, M) )$, we have that $\tilde\theta(\hat\theta)\in \Theta_M(\alpha - \varepsilon_\alpha(\eta, M), \beta-\varepsilon_\beta(\eta, M) )$. Thus $\hat g_{0,M}(\hat\theta) \geq \hat g_{0,M}(\tilde\theta(\hat\theta)) -\epsilon_1(\eta, M) - c_{g_0}K\Delta\geq 1-\alpha$, and $\hat g_{a,M}(\hat\theta)\leq \hat g_{a,M}(\tilde\theta(\hat\theta)) + \epsilon_1(\eta, M) + c_{g_a}K\Delta \leq \beta$. In other words, we have shown that $\hat\theta\in \Theta_M(\alpha, \beta)$. Based on the previous arguments, we have that
\begin{align*}
&\,\Pr(\Theta(\alpha - \hat\epsilon_\alpha(\eta, M), \beta -\hat\epsilon_\beta(\eta, M))\subseteq \Theta_M(\alpha, \beta),\\
&\quad\quad\quad\quad
\Theta(\alpha - \hat\epsilon_\alpha(\eta, M) + l_{g_0}K\Delta, \beta -\hat\epsilon_\beta(\eta, M) + l_{g_a}K\Delta)\cap \hat C(\Delta)\\
&\quad\quad\quad\quad\quad\quad\quad
\subseteq \Theta_M(\alpha - \varepsilon_\alpha(\eta, M), \beta-\varepsilon_\beta(\eta, M))
)\\
\geq & \Pr(\min_{\theta: \Vert\theta-\tilde\theta\Vert_\infty\leq \Delta}\hat g_{0,M}(\theta) - \hat g_{0,M}(\tilde\theta)\geq -\epsilon_1(\eta, M) - c_{g_0}K\Delta,\,\forall \tilde\theta\in\hat C(\Delta),\\
& 
\quad\quad\quad\quad\max_{\theta: \Vert\theta-\tilde\theta\Vert_\infty\leq \Delta}\hat g_{a,M}(\theta) - \hat g_{a,M}(\tilde\theta)\leq \epsilon_1(\eta, M) + c_{g_a}K\Delta,\, \forall \tilde\theta\in\hat C(\Delta),\\
&\quad\quad\quad\quad\quad\quad\quad
\Theta(\alpha - \hat\epsilon_\alpha(\eta, M) + l_{g_0}K\Delta, \beta-\hat\varepsilon_\beta(\eta, M) + l_{g_a}K\Delta)\cap \hat C(\Delta)\\
&\quad\quad\quad\quad\quad\quad\quad\quad\quad\quad
\subseteq \Theta_M(\alpha -\varepsilon_\alpha(\eta, M), \beta-\varepsilon_\beta(\eta, M)))\\
\geq &1-\Pr(\min_{\theta: \Vert\theta-\tilde\theta\Vert_\infty\leq \Delta}\hat g_{0,M}(\theta) - \hat g_{0,M}(\tilde\theta)< -\epsilon_1(\eta, M) - c_{g_0}K\Delta,\,\text{for some } \tilde\theta\in\hat C(\Delta))\\
& 
\quad\quad\quad\quad-\Pr(\max_{\theta: \Vert\theta-\tilde\theta\Vert_\infty\leq \Delta}\hat g_{a,M}(\theta) - \hat g_{a,M}(\tilde\theta)> \epsilon_1(\eta, M) + c_{g_a}K\Delta,\, \text{for some } \tilde\theta\in\hat C(\Delta))\\
&\quad\quad\quad\quad\quad\quad\quad-\Pr(\Theta(\alpha - \hat\epsilon_\alpha(\eta, M) + l_{g_0}K\Delta, \beta -\hat\epsilon_\beta(\eta, M) + l_{g_a}K\Delta)\cap \hat C(\Delta)\\
&\quad\quad\quad\quad\quad\quad\quad\quad\quad\quad
\nsubseteq \Theta_M(\alpha - \varepsilon_\alpha(\eta, M), \beta-\varepsilon_\beta(\eta, M))
)\\
\geq &1- 2\eta, 
\end{align*}
where the last inequality follows from Lemmas~\ref{lem:finite-feasibility} and \ref{lem:SAA-constraint-bound}(i)-(ii). 

Combining the above arguments, 
$$
\Pr(\Theta(\alpha - \hat\epsilon_\alpha(\eta, M), \beta -\hat\epsilon_\beta(\eta, M))\subseteq \Theta_M(\alpha, \beta)\subseteq \Theta(\alpha + \hat\epsilon_\alpha(\eta, M), \beta+ \hat\epsilon_\beta(\eta, M)))\geq 1-4\eta. 
$$

\sloppy (i) The argument follows by letting $\Delta = (\bar\theta - \underline{\theta})/\sqrt{M}$, and 
\begin{equation}\label{equ-appendix:epsilon}
    \epsilon(\eta, M) = 2\epsilon_1(\eta/4, M) + \max\{l_{g_0} + c_{g_0}, l_{g_a}+c_{g_a}\}K(\bar\theta - \underline{\theta})/\sqrt{M},
\end{equation}
where recall $\epsilon_1(\eta, M) = \sqrt{\log(4\ceil{\sqrt{M}}^K/\eta)/(2M)}$, $l_{g_0}, l_{g_a}$ are constants given in Lemma~\ref{lem:lipshitz-continuity}, and $c_{g_0}, c_{g_a}$ are constants given in Lemma~\ref{lem:SAA-constraint-bound}.

(ii) The argument follows from solving for $\eta$ by letting $\epsilon(\eta, M) = \epsilon$.

\end{proof}
\begin{proof}[Proof of Corollary~\ref{prop:prob-feasibility}]
    The result directly follows from Proposition~\ref{prop:feasible-set}, as we have $\hat\theta_M(\alpha, \beta)\in \Theta_M(\alpha, \beta)$.  
\end{proof}

\subsection{Proof of Proposition~\ref{prop:confidence-interval}}
We have that
\begin{equation}\label{proof-equ:ci-A0}
    \begin{split}
        &\Pr\big(\Theta_M(\alpha - \epsilon(\eta, M), \beta -\epsilon(\eta, M))\neq\emptyset, \text{ and }v^*(\alpha, \beta)\in\big[f\big(\hat\theta_M(\alpha + \epsilon(\eta, M), \beta+ \epsilon(\eta, M))\big) \\
    &\quad\quad\quad- \epsilon_f(\eta, M), f\big(\hat\theta_M(\alpha - \epsilon(\eta, M), \beta - \epsilon(\eta, M))\big)\big]\big)\\
    \geq &1-\Pr\big(\Theta_M(\alpha - \epsilon(\eta, M), \beta -\epsilon(\eta, M))\neq\emptyset, \text{ and }v^*(\alpha, \beta)
    < f\big(\hat\theta_M(\alpha + \epsilon(\eta, M), \beta+ \epsilon(\eta, M))\big)\\
    &\quad\quad\quad
    - \epsilon_f(\eta, M)\big) \\
    &\quad- \Pr(\Theta_M(\alpha - \epsilon(\eta, M), \beta -\epsilon(\eta, M))\neq\emptyset, \text{ and } 
    v^*(\alpha, \beta)> 
    f(\hat\theta_M(\alpha - \epsilon(\eta, M), \beta - \epsilon(\eta, M)))
    )\\
    &\quad
    -\Pr(\Theta_M(\alpha - \epsilon(\eta, M), \beta -\epsilon(\eta, M))=\emptyset)
    \end{split}
\end{equation}
First of all, we have that 
\begin{equation}\label{proof-equ:ci-A1}
\begin{split}
    &\Pr(\Theta_M(\alpha - \epsilon(\eta, M), \beta -\epsilon(\eta, M))\neq\emptyset, \text{ and } v^*(\alpha, \beta)> 
    f(\hat\theta_M(\alpha - \epsilon(\eta, M), \beta - \epsilon(\eta, M))))\\
    \leq &\Pr(\Theta_M(\alpha - \epsilon(\eta, M), \beta -\epsilon(\eta, M))\neq\emptyset, \text{ and }\Theta_M(\alpha - \epsilon(\eta, M), \beta -\epsilon(\eta, M))\nsubseteq \Theta(\alpha, \beta))\\
    \leq &\Pr(\Theta_M(\alpha - \epsilon(\eta, M), \beta -\epsilon(\eta, M))\nsubseteq \Theta(\alpha, \beta)))\\
    \leq &\eta,
    \end{split}
\end{equation}
where the last inequality follows according to Proposition~\ref{prop:feasible-set}. 

We now derive the term $\Pr(\Theta_M(\alpha - \epsilon(\eta, M), \beta -\epsilon(\eta, M))\neq\emptyset, \text{ and }v^*(\alpha, \beta)
    < f(\hat\theta_M(\alpha + \epsilon(\eta, M), \beta+ \epsilon(\eta, M)))
    - \epsilon_f(\eta, M))$. With some abuse of notation, let $\hat\theta_M(\alpha, \beta;\hat C(\Delta))\in\arg\min_{\theta\in \hat C(\Delta)}\Vert \theta - \hat\theta_M(\alpha, \beta)\Vert_1$. We have that
\begin{equation}\label{proof-equ:ci-A2}
    \begin{split}
&f(\hat\theta_M(\alpha + \epsilon(\eta, M), \beta+ \epsilon(\eta, M))) - v^*(\alpha, \beta) \\
=& f(\hat\theta_M(\alpha + \epsilon(\eta, M), \beta+ \epsilon(\eta, M)))  - f(\hat\theta_M(\alpha + \epsilon(\eta, M), \beta+ \epsilon(\eta, M); \hat C(\Delta))) \\
&+ f(\hat\theta_M(\alpha + \epsilon(\eta, M), \beta+ \epsilon(\eta, M); \hat C(\Delta)))
- \hat f_M(\hat\theta_M(\alpha + \epsilon(\eta, M), \beta+ \epsilon(\eta, M); \hat C(\Delta))) \\
&+ \hat f_M(\hat\theta_M(\alpha + \epsilon(\eta, M), \beta+ \epsilon(\eta, M); \hat C(\Delta)))-\hat f_M(\hat\theta_M(\alpha + \epsilon(\eta, M), \beta+ \epsilon(\eta, M))) \\
&+ \hat f_M(\hat\theta_M(\alpha + \epsilon(\eta, M), \beta+ \epsilon(\eta, M))) -\hat f_M(\theta^*(\alpha, \beta)) \\
&+ \hat f_M(\theta^*(\alpha, \beta))- v^*(\alpha, \beta).
    \end{split}
\end{equation}

\textbf{A1. Upper bound on $f(\hat\theta_M(\alpha + \epsilon(\eta, M), \beta+ \epsilon(\eta, M)))  - f(\hat\theta_M(\alpha + \epsilon(\eta, M), \beta+ \epsilon(\eta, M); \hat C(\Delta)))$.}

By definition of $\hat\theta_M(\alpha, \beta; \hat C(\Delta))$ and $\hat C(\Delta)$, we must have $\Vert\hat\theta_M(\alpha, \beta) - \hat\theta_M(\alpha, \beta; \hat C(\Delta))\Vert_1\leq K\Delta/2$. Then, according to Lemma~\ref{lem:lipshitz-continuity}, we have that $f(\hat\theta_M(\alpha + \epsilon(\eta, M), \beta+ \epsilon(\eta, M)))  - f(\hat\theta_M(\alpha + \epsilon(\eta, M), \beta+ \epsilon(\eta, M); \hat C(\Delta)))\leq l_fK\Delta/2$. Thus
\begin{equation}\label{proof-equ:ci-A2-1}
\begin{split}
   &\Pr(\Theta_M(\alpha - \epsilon(\eta, M), \beta -\epsilon(\eta, M))\neq\emptyset, \text{ and }f(\hat\theta_M(\alpha + \epsilon(\eta, M), \beta+ \epsilon(\eta, M)))  \\
   &\quad\quad\quad\quad\quad
   - f(\hat\theta_M(\alpha + \epsilon(\eta, M), \beta+ \epsilon(\eta, M); \hat C(\Delta)))\leq l_fK\Delta/2)\\
   =& \Pr(\Theta_M(\alpha - \epsilon(\eta, M), \beta -\epsilon(\eta, M))\neq\emptyset).
\end{split}
\end{equation}

\textbf{A2. Upper bound on $f(\hat\theta_M(\alpha + \epsilon(\eta, M), \beta+ \epsilon(\eta, M); \hat C(\Delta)))
- \hat f_M(\hat\theta_M(\alpha + \epsilon(\eta, M), \beta+ \epsilon(\eta, M); \hat C(\Delta)))$.} Let $n_{2:K} = \sum^K_{k=2}n_k$. Then 
\begin{equation}\label{proof-equ:ci-A2-2}
    \begin{split}
& \Pr(\Theta_M(\alpha - \epsilon(\eta, M), \beta -\epsilon(\eta, M))\neq\emptyset, \text{ and } f(\hat\theta_M(\alpha+ \epsilon(\eta, M), \beta+ \epsilon(\eta, M); \hat C(\Delta))) \\
&\quad\quad\quad
- \hat f_M(\hat\theta_M(\alpha + \epsilon(\eta, M), \beta+ \epsilon(\eta, M); \hat C(\Delta)))\leq n_{2:K}\epsilon_1(\eta, M))\\
\geq &\Pr(\Theta_M(\alpha - \epsilon(\eta, M), \beta -\epsilon(\eta, M))\neq\emptyset, \text{ and }f(\theta)
- \hat f_M(\theta)\leq n_{2:K}\epsilon_1(\eta, M), \forall \theta\in \hat C(\Delta))\\
\geq & 1-\eta - \Pr(\Theta_M(\alpha - \epsilon(\eta, M), \beta -\epsilon(\eta, M))=\emptyset), 
    \end{split}
\end{equation}
where the first inequality follows since $\hat\theta_M(\alpha + \epsilon(\eta, M), \beta+ \epsilon(\eta, M); \hat C(\Delta))\in \hat C(\Delta)$, and the second inequality follows from Hoeffding's inequality. 

\textbf{A3. Upper bound on $\hat f_M(\hat\theta_M(\alpha + \epsilon(\eta, M), \beta+ \epsilon(\eta, M); \hat C(\Delta)))-\hat f_M(\hat\theta_M(\alpha + \epsilon(\eta, M), \beta+ \epsilon(\eta, M)))$.}

By definition of $\hat\theta_M(\alpha + \epsilon(\eta, M), \beta+ \epsilon(\eta, M); \hat C(\Delta))$, we must have $\Vert \hat\theta_M(\alpha + \epsilon(\eta, M), \beta+ \epsilon(\eta, M); \hat C(\Delta)) - \hat\theta_M(\alpha + \epsilon(\eta, M), \beta+ \epsilon(\eta, M))\Vert_\infty\leq \Delta$. Then, according to Lemma~\ref{lem:SAA-constraint-bound}(iii), we have that
\begin{equation}\label{proof-equ:ci-A2-3}
    \begin{split}
    &\Pr(\Theta_M(\alpha - \epsilon(\eta, M), \beta -\epsilon(\eta, M))\neq\emptyset, \text{ and }\hat f_M(\hat\theta_M(\alpha+\epsilon(\eta, M), \beta+\epsilon(\eta, M); \hat C(\Delta))) \\
    &\quad
    - \hat f_M(\hat\theta_M(\alpha+\epsilon(\eta, M), \beta+\epsilon(\eta, M)))
    \leq c_fK(K-1)\Delta/2 + n_{2:K}\epsilon_1(4\eta/(K-1), M))\\
    \geq&
     1-\eta - \Pr(\Theta_M(\alpha - \epsilon(\eta, M), \beta -\epsilon(\eta, M))=\emptyset). 
    \end{split}
\end{equation}

\textbf{A4. Upper bound on $\hat f_M(\hat\theta_M(\alpha+\epsilon(\eta, M), \beta + \epsilon(\eta, M))) - \hat f_M(\theta^*(\alpha, \beta ))$.}
\begin{equation}\label{proof-equ:ci-A2-4}
    \begin{split}
        &\Pr(\hat f_M(\hat\theta_M(\alpha+\epsilon(\eta, M), \beta+\epsilon(\eta, M))) - \hat f_M(\theta^*(\alpha , \beta))\leq 0)\\
        \geq & \Pr(\Theta(\alpha, \beta)\subseteq \Theta_M(\alpha+\epsilon(\eta, M), \beta+\epsilon(\eta, M)))\\
        \geq & 1-\eta. 
    \end{split}
\end{equation}

\textbf{A5. Upper bound on $\hat f_M(\theta^*(\alpha, \beta))- v^*(\alpha, \beta)$.}

According to Hoeffding's inequality, 
\begin{equation}\label{proof-equ:ci-A2-5}
    \Pr(\hat f_M(\theta^*(\alpha, \beta)) - v^*(\alpha, \beta)\leq n_{2:K}\epsilon_1(4|\hat C(\Delta)|\eta, M))\geq 1-\eta. 
\end{equation}

Combining Eqs.~\eqref{proof-equ:ci-A2}-\eqref{proof-equ:ci-A2-5}, and applying the union bound, 
\begin{equation}\label{proof-equ:ci-bound-A}
\begin{split}
    &\Pr\big(\Theta_M(\alpha - \epsilon(\eta, M), \beta -\epsilon(\eta, M))\neq\emptyset, f(\hat\theta_M(\alpha+\epsilon(\eta, M), \beta+\epsilon(\eta, M))) - v^*(\alpha, \beta)\\
    &\quad
    > l_fK\Delta/2 + n_{2:K}\epsilon_1(\eta, M) + c_fK(K-1)\Delta/2 + n_{2:K}\epsilon_1(\eta/(K-1)) + n_{2:K}\epsilon_1(2|\hat C(\Delta)|\eta, M)\big)\\
    \leq&\,
     4\eta+ 3\Pr(\Theta_M(\alpha - \epsilon(\eta, M), \beta -\epsilon(\eta, M))=\emptyset). 
\end{split}
\end{equation}
Additionally, we let $\Delta = (\bar\theta - \underline{\theta})/\sqrt{M}$, and
\begin{equation}\label{equ-appendix:epsilon-f}
    \epsilon_f(\eta, M) = l_fK\Delta/2 + n_{2:K}\epsilon_1(\eta, M) + c_fK(K-1)\Delta/2 + n_{2:K}\epsilon_1(\eta/(K-1)) + n_{2:K}\epsilon_1(2|\hat C(\Delta)|\eta, M),
\end{equation}
where recall that $l_f$ is the constant given in Lemma~\ref{lem:lipshitz-continuity}, $c_f$ is the constant given in Lemma~\ref{lem:SAA-constraint-bound},  $\epsilon_1(\eta, M) = \sqrt{\log(4\ceil{\sqrt{M}}^K/\eta)/(2M)}$, and $\epsilon_1(2|\hat C(\Delta)|\eta, M) = \sqrt{\log(2/\eta)/(2M)}$. Then, 
combining Eqs.~\eqref{proof-equ:ci-A0}, \eqref{proof-equ:ci-A1}, and \eqref{proof-equ:ci-bound-A}, we have that
\begin{align*}
        &\Pr\big(\Theta_M(\alpha - \epsilon(\eta, M), \beta -\epsilon(\eta, M))\neq\emptyset, \text{ and }v^*(\alpha, \beta)\in\big[f\big(\hat\theta_M(\alpha + \epsilon(\eta, M), \beta+ \epsilon(\eta, M))\big) \\
    &\quad\quad\quad- \epsilon_f(\eta, M), f\big(\hat\theta_M(\alpha - \epsilon(\eta, M), \beta - \epsilon(\eta, M))\big)\big]\big)\\
    \geq &1-5\eta - 4\Pr(\Theta_M(\alpha - \epsilon(\eta, M), \beta -\epsilon(\eta, M))=\emptyset),
\end{align*}

Finally, by letting 
\begin{equation}\label{proof-equ:ci-M-bar}
    \bar M(\bar\epsilon, \eta) = \min\{M: \epsilon(\eta, M)\leq \bar\epsilon/2, \, 2\exp(-M\bar\epsilon^2/4)\leq \eta\},
\end{equation}
we have that $\Pr(\Theta_M(\alpha - \epsilon(\eta, M), \beta -\epsilon(\eta, M))=\emptyset)\leq \eta$, and the argument follows.

\subsection{Proof of Proposition~\ref{prop:objective-bound}}
We give the following Lemma stating that if MFCQ holds at any $\theta^*(\alpha, \beta)\in S(\alpha, \beta)$, $v^*(\alpha, \beta)$ is lipschitz continuous in a neighborhood of $(\alpha, \beta)$. Lemma~\ref{lem:optimal-value-lipshitz} directly follows from Lemma 6.2 of \citet{still2018lectures}. We thus omit its proof. 
\begin{lemma}\label{lem:optimal-value-lipshitz}[Lemma 6.2 of \citet{still2018lectures}]
    Suppose MFCQ holds at any $\theta^*(\alpha, \beta)\in S(\alpha, \beta)$. There exist constants $L_v\geq 0$ and $\epsilon_v > 0$ such that 
    $$
    |v^*(\alpha - \epsilon, \beta -\epsilon) - v^*(\alpha, \beta )|\leq L_v|\epsilon|, \, \forall \epsilon\in [-\epsilon_v, \epsilon_v].
    $$
\end{lemma}
% \begin{proof}[Proof of Lemma~\ref{lem:lipshitz-continuity}]
% Let $h_1(\theta) = 1-\alpha - g_0(\theta)$, $h_2(\theta) = g_a(\theta) - \beta$, $h_i(\theta) = \theta_{i-2} - \bar\theta$ for $i=3, \cdots, k+2$, $h_i(\theta) = -\theta_{i-2-k}$, for all $i = k+3, \cdots, 2k+2$. Then problem~\eqref{equ:main-1} can be equivalently written as:
% $$
% \min_{\theta\in\mathbb{R}^n}\{f(\theta): h_i(\theta)\leq 0, \,\forall i\in[2K+2]\}. 
% $$
% Then, according to Lemma 6.2 of \citet{still2018lectures}, since $f(\theta)$ and $h_i(\theta)$ for $i\in[2k+2]$ are continuously differentiable, 
% it suffices to show that MFCQ holds at any $\theta^*\in S(\alpha, \beta)$,  where MFCQ is stated as: there exists some $\xi$ such that $\nabla h_j(\theta^*)\xi<0$ for all $j\in J(\theta^*)$ where $J(\theta^*)$ is the set of active index set at $\theta^*$, i.e. $J(\theta^*) = \{j\in [2k+2]: h_j(\theta^*) = 0\}$. The MFCQ holds by assumption.

% \end{proof}

Let 
\begin{equation}\label{proof-equ:objective-M-bar}
    \bar M_f(\bar\epsilon, \eta) = \min\{M\geq 1: \epsilon(\eta, M) \leq \min\{\bar\epsilon, \epsilon_v\}, \epsilon(\eta, M)\leq \bar\epsilon/2, 2\exp(-M\bar\epsilon^2/4)\leq \eta\}.
\end{equation}
Then when $M\geq \bar M_f(\bar\epsilon, \eta)$, we have $\Theta(\alpha - \epsilon(\eta, M), \beta -\epsilon(\eta, M))\neq \emptyset$ and $\Pr(\Theta_M(\alpha, \beta)=\emptyset)\leq \eta$. 

\begin{proof}[Proof of Proposition~\ref{prop:objective-bound}]
We have: 
\begin{equation}\label{proof-equ:objective-bound}
    \begin{split}
|f(\hat\theta_M(\alpha, \beta)) - v^*(\alpha, \beta)|&\leq \max\{f(\hat\theta_M(\alpha, \beta)) - v^*(\alpha, \beta), v^*(\alpha, \beta) - f(\hat\theta_M(\alpha, \beta))\}. 
    \end{split}
\end{equation}

\textbf{A. Upper bound on $f(\hat\theta_M(\alpha, \beta)) - v^*(\alpha, \beta)$}. We have that when $M\geq \bar M_f(\bar\epsilon, \eta)$, $\Theta(\alpha - \epsilon(\eta, M), \beta -\epsilon(\eta, M))\neq \emptyset$, so $v^*(\alpha - \epsilon(\eta, M), \beta -\epsilon(\eta, M))<\infty$. Then, we have
\begin{equation}\label{proof-equ:objective-bound-A0}
\begin{split}
    &f(\hat\theta_M(\alpha, \beta)) - v^*(\alpha, \beta) \\
    = &f(\hat\theta_M(\alpha, \beta)) - v^*(\alpha - \epsilon(\eta, M), \beta -\epsilon(\eta, M))
    + v^*(\alpha - \epsilon(\eta, M), \beta -\epsilon(\eta, M))  - 
    v^*(\alpha, \beta). 
\end{split}
\end{equation}
Since $M\geq \bar M_f(\bar\epsilon, \eta)$, $\epsilon(\eta, M)\leq \epsilon_v$. Then, according to Lemma~\ref{lem:optimal-value-lipshitz}, there exists some constant $L_v$ such that
we have that 
\begin{equation}\label{proof-equ:objective-bound-A1}
    v^*(\alpha - \epsilon(\eta, M), \beta -\epsilon(\eta, M))  - 
    v^*(\alpha, \beta)\leq L_v\epsilon(\eta, M). 
\end{equation}

Now, we derive the upper bound on $f(\hat\theta_M(\alpha, \beta)) - v^*(\alpha - \epsilon(\eta, M), \beta -\epsilon(\eta, M))$. With some abuse of notation, let $\hat\theta_M(\alpha, \beta;\hat C(\Delta))\in\arg\min_{\theta\in \hat C(\Delta)}\Vert \theta - \hat\theta_M(\alpha, \beta)\Vert_1$. Then we have
\begin{equation}\label{proof-equ:objective-bound-A2}
    \begin{split}
&f(\hat\theta_M(\alpha, \beta)) - v^*(\alpha-\epsilon(\eta, M), \beta-\epsilon(\eta, M)) \\
=& f(\hat\theta_M(\alpha, \beta)) - f(\hat\theta_M(\alpha, \beta; \hat C(\Delta))) + f(\hat\theta_M(\alpha, \beta; \hat C(\Delta)))
- \hat f_M(\hat\theta_M(\alpha, \beta; \hat C(\Delta))) \\
&\quad+ \hat f_M(\hat\theta_M(\alpha, \beta; \hat C(\Delta)))
-\hat f_M(\hat\theta_M(\alpha, \beta)) + \hat f_M(\hat\theta_M(\alpha, \beta)) -\hat f_M(\theta^*(\alpha-\epsilon(\eta, M), \beta-\epsilon(\eta, M))) \\
&\quad+ \hat f_M(\theta^*(\alpha-\epsilon(\eta, M), \beta-\epsilon(\eta, M)))- v^*(\alpha-\epsilon(\eta, M), \beta-\epsilon(\eta, M)).
    \end{split}
\end{equation}

\textbf{A1. Upper bound on $f(\hat\theta_M(\alpha, \beta)) - f(\hat\theta_M(\alpha, \beta; \hat C(\Delta)))$.}

By definition of $\hat\theta_M(\alpha, \beta; \hat C(\Delta))$ and $\hat C(\Delta)$, we must have $\Vert\hat\theta_M(\alpha, \beta) - \hat\theta_M(\alpha, \beta; \hat C(\Delta))\Vert_1\leq K\Delta/2$. Then, according to Lemma~\ref{lem:lipshitz-continuity}, we have that $f(\hat\theta_M(\alpha, \beta)) - f(\hat\theta_M(\alpha, \beta; \hat C(\Delta)))\leq l_f\Vert \hat\theta_M(\alpha, \beta;\hat C(\Delta)) - \hat\theta_M(\alpha, \beta)\Vert_1\leq l_fK\Delta/2$. Then
\begin{equation}\label{proof-equ:objective-bound-A2-1}
\Pr(\Theta_M(\alpha, \beta)\neq\emptyset, \text{ and }f(\hat\theta_M(\alpha, \beta)) - f(\hat\theta_M(\alpha, \beta; \hat C(\Delta)))\leq l_fK\Delta/2) = \Pr(\Theta_M(\alpha, \beta)\neq\emptyset)
    .
\end{equation}

\textbf{A2. Upper bound on $f(\hat\theta_M(\alpha, \beta; \hat C(\Delta))) - \hat f_M(\hat\theta_M(\alpha, \beta; \hat C(\Delta)))$.} Let $n_{2:K} = \sum^K_{k=2}n_k$. Then
\begin{equation}\label{proof-equ:objective-bound-A2-2}
    \begin{split}
& \Pr(\Theta_M(\alpha, \beta)\neq\emptyset,\text{ and }f(\hat\theta_M(\alpha, \beta; \hat C(\Delta))) - \hat f_M(\hat\theta_M(\alpha, \beta; \hat C(\Delta)))\leq n_{2:K}\epsilon_1(4\eta, M))\\
\geq &\Pr(\Theta_M(\alpha, \beta)\neq\emptyset,\text{ and }f(\theta) - \hat f_M(\theta)\leq n_{2:K}\epsilon_1(4\eta, M), \forall \theta\in \hat C(\Delta))\\
\geq & 1-\eta - \Pr(\Theta_M(\alpha, \beta)=\emptyset), 
    \end{split}
\end{equation}
where the first inequality follows since $\hat\theta_M(\alpha, \beta; \hat C(\Delta))\in \hat C(\Delta)$, and the second inequality follows from Hoeffding's inequality and the union bound. 

\textbf{A3. Upper bound on $\hat f_M(\hat\theta_M(\alpha, \beta; \hat C(\Delta))) - \hat f_M(\hat\theta_M(\alpha, \beta))$.}

By definition of $\hat\theta_M(\alpha, \beta; \hat C(\Delta))$, we must have $\Vert \hat\theta_M(\alpha, \beta; \hat C(\Delta)) - \hat\theta_M(\alpha, \beta)\Vert_\infty\leq \Delta$. Then, according to Lemma~\ref{lem:SAA-constraint-bound}(iii), for $K\geq 2$ we have that
\begin{align}\label{proof-equ:objective-bound-A2-3}
    &\Pr(\Theta_M(\alpha, \beta)\neq\emptyset,\text{ and }\hat f_M(\hat\theta_M(\alpha, \beta; \hat C(\Delta))) - \hat f_M(\hat\theta_M(\alpha, \beta))\\
    &\quad\quad\quad
    \leq c_fK(K-1)\Delta/2 + n_{2:K}\epsilon_1(4\eta/(K-1), M))\geq 1-\eta - \Pr(\Theta_M(\alpha, \beta)=\emptyset). 
\end{align}

\textbf{A4. Upper bound on $\hat f_M(\hat\theta_M(\alpha, \beta)) - \hat f_M(\theta^*(\alpha - \epsilon(\eta, M), \beta -\epsilon(\eta, M)))$.}
\begin{equation}\label{proof-equ:objective-bound-A2-4}
    \begin{split}
        &\Pr(\hat f_M(\hat\theta_M(\alpha, \beta)) - \hat f_M(\theta^*(\alpha - \epsilon(\eta, M), \beta -\epsilon(\eta, M)))\leq 0)\\
        \geq & \Pr(\Theta(\alpha - \epsilon(\eta, M), \beta -\epsilon(\eta, M))\subseteq \Theta_M(\alpha, \beta))\\
        \geq & 1-\eta. 
    \end{split}
\end{equation}

\textbf{A5. Upper bound on $\hat f_M(\theta^*(\alpha - \epsilon(\eta, M), \beta - \epsilon(\eta, M)))- v^*(\alpha - \epsilon(\eta, M), \beta - \epsilon(\eta, M))$.}

According to Hoeffding's inequality, 
\begin{equation}\label{proof-equ:objective-bound-A2-5}
    \Pr(\hat f_M(\theta^*(\alpha- \epsilon(\eta, M), \beta- \epsilon(\eta, M))) - v^*(\alpha- \epsilon(\eta, M), \beta- \epsilon(\eta, M))\leq n_{2:K}\epsilon_1(4|\hat C(\Delta)|\eta, M))\geq 1-\eta. 
\end{equation}

Combining Eqs.~\eqref{proof-equ:objective-bound-A0}-\eqref{proof-equ:objective-bound-A2-5}, and applying the union bound, 
\begin{equation}\label{proof-equ:objective-bound-A}
\begin{split}
    &\Pr\big(\Theta_M(\alpha, \beta)\neq\emptyset,\text{ and }f(\hat\theta_M(\alpha, \beta)) - v^*(\alpha, \beta)> L_v\epsilon(\eta, M) + l_fK\Delta/2 + n_{2:K}\epsilon_1(\eta, M) \\
    &\quad\quad+ c_fK(K-1)\Delta/2 + n_{2:K}\epsilon_1(\eta/(K-1)) + n_{2:K}\epsilon_1(2|\hat C(\Delta)|\eta, M)\big)\\
    \leq & \Pr(\Theta_M(\alpha, \beta)\neq\emptyset,\text{ and }v^*(\alpha - \epsilon(\eta, M), \beta -\epsilon(\eta, M))  - 
    v^*(\alpha, \beta)> L_v\epsilon(\eta, M)) \\
    &\quad\quad+ \Pr(\Theta_M(\alpha, \beta)\neq\emptyset,\text{ and }f(\hat\theta_M(\alpha, \beta)) - f(\hat\theta_M(\alpha, \beta; \hat C(\Delta)))> l_fK\Delta/2) \\
    &\quad\quad + \Pr(\Theta_M(\alpha, \beta)\neq\emptyset,\text{ and }f(\hat\theta_M(\alpha, \beta; \hat C(\Delta))) - \hat f_M(\hat\theta_M(\alpha, \beta; \hat C(\Delta)))> n_{2:K}\epsilon_1(\eta, M))\\
    &\quad\quad + \Pr(\Theta_M(\alpha, \beta)\neq\emptyset,\text{ and }\hat f_M(\hat\theta_M(\alpha, \beta; \hat C(\Delta))) - \hat f_M(\hat\theta_M(\alpha, \beta))\\
    &\quad\quad\quad\quad\quad\quad\quad\quad
    > c_fK(K-1)\Delta/2 + n_{2:K}\epsilon_1(\eta/(K-1), M))\\
    &\quad\quad + \Pr(\hat f_M(\hat\theta_M(\alpha, \beta)) - \hat f_M(\theta^*(\alpha - \epsilon(\eta, M), \beta -\epsilon(\eta, M)))> 0)\\
    &\quad\quad + \Pr(\hat f_M(\theta^*(\alpha- \epsilon(\eta, M), \beta- \epsilon(\eta, M))) - v^*(\alpha- \epsilon(\eta, M), \beta- \epsilon(\eta, M))> n_{2:K}\epsilon_1(2|\hat C(\Delta)|\eta, M))\\
    \leq & 2\eta + 3\Pr(\Theta_M(\alpha, \beta)=\emptyset). 
\end{split}
\end{equation}

\textbf{B. Upper bound on $v^*(\alpha, \beta) - f(\hat\theta_M(\alpha, \beta))$}. Since $\Theta(\alpha + \epsilon(\eta, M), \beta+ \epsilon(\eta, M))\neq \emptyset$, we have
\begin{equation*}
\begin{split}
    &v^*(\alpha, \beta)- f(\hat\theta_M(\alpha, \beta)) \\
    = &v^*(\alpha, \beta) - v^*(\alpha + \epsilon(\eta, M), \beta+ \epsilon(\eta, M))
    + v^*(\alpha + \epsilon(\eta, M), \beta+ \epsilon(\eta, M))  - f(\hat\theta_M(\alpha, \beta)).
\end{split}
\end{equation*}

Since $M\geq \bar M_f(\bar\epsilon, \eta)$, we have that $\epsilon(\eta, M)\leq \epsilon_v$. Then, according to Lemma~\ref{lem:optimal-value-lipshitz}, 
\begin{equation*}\label{proof-equ:objective-bound-B1}
    v^*(\alpha, \beta) - v^*(\alpha + \epsilon(\eta, M), \beta+ \epsilon(\eta, M))
\leq L_v \epsilon(\eta, M).
\end{equation*}

According to Proposition~\ref{prop:feasible-set}, we have that 
\begin{equation*}\label{proof-equ:objective-bound-B2}
\begin{split}
&\Pr(\Theta_M(\alpha, \beta)\neq\emptyset, \text{ and }v^*(\alpha + \epsilon(\eta, M), \beta+ \epsilon(\eta, M))> f(\hat\theta_M(\alpha , \beta)))\\
\leq &\Pr(\Theta_M(\alpha, \beta)\neq\emptyset, \text{ and }\Theta_M(\alpha, \beta)\nsubseteq \Theta(\alpha + \epsilon(\eta, M), \beta+ \epsilon(\eta, M)))\\
\leq & \eta. 
\end{split}
\end{equation*}
It thus follows that
\begin{equation}\label{proof-equ:objective-bound-B}
    \Pr(\Theta_M(\alpha, \beta)\neq\emptyset, \text{ and }v^*(\alpha, \beta) - f(\hat\theta_M(\alpha, \beta))> L_v\epsilon(\eta, M))\leq \eta. 
\end{equation}

Let $\tilde \epsilon_f(\eta, M) = L_v\epsilon(\eta, M) + l_fK\Delta/2 + n_{2:K}\epsilon_1(\eta, M)+c_fK(K-1)\Delta/2 + n_{2:K}\epsilon_1(\eta/(K-1)) + n_{2:K}\epsilon_1(2|\hat C(\Delta)|\eta, M)$.
Combining Eqs.~\eqref{proof-equ:objective-bound-A} and \eqref{proof-equ:objective-bound-B}, we have that
\begin{equation*}
    \begin{split}
&\Pr(\Theta_M(\alpha, \beta)\neq\emptyset, \text{ and }|f(\hat\theta_M(\alpha, \beta)) - v^*(\alpha, \beta)|\leq  \tilde \epsilon_f(\eta, M)) \\
\geq &1-\Pr(\Theta_M(\alpha, \beta)\neq\emptyset, \text{ and }f(\hat\theta_M(\alpha, \beta)) - v^*(\alpha, \beta> \tilde \epsilon_f(\eta, M))
\\
 &\quad-\Pr(\Theta_M(\alpha, \beta)\neq\emptyset, \text{ and }-f(\hat\theta_M(\alpha, \beta)) + v^*(\alpha, \beta)\geq \tilde \epsilon_f(\eta, M))\\
&\quad
- \Pr(\Theta_M(\alpha, \beta)=\emptyset)\\
> & 1-3\eta - 4\Pr(\Theta_M(\alpha, \beta)=\emptyset). 
    \end{split}
\end{equation*}
Additionally, for $M\geq \bar M_f(\bar\epsilon, \eta)$, we have that $\Pr(\Theta_M(\alpha, \beta)=\emptyset)\leq \eta$. 

\sloppy Last, the argument follows by letting $\Delta = (\bar\theta - \underline{\theta})/\sqrt{M}$, and then $\tilde\epsilon_f(\eta, M) = \sqrt{c_{\mathrm{v,1}}(K, n) + c_{\mathrm{v},2}(K, n)\log(1/\eta) + c_{\mathrm{v, 3}}(K, n)\log(M)/M}$.

\end{proof}

\subsection{Proof of Lemma~\ref{lem:MFCQ}}
\begin{proof}[Proof of Lemma~\ref{lem:MFCQ}]
To prove this lemma, we first reformulate the original optimization problem \eqref{equ:main} into the standard form of an optimization problem \eqref{eq:MFCQ_form}. Specifically, we let $h_0(\theta) = f(\theta)$, $h_1(\theta) = 1-\alpha - g_0(\theta)$, $h_2(\theta) = g_a(\theta) - \beta$, $h_{i}(\theta) = \underline{\theta} - \theta_{i-2}$ for $3\leq i\leq 2+K$, and $h_{i}(\theta) = \theta_{i-2-K} - \bar\theta$ for $3+K\leq i\leq 2+2K$. 

We notice that the set $\{\theta\in\mathbb{R}^K: \theta_k = \underline{\theta}\text{ or }\theta_k = \bar\theta \text{ for some }k\in[K]\}$ has a zero Lebesgue measure. Thus $\{\theta\in \mathbb{R}^K: i\in J_0(\theta) \text{ for some }3\leq i\leq 2K+2\}$ has a zero Lebesgue measure. 
It thus suffices to show that $\{\theta\in\mathbb{R}^K: \nabla g_0(\theta)\text{ and }\nabla g_a(\theta) \text{ are linearly dependent}\}$ has a zero Lebesgue measure. 

We first consider the one-sided $z$ tests (either one-sample or two-sample), and show that there exists some $\theta\in\mathbb{R}^{K}$ such that $\nabla g_0(\theta)$ and $\nabla g_a(\theta)$ are linearly independent. We notice that under $H_0$, $S = (S_k)^K_{k=1}$ has a multivariate normal distribution with mean $0$ and covariance $\Sigma$, whose $(k,k')$-th element is $\Sigma_{k, k'} = \sqrt{\min\{n_{1:k}, n_{1:k'}\}/\max\{n_{1:k}, n_{1:k'}\}}$. Let $p_{0,k}(\theta_k)$ and $p_{a,k}(\theta_k)$ be the density functions of the marginal distributions of $S_k$ under $H_0$ and $H_a$. Let $S_{-k}$ and $\theta_{-k}$ be the vectors of $S$ and $\theta$ without the $k$-th element. Then we have for all $k\in [K]$, 
\begin{align*}
    \partial g_0(\theta)/\partial \theta_k = p_{0,k}(\theta_k)\Pr(S_{-k}\leq \theta_{-k}|S_k = \theta_k, H_0),\,\nabla g_0(\theta) = (\partial g_0(\theta)/\partial \theta_k)^K_{k=1},\\
    \partial g_a(\theta)/\partial \theta_k = p_{a,k}(\theta_k)\Pr(S_{-k}\leq \theta_{-k}|S_k = \theta_k, H_a),\,
    \nabla g_a(\theta) = (\partial g_a(\theta)/\partial \theta_k)^K_{k=1}.
\end{align*}
We first show that there exists some $x\in \mathbb{R}$ such that $\nabla g_0(\theta(x))$ and $\nabla g_a(\theta(x))$ are linearly independent for $\theta(x) = (x, \cdots, x)$. We let $l_0(\theta) = (\partial g_0(\theta)/\partial \theta_1, \partial g_0(\theta)/\partial \theta_2)^\top$ and $l_a(\theta) = (\partial g_a(\theta)/\partial \theta_1, \partial g_a(\theta)/\partial \theta_2)^\top$. It is then sufficient to show that $l_0(\theta(x))$ and $l_a(\theta(x))$ are linearly independent. Since $S_1$ and $S_2$ have the same marginal distributions under $H_0$, $p_{0,1}(x) = p_{0,2}(x)$, and thus
$$
\lim_{x\to\infty}\frac{\partial g_0(\theta(x))/\partial \theta_1}{\partial g_0(\theta(x))/\partial \theta_2} = \lim_{x\to\infty}\frac{\Pr(S_{-1}\leq \theta_{-1}(x)|S_1 = x, H_0)}{\Pr(S_{-2}\leq \theta_{-2}(x)|S_2 = x, H_0)} = 1,
$$
where the second equality follows since for $S\sim N(\mu, \Sigma)$ where $\mu = (\mu_k)^K_{k=1}$, $\Sigma_{k,k}= 1$ and $\Sigma_{k,k'}<1$ for $k\neq k'$, $S_{-1}|S_1 = x$ ($S_{-2}|S_2 = x$) is normally distributed with mean $\mu_{-1} + (\Sigma_{2,1}, \cdots, \Sigma_{K,1})^\top (x-\mu_1)< \mu_{-1} + x - \mu_1$ ($\mu_{-2} + (\Sigma_{1,2}, \Sigma_{3,2}, \cdots, \Sigma_{K,2})^\top (x-\mu_2)< \mu_{-2} + x - \mu_2$), and thus $\lim_{x\to\infty}\Pr(S_{-1}\leq \theta_{-1}(x)|S_1 = x, H_0) = 1$ ($\lim_{x\to\infty}\Pr(S_{-2}\leq \theta_{-2}(x)|S_2 = x, H_0) = 1$). 
In addition, 
$$
\frac{p_{a,1}(x)}{p_{a,2}(x)} = \exp(- (x - \sqrt{n_1}\delta)^2/(2\sigma^2) + (x - \sqrt{n_1+n_2}\delta)^2/(2\sigma^2))\neq 1, 
$$
and
$$
\lim_{x\to\infty}\frac{\Pr(S_{-1}\leq \theta_{-1}(x)|S_1 = x, H_a)}{\Pr(S_{-2}\leq \theta_{-2}(x)|S_2 = x, H_a)} = 1, 
$$
and thus 
$$
\lim_{x\to\infty}\frac{\partial g_a(\theta(x))/\partial \theta_1}{\partial g_a(\theta(x))/\partial \theta_2} = \lim_{x\to\infty}\frac{p_{a,1}(x)}{p_{a,2}(x)} \neq 1.
$$
Since both $\frac{\partial g_0(\theta(x))/\partial \theta_1}{\partial g_0(\theta(x))/\partial \theta_2}$ and  $\frac{\partial g_a(\theta(x))/\partial \theta_1}{\partial g_a(\theta(x))/\partial \theta_2}$ are continuous in $x$, it follows that there exists some $x\in\mathbb{R}$ such that $\frac{\partial g_0(\theta(x))/\partial \theta_1}{\partial g_0(\theta(x))/\partial \theta_2}\neq \frac{\partial g_a(\theta(x))/\partial \theta_1}{\partial g_a(\theta(x))/\partial \theta_2}$. So $\det((l_0(\theta(x)), l_a(\theta(x))))\neq 0$.

In the meanwhile, the determinant $\det((l_0(\theta), l_a(\theta)))$ is a real analytic function. In fact, 
\begin{equation*}
    \begin{split}
&\det((l_0(\theta), l_a(\theta))) \\
= &p_{0,1}(\theta_1)\Pr(S_{-1}\leq \theta_{-1}|S_1 = \theta_1, H_0)p_{a,2}(\theta_2)\Pr(S_{-2}\leq \theta_{-2}|S_2 = \theta_2, H_a) \\
&- p_{0,2}(\theta_2)\Pr(S_{-2}\leq \theta_{-2}|S_2 = \theta_2, H_0)p_{a,1}(\theta_1)\Pr(S_{-1}\leq \theta_{-1}|S_1 = \theta_1, H_a), 
    \end{split}
\end{equation*}
and $p_{0,k}(\theta_k), p_{a,k}(\theta_k), \Pr(S_{-k}\leq\theta_{-k}|S_k = \theta_k, H_0)$ and $\Pr(S_{-k}\leq\theta_{-k}|S_k = \theta_k, H_a)$ are all real analytic functions (since each $S_k$ has a normal distribution and $S_{-k}|S_k = \theta_k$ is multivariate normal under both $H_0$ and $H_a$). 

By previous argument, there exists some $\hat\theta$ such that $\det(l_0(\hat\theta), l_a(\hat\theta))\neq 0$, and that $\det(l_0(\theta), l_a(\theta))$ is a real analytic function. 
Then, according to Proposition~0 of \citet{mityagin2015zero}, which states that the zero set of a real analytic function that is not identically zero on $\mathbb{R}^K$ has a zero measure,
the set $\{\theta\in\mathbb{R}^K: \det((l_0(\theta), l_a(\theta))) = 0\}$ has a zero measure. As $\det((l_0(\theta), l_a(\theta))) = 0$ is a necessary condition for $\nabla g_0(\theta)$ to be linearly dependent with $\nabla g_a(\theta)$, we have that $\{\theta\in\mathbb{R}^K: \text{MFCQ fails at }\theta\}$ has a zero Lebesgue measure. 

Now, consider the two sided tests. For convenience, we let $S_k = |\tilde S_k|$ for each $k\in[K]$ where $\tilde S = (\tilde S_k)^K_{k=1}$ has multivariate normal distribution with mean $0$ and covariance $\Sigma$ under $H_0$, and mean $\mu = (\mu_k)^K_{k=1}$ and covariance $\Sigma$ under $H_a$. Then we have, 
\begin{align*}
    &\partial g_0(\theta)/\partial \theta_k = p_0(\theta_k)\Pr(S_{-k}\leq \theta_{-k}|\tilde S_k = \theta_k, H_0) + p_0(-\theta_k)\Pr(S_{-k}\leq \theta_{-k})|\tilde S_k = -\theta_k, H_0),
    \\
    &\partial g_a(\theta)/\partial \theta_k = p_{a,k}(\theta_k)\Pr(S_{-k}\leq \theta_{-k}|\tilde S_k = \theta_k, H_a) 
    + p_{a,k}(-\theta_k)\Pr(S_{-k}\leq \theta_{-k}|\tilde S_k = -\theta_k, H_a).
\end{align*}
Additionally,
\begin{align*}
&\lim_{x\to\infty}\frac{\partial g_0(\theta(x))/\partial \theta_1}{\partial g_0(\theta(x))/\partial \theta_2} =1, \text{ and }\lim_{x\to\infty}\frac{\partial g_a(\theta(x))/\partial \theta_1}{\partial g_a(\theta(x))/\partial \theta_2}
\neq 1.
\end{align*}
Thus following previous arguments, we can similarly conclude $\{\theta\in\mathbb{R}^K: \text{MFCQ fails at }\theta\}$ has a zero Lebesgue measure for two-sided tests as well. 

\end{proof}

\subsection{More on MFCQ Condition for Proposition~\ref{prop:objective-bound}}\label{appendix-sec:MFCQ-example}
To gain more intuition on the MFCQ condition for Proposition~\ref{prop:objective-bound}, 
in Example~\ref{example:MFCQ} below, we construct an original problem and its SAA surrogate, for which the MFCQ fails to hold at the optimal solution of the original problem, and that the optimal value of the SAA problem does not converge to the optimal value for the original problem.

\begin{example}\label{example:MFCQ}
    We consider the following problem:
    \begin{equation}\label{equ:example-true}
        \begin{split}
            \min_{\theta\in\mathbb{R}}&\,\{\theta: \Pr(X\leq |\theta|) \geq p_1,\Pr(Y\leq |\theta - 1/4|) \leq p_2\}.
        \end{split}
    \end{equation}
    and the SAA counterpart: 
    \begin{equation}\label{equ:example-SAA}
        \begin{split}
            \min_{\theta\in\mathbb{R}}&\,\{\theta: 1/M\sum^M_{m=1}\I(X_m\leq |\theta|) \geq p_1, 1/M\sum^M_{m=1}\I(Y_m\leq |\theta - 1/4|) \leq p_2\}.
        \end{split}
    \end{equation}
    where $X$ and $Y$ are independently distributed Uniform random variables on $[0,1]$. Let $(X_m, Y_m), m\in[M]$ be i.i.d. samples drawn from the distribution of $(X, Y)$. 
    Let $p_1 = 1/2$ and $p_2 = 3/4$. 
    Then the optimal value of \eqref{equ:example-true} is $v^* = -1/2$. We have that MFCQ fails to hold at $\theta^* = -1/2$, the optimal solution of \eqref{equ:example-true}. To see this, let $h_1(\theta) = -\Pr(X\leq |\theta|)$ and $h_2(\theta) = \Pr(Y\leq |\theta-1/4|)$. Then the set of active constraints at $\theta^*$,  $J_0(\theta^*) = \{1, 2\}$. Additionally, $(dh_1(\theta^*)/d\theta, dh_2(\theta^*)/d\theta) = (1, -1)$. Then it is straightforward to see that there does not exist $\xi\in\mathbb{R}$ such that $\xi dh_1(\theta^*)/d\theta < 0$ and $\xi dh_2(\theta^*)/d\theta< 0$.

    Let $\hat\theta_M$ be any optimal solution of \eqref{equ:example-SAA} and let $f(\theta) = \theta$. Then, for all $0< \epsilon\leq 1/4$, 
    $$
    \lim\sup_{M\to\infty}\Pr(f(\hat\theta_M)\in [-\epsilon + v^*, \epsilon + v^*])\leq 1/2.
    $$
    To see this, let $X_{(1)}\leq \cdots\leq X_{(M)}$ and $Y_{(1)}\leq \cdots\leq Y_{(M)}$ be the order statistics of $\{X_i\}^M_{i=1}$ and $\{Y_i\}^M_{i=1}$, respectively. Then we have 
    $$
    \hat\theta = \begin{cases}
        X_{(\ceil{M/2})}, &\text{if }X_{(\ceil{M/2})}> Y_{(\floor{3M/4} + 1)} - 1/4\text{ and }X_{(\ceil{M/2})}\leq Y_{(\floor{3M/4} + 1)} + 1/4,\\
        1/4 - Y_{(\floor{3M/4} + 1)},&\text{if }X_{(\ceil{M/2})}\leq Y_{(\floor{3M/4} + 1)} - 1/4.\\
    \end{cases}
    $$
    As $X_{(\ceil{M/2})}\geq 0$, we have that $\Pr(\hat\theta_M\in [-\epsilon + v^*, \epsilon + v^*]) = \Pr(1/4 - Y_{(\floor{3M/4} + 1)} \in [-\epsilon + v^*, \epsilon + v^*])\text{ and } X_{(\ceil{M/2})}\leq Y_{(\floor{3M/4} + 1)} - 1/4)\leq \Pr(X_{(\ceil{M/2})}\leq Y_{(\floor{3M/4} + 1)} - 1/4)$, where we note that $X_{(\ceil{M/2})}$ and $Y_{(\floor{3M/4} + 1)}$ are independent (for each $M$) with $X_{(\ceil{M/2})}\sim \text{Beta}(\ceil{M/2}, M+1 - \ceil{M/2})$ and $Y_{(\floor{3M/4} + 1)}\sim \text{Beta}(\floor{3M/4} + 1, M-\floor{3M/4})$. According to Central Limit Theorem for Sample Quantiles (Theorem 1 of \citet{walker1968note}), $\sqrt{M}(X_{(\ceil{M/2})} - 1/2)\overset{d}{\Rightarrow} \tilde X$ and $\sqrt{M}(Y_{(\floor{3M/4} + 1)}-3/4)\overset{d}{\Rightarrow}\tilde Y$ as $M\to\infty$, where $\tilde X\sim \mathcal{N}(0, 1/4)$ and $\tilde Y\sim \mathcal{N}(0, 3/16)$. 
    Additionally, according to Example 3.2 of \citet{billingsley2013convergence}, $\tilde X$ and $\tilde Y$ are independent and $(\sqrt{M}(X_{(\ceil{M/2})}-1/2), \sqrt{M}(Y_{(\floor{3M/4} + 1)}-3/4))\overset{d}{\Rightarrow}(\tilde X, \tilde Y)$. We thus have that 
    $\Pr(X_{(\ceil{M/2})}\leq Y_{(\floor{3M/4} + 1)} - 1/4)\to 1/2$ as $M\to\infty$ and $\lim\sup_{M\to\infty}\Pr(f(\hat\theta_M)\in [-\epsilon + v^*, \epsilon + v^*])\leq 1/2$. 
\end{example}

\subsection{Convergence of SAA Optimal Solution}\label{appendix-sec:solution-convergence}
\begin{proposition}\label{prop:solution-bound}
Suppose MFCQ holds at any $\theta^*(\alpha, \beta)\in S(\alpha, \beta)$ for problem~\eqref{equ:main}. Then, $\sup_{\theta\in S_M(\alpha, \beta)}\dist(\theta, S(\alpha, \beta))\to 0$ with probability $1$ as $M\to\infty$. 

% (ii)
% For all $M\geq \bar M_f(\bar\epsilon, \eta)$ and any $\epsilon>0$,
% $$
% \Pr\big(\inf_{\theta\in S_M(\alpha, \beta)}\dist(\theta, S(\alpha, \beta))\geq \epsilon\big)\leq c_{\s, 1}(K, n)M^{c_{\s, 2}(K,n)}\exp(-c_{\s, 3}(K, n)M\Gamma(\epsilon)), 
% $$
% where $c_{\s, 1}(K, n), c_{\s, 2}(K, n)$ and $c_{\s, 3}(K, n) > 0$ are positive constants that depend on $K$ and $n$ and not depend on $M$. $\Gamma(\epsilon)\geq 0$ is an increasing function in $\epsilon$, with $\Gamma(\epsilon)> 0$ for all $\epsilon > 0$.

\end{proposition}
In particular, Proposition~\ref{prop:solution-bound} says that for \textit{any} SAA optimal solution $\hat\theta_M(\alpha, \beta)$, it converges to some $\theta^*(\alpha, \beta)\in S(\alpha, \beta)$ with probability $1$ as $M\to\infty$. 
% Proposition~\ref{prop:solution-bound}(ii) gives the convergence rate on the shortest distance between solution sets $S_M(\alpha, \beta)$ and $S(\alpha, \beta)$. 

To prove Proposition~\ref{prop:solution-bound}, we need the following auxiliary result. 
\begin{proposition}[Theorem 5.5 of \citet{shapiro2021lectures}]\label{prop:appendix-solution-convergence}
    Suppose that: (a) The set $S(\alpha, \beta)$ of the optimal solutions of the original problem is nonempty and is contained in a compact set; (b) The function $f(\theta)$ is finite valued and continuous on $C$; (c) $\hat f_M(\theta)$ converges to $f(\theta)$ with probability $1$ as $M\to\infty$ uniformly in $\theta\in C$; (d) With probability $1$ for $M$ large enough the set $S_M(\alpha, \beta)$ is nonempty and $S_M(\alpha, \beta)\subseteq C$; (e) For some $\theta^*\in S(\alpha, \beta)$ there exists a sequence $\theta_M\in \Theta_M(\alpha, \beta)$ such that $\theta_M\to \theta^*$ as $M\to\infty$ with probability $1$. Then $\sup_{\theta\in S_M(\alpha, \beta)}\dist(\theta, S(\alpha, \beta))\to 0$ with probability $1$ as $M\to\infty$. 
\end{proposition} 
% \begin{lemma}[Lemma EC.3 of \citet{kannan2025data}]\label{lem:appendix-solution-bound}
%     Let $W\in\mathbb{R}^K$ be a nonempty and compact set and $h: W\to \mathbb{R}$ be a lower semi-continuous function. Define $W^* = \arg\min_{w\in W}h(w)$. Suppose there exists $\epsilon > 0$ and $\bar w\in W$ such that $\dist(\bar w, W^*)\geq \epsilon$. Then there exists increasing $\Gamma(\epsilon) > 0$ whenever $\epsilon > 0$ such that $h(\bar w) \geq \min_{w\in W}h(w) + \Gamma(\epsilon)$. 
% \end{lemma}

\begin{proof}[Proof of Proposition~\ref{prop:solution-bound}]
It suffices to verify that conditions (a)-(e) of Proposition~\ref{prop:appendix-solution-convergence} hold. It is straightforward to see that (a)-(b) hold. (c) holds according to Theorem 7.48 of \citet{shapiro2021lectures}. (d) holds since we can always find $\hat\theta\in \Theta(\alpha, \beta)$ such that $g_0(\hat\theta) > 1-\alpha$ and $g_a(\hat\theta) < \beta$, and since $\hat g_{0,M}(\hat\theta)$ and $\hat g_{a,M}(\hat\theta)$ converge to $g_0(\hat\theta)$ and $g_a(\hat\theta)$ with probability $1$ respectively, we have $\hat\theta\in \Theta_M(\alpha, \beta)$ with probability $1$ for $M$ large enough. 
Now it remains to verify (e) holds. Since MFCQ holds at all $\theta^*(\alpha, \beta)\in S(\alpha, \beta)$, there exists some $\theta^*(\alpha, \beta)\in S(\alpha, \beta)$ such that for all sufficiently small $\epsilon_M> 0$, there is $\theta_M\in \Theta(\alpha, \beta)$ with $\Vert \theta_M - \theta^*(\alpha, \beta)\Vert< \epsilon_M$ such that $\underline{\theta}\leq\theta_M\leq\bar\theta$, $g_0(\theta_M)> 1-\alpha$ and $g_a(\theta_M)< \beta$, which follows by the definition of MFCQ. Additionally, according to Theorem 7.48 of \citet{shapiro2021lectures}, $\hat g_{0,M}(\theta)$ and $\hat g_{a, M}(\theta)$ converge to $g_0(\theta)$ and $g_a(\theta)$ uniformly on $\theta\in C$ with probability $1$, respectively. Consequently $\hat g_{0,M}(\theta_M)> 1-\alpha$ and $\hat g_{a,M}(\theta_M)< \beta$ almost surely as $M\to\infty$. It follows that $\theta_M\in\Theta_M(\alpha, \beta)$ almost surely as $M\to \infty$. Letting $\epsilon_M\to 0$, we also have that $\theta_M\to \theta^*$, and thus (e) holds.

\end{proof}

\section{Additional Materials for Section~\ref{sec:t-tests}}\label{appendix:t-test}
For proofs in this section, we define some new notations. Recall $f(\theta; \delta) = n_1 + \sum^K_{k=2}n_k\Pr(T_i < \theta_i, \forall i\in[k-1]|H_a)$ is the objective of \eqref{equ:main-unknown-sigma}. Let $g_a(\theta; \delta) = \Pr(T_k< \theta_k, \,\forall k\in[K]|H_a)$ be the LHS of the second constraint of \eqref{equ:main-unknown-sigma}. 
As the LHS of first constraint of \eqref{equ:main-unknown-sigma} does not depend on $\delta$, we let $g_0(\theta) = \Pr(T_k\leq \theta_k, \,\forall k\in[K]|H_0)$ and $\hat g_{0,M}(\theta) = (1/M)\sum^M_{m=1}\I(T^m_k\leq \theta_k, \,\forall k\in[K])$ with some abuse of notation. 
Correspondingly, we let $f(\theta; \hat\delta_{n_0}) = n_1 + \sum^K_{k=2}n_k\Pr(T_i < \theta_i, \forall i\in[k-1]|H_a = \mu_a, \sigma = \hat\sigma_0)$ and $g_a(\theta; \hat\delta_{n_0}) = \Pr(T_k< \theta_k, \,\forall k\in[K]|H_a = \mu_a, \sigma = \hat\sigma_0)$. Finally, we let $\hat f_M(\theta; \hat\delta_{n_0}) = n_1 + (1/M)\sum^M_{m=1}\sum^K_{k=2}n_k\I(T^m_i(\mu_a/\hat\sigma_0) < \theta_i, \forall i\in[k-1])$ and $\hat g_{a, M}(\theta; \hat\delta_{n_0}) = (1/M)\sum^M_{m=1}\I(T^m_k(\mu_a/\hat\sigma_0)<\theta_k, \,\forall k\in[K])$. Additionally, we continue to let $\hat C(\Delta) = \{\underline{\theta}, \underline{\theta} + \Delta, \underline{\theta} + 2\Delta, \cdots, \bar\theta\}^K$ be a discretization of $C$, and $\epsilon_1(\eta, M) = \sqrt{\log(4|\hat C(\Delta)|/\eta)/(2M)}$. 

\subsection{Proof of Proposition~\ref{prop:convergence-t}}\label{appendix:t-test-effect size}

Proposition~\ref{prop:convergence-t} can be proved similarly as Propositions~\ref{prop:feasible-set} and \ref{prop:objective-bound}. Specifically, 
we notice that an analog of Lemma~\ref{lem:finite-feasibility} holds regardless of the distribution of the summary statistics. Analogs of Lemmas~\ref{lem:lipshitz-continuity} and \ref{lem:SAA-constraint-bound} hold because the density function of a $t$-distribution is upper bounded (see Lemma~\ref{lem:lipshitz-continuity-t}). An analog of Lemma~\ref{lem:Delta-net} holds as it follows from Analogs of Lemmas~\ref{lem:lipshitz-continuity} and \ref{lem:SAA-constraint-bound}.  Lemma~\ref{lem:optimal-value-lipshitz} holds since MFCQ holds at any optimal solution of \eqref{equ:main-unknown-sigma-1}. Then Proposition~\ref{prop:convergence-t} can be proved based on the preceding lemmas.

\subsection{Proof of Proposition~\ref{prop:pilot-convergence}}\label{sec-appendix:t-test-pilot}

Lemmas~\ref{lem:SAA-feasibility-pilot}, \ref{lem:feasible-set-pilot} and \ref{lem:pilot-feasible-set-2} are auxiliary to  prove Proposition~\ref{prop:pilot-convergence}.

\begin{lemma}\label{lem:SAA-feasibility-pilot}
    Suppose there exists some $0< \bar\epsilon_\p < 4|\delta|c^g_\delta$ such that  $\Theta(\alpha - \bar\epsilon_\p, \beta -\bar\epsilon_\p; \delta)\neq \emptyset$, where $c^g_\delta$ is a constant given in Lemma~\ref{lem:pilot-sample-convergence}. Then $\Pr(\Theta_M(\alpha, \beta; \hat\delta_{n_0})\neq \emptyset)\geq 1-2\exp(-M\bar\epsilon_\p^2/2) - 2\exp(-(n_0 - 1)(\bar\epsilon_\p/(8|\delta|c^g_\delta))^2)$.
\end{lemma}
\begin{proof}[Proof of Lemma~\ref{lem:SAA-feasibility-pilot}]
     As $\Theta(\alpha - \bar\epsilon_\p, \beta -\bar\epsilon_\p; \delta)\neq\emptyset$, we let $\hat\theta$ be arbitrary such that
    $\hat\theta\in \Theta(\alpha - \bar\epsilon_\p, \beta -\bar\epsilon_\p; \delta)$. Then,
    \begin{align*}
    &\Pr\big(\Theta_M(\alpha, \beta; \hat\delta_{n_0})\neq\emptyset\big)\\
    \geq &\Pr\big(\hat\theta\in \Theta_M(\alpha, \beta; \hat\delta_{n_0})\big)\\
    = & \Pr\big(\hat g_{0,M}(\hat\theta)\geq 1-\alpha \text{ and }\hat g_{a, M}(\hat\theta; \hat\delta_{n_0})\leq \beta\big)\\
    \geq & \Pr\big(\hat g_{0,M}(\hat\theta)\geq g_0(\hat\theta)-\bar\epsilon_\p, \hat g_{a,M}(\hat\theta; \hat\delta_{n_0})\leq g_a(\hat\theta; \hat\delta_{n_0})+\bar\epsilon_\p/2, \\
    &\quad\quad
    \text{ and }
    g_a(\hat\theta; \hat\delta_{n_0}) - g_a(\hat\theta; \delta)\leq \bar\epsilon_\p/2
    \big)\\
    \geq & 1- \Pr\big(\hat g_{0,M}(\hat\theta)\leq g_0(\hat\theta)-\bar\epsilon_\p\big) - \Pr\big(\hat g_{a, M}(\hat\theta; \hat\delta_{n_0})\geq g_a(\hat\theta; \hat\delta_{n_0})+\bar\epsilon_\p/2\big)\\
    &\quad\quad - \Pr\big(g_a(\hat\theta; \hat\delta_{n_0}) - g_a(\hat\theta; \delta)\geq \bar\epsilon_\p/2\big)\\
    \geq & 1-2\exp(-M\bar\epsilon_\p^2/2) - 2\exp(-(n_0 - 1)(\bar\epsilon_\p/(8|\delta|c^g_\delta))^2), 
    \end{align*}
    where the last inequality follows from Lemma~\ref{lem:pilot-sample-convergence} and that
    \begin{align*}
    &\Pr\big(|\hat g_{a, M}(\theta; \hat\delta_{n_0}) - g_a(\theta; \hat\delta_{n_0})|\leq \epsilon|\{X_{-i}\}^{n_0}_{i=1}\big)\\
     = &\Pr\big(|\hat g_{a, M}(\theta; \hat\delta_{n_0}) - \E[\hat g_{a, M}(\theta; \hat\delta_{n_0})|X_{-1}, \cdots, X_{-n_0}]|\leq \epsilon|\{X_{-i}\}^{n_0}_{i=1}\big)\\
     \geq &1 - 2\exp(-2M\epsilon^2), \text{ almost surely}, 
    \end{align*}
    and so by law of iterated expectation, $\Pr\big(|\hat g_{a, M}(\theta; \hat\delta_{n_0}) - g_a(\theta; \hat\delta_{n_0})|\leq \epsilon\big)\geq 1-2\exp(-2M\epsilon^2)$.
\end{proof}

\begin{lemma}\label{lem:feasible-set-pilot}
\begin{align*}
    \Pr\big(\Theta(\alpha - \epsilon_\p(\eta, M), \beta - \epsilon_\p(\eta, M); \hat\delta_{n_0})\subseteq \Theta_M(\alpha, \beta; \hat\delta_{n_0})
    \subseteq \Theta(\alpha + \epsilon_\p(\eta, M), \beta+ \epsilon_\p(\eta, M); \hat\delta_{n_0})\big)\\
    \geq 1-\eta, 
\end{align*}
where $\epsilon_\p(\eta, M) = \sqrt{\big(C^\p_1(K)\log(1/\eta) + C^\p_2(K)\log(M) + C^\p_3(K)\big)/M}$, and $C^\p_1(K), C^\p_2(K)$ and $C^\p_3(K)$ are some constants that depend on $K$ and not $M$. 

% (ii) For all $\epsilon > 0$, we have the inequality below holds almost surely,
% \begin{align*}
%     \Pr(\Theta(\alpha - \epsilon, \beta - \epsilon; \hat\delta_{n_0})\subseteq \Theta_M(\alpha, \beta; \hat\delta_{n_0})\subseteq \Theta(\alpha + \epsilon, \beta+ \epsilon; \hat\delta_{n_0})|X_{-1}, \cdots, X_{-n_0})\\
%     \geq 1-\tilde C^\p_1(K)M^{\tilde C^\p_2(K)}\exp{(-\tilde C^\p_3(K)M\epsilon^2)}. 
% \end{align*}
% where $\tilde C^\p_1(K), \tilde C^\p_2(K)$ and $\tilde C^\p_3(K)$ are some constants that depend on $K$ and not $M$. 

\end{lemma}
We leave the proof of Lemma~\ref{lem:feasible-set-pilot} to Appendix~\ref{sec-appendix:feasible-set-pilot}. 

\begin{lemma}\label{lem:pilot-feasible-set-2}
    With probability at least $1-\eta$, $\Theta(\alpha, \beta - \epsilon(\eta, n_0); \delta)\subseteq \Theta(\alpha, \beta; \hat\delta_{n_0})\subseteq \Theta(\alpha, \beta + \epsilon(\eta, n_0); \delta)$ where $\epsilon(\eta, n_0) = \sqrt{\big(C_1(K)\log(1/\eta) + C_2(K)\log(n_0) + C_3(K)\big)/n_0}$, where $C_1(K), C_2(K)$ and $C_3(K)$ are some constants that depend on $K$ and not $n_0$ nor $M$; 
    
%     Alternatively, for any $\epsilon > 0$, we have $\Pr(\Theta(\alpha - \epsilon, \beta - \epsilon; \delta)\subseteq \Theta_M(\alpha, \beta; \hat\delta_{n_0})\subseteq \Theta(\alpha + \epsilon, \beta+ \epsilon); \delta)\geq 1-\tilde C_1(K)n_0^{\tilde C_2(K)}\exp{(-\tilde C_3(K)n_0\epsilon^2)}$ 
% where $\tilde C_1(K), \tilde C_2(K)$ and $\tilde C_3(K)$ are some constants that depend on $K$ and not $M$. 

\end{lemma}
We leave the proof of Lemma~\ref{lem:pilot-feasible-set-2} to Appendix~\ref{sec-appendix:pilot-feasible-set-2}.

\begin{proof}[Proof of Proposition~\ref{prop:pilot-convergence}]

(i)
follows from Lemmas~\ref{lem:feasible-set-pilot} and \ref{lem:pilot-feasible-set-2}.

(ii) can be similarly proved as Proposition~\ref{prop:objective-bound}. In particular, according to Lemma~\ref{lem:optimal-value-lipshitz}, there exist constants $L^\p_v$ and $\epsilon_\p> 0$ such that $|v^*(\alpha, \beta; \delta) - v^*(\alpha - \epsilon, \beta - \epsilon; \delta)|\leq L^\p_v\epsilon$ for all $\epsilon \in[-\epsilon_\p, \epsilon_\p]$. Then, let $\bar n_\p\geq 64\log(4/\eta) + 1$ and $\bar M_\p$ be such that for all $n_0\geq \bar n_\p$ and $M\geq \bar M_\p$, $\epsilon_\p(\eta, M, n_0) \leq \min\{\bar\epsilon_\p, \epsilon_\p\}$ and $\Pr(\Theta_M(\alpha, \beta; \hat\delta_{n_0})\neq \emptyset)\geq 1-\eta$. 

% \begin{equation}\label{proof-equ:objective-M-bar-pilot}
%     \bar M_\p(\bar\epsilon_\p, \eta) = \min\{M\geq 1: \epsilon_\p(\eta, M, n_0) \leq \min\{\bar\epsilon_\p, \epsilon_\p\}, \epsilon_\p(\eta, M, n_0)\leq \bar\epsilon_\p/2, 2\exp(-M\bar\epsilon_\p^2/4)\leq \eta\}.
% \end{equation}
We have that: 
\begin{equation}\label{proof-equ:objective-bound}
    \begin{split}
|f(\hat\theta_M(\alpha, \beta; \hat\delta_{n_0}); \delta) - v^*(\alpha, \beta; \delta)|\\
\leq \max\{f(\hat\theta_M(\alpha, \beta;\hat\delta_{n_0}); \delta) - v^*(\alpha, \beta; \delta), v^*(\alpha, \beta; \delta) - f(\hat\theta_M(\alpha, \beta; \hat\delta_{n_0}); \delta)\}.
    \end{split}
\end{equation}

\textbf{A. Upper bound on $f(\hat\theta_M(\alpha, \beta; \hat\delta_{n_0});\delta) - v^*(\alpha, \beta;\delta)$}. We have that when $M\geq \bar M_\p$ and $n_0\geq \bar n_\p$, $\Theta(\alpha - \epsilon_\p(\eta, M, n_0), \beta -\epsilon_\p(\eta, M, n_0);\delta)\neq \emptyset$, so $v^*(\alpha - \epsilon_\p(\eta, M, n_0), \beta -\epsilon_\p(\eta, M, n_0);\delta)<\infty$. Then, we have
\begin{equation}\label{proof-equ:objective-bound-A-pilot}
\begin{split}
    &f(\hat\theta_M(\alpha, \beta;\hat\delta_{n_0});\delta) - v^*(\alpha, \beta;\delta) \\
    = &f(\hat\theta_M(\alpha, \beta;\hat\delta_{n_0});\delta) - v^*(\alpha - \epsilon_\p(\eta, M, n_0), \beta -\epsilon_\p(\eta, M, n_0); \delta)\\
    &\quad\quad+ v^*(\alpha - \epsilon_\p(\eta, M, n_0), \beta -\epsilon_\p(\eta, M, n_0);\delta)  - 
    v^*(\alpha, \beta;\delta). 
\end{split}
\end{equation}
Since $M\geq \bar M_\p$ and $n_0\geq \bar n_\p$, $\epsilon_\p(\eta, M, n_0)\leq \epsilon_\p$. Then, according to Lemma~\ref{lem:optimal-value-lipshitz},
we have that 
\begin{equation}\label{proof-equ:objective-bound-A1-pilot}
    v^*(\alpha - \epsilon_\p(\eta, M, n_0), \beta -\epsilon_\p(\eta, M, n_0);\delta)  - 
    v^*(\alpha, \beta;\delta)\leq L^\p_v\epsilon_\p(\eta, M, n_0). 
\end{equation}

Now, we derive the upper bound on $f(\hat\theta_M(\alpha, \beta;\hat\delta_{n_0});\delta) - v^*(\alpha - \epsilon_\p(\eta, M, n_0), \beta -\epsilon_\p(\eta, M, n_0);\delta)$. With some abuse of notation, let $\hat\theta_M(\alpha, \beta;\hat\delta_{n_0},\hat C(\Delta))\in\arg\min_{\theta\in \hat C(\Delta)}\Vert \theta - \hat\theta_M(\alpha, \beta;\hat\delta_{n_0})\Vert_1$. Then we have
\begin{align*}\label{proof-equ:objective-bound-A2}
&f(\hat\theta_M(\alpha, \beta;\hat\delta_{n_0});\delta) - v^*(\alpha-\epsilon_\p(\eta, M, n_0), \beta-\epsilon_\p(\eta, M, n_0);\delta) \\
=& \underbrace{f(\hat\theta_M(\alpha, \beta;\hat\delta_{n_0});\delta) -f(\hat\theta_M(\alpha, \beta;\hat\delta_{n_0});\hat\delta_{n_0})}_{A_1}\\
&
+\underbrace{f(\hat\theta_M(\alpha, \beta;\hat\delta_{n_0});\hat\delta_{n_0})
- f(\hat\theta_M(\alpha, \beta; \hat\delta_{n_0},\hat C(\Delta)); \hat\delta_{n_0})}_{A_2} \\
&
+ \underbrace{f(\hat\theta_M(\alpha, \beta; \hat\delta_{n_0},\hat C(\Delta));\hat\delta_{n_0})
- \hat f_M(\hat\theta_M(\alpha, \beta; \hat\delta_{n_0},\hat C(\Delta));\hat\delta_{n_0})}_{A_3} \\
&
+ \underbrace{\hat f_M(\hat\theta_M(\alpha, \beta; \hat\delta_{n_0},\hat C(\Delta));\hat\delta_{n_0})-\hat f_M(\hat\theta_M(\alpha, \beta;\hat\delta_{n_0});\hat\delta_{n_0})}_{A_4} \\
&
+ \underbrace{\hat f_M(\hat\theta_M(\alpha, \beta;\hat\delta_{n_0}); \hat\delta_{n_0}) -\hat f_M(\theta^*(\alpha-\epsilon_\p(\eta, M, n_0), \beta-\epsilon_\p(\eta, M, n_0); \delta);\hat\delta_{n_0})}_{A_5} \\
&+ \underbrace{\hat f_M(\theta^*(\alpha-\epsilon_\p(\eta, M, n_0), \beta-\epsilon_\p(\eta, M, n_0); \delta);\hat\delta_{n_0})
- f(\theta^*(\alpha-\epsilon_\p(\eta, M, n_0), \beta-\epsilon_\p(\eta, M, n_0); \delta);\hat\delta_{n_0})}_{A_6}\\
&+
\underbrace{f(\theta^*(\alpha-\epsilon_\p(\eta, M, n_0), \beta-\epsilon_\p(\eta, M, n_0); \delta);\hat\delta_{n_0}) - v^*(\alpha-\epsilon_\p(\eta, M, n_0), \beta-\epsilon_\p(\eta, M, n_0);\delta)}_{A_7}.
\end{align*}

For $A_1$, 
according to Lemma~\ref{lem:pilot-sample-convergence}, $\Pr(\Theta_M(\alpha, \beta; \hat\delta_{n_0})\neq \emptyset, \text{ and }f(\hat\theta_M(\alpha, \beta;\hat\delta_{n_0});\delta) -f(\hat\theta_M(\alpha, \beta;\hat\delta_{n_0});\hat\delta_{n_0})\leq c^f_\delta|\delta|
(2\sqrt{\log(4/\eta)/(n_0 - 1)} + 2\log(4/\eta)/(n_0 - 1)))\geq 1-\eta - \Pr(\Theta_M(\alpha, \beta; \hat\delta_{n_0})= \emptyset)$ for some constant $c^f_\delta$.

For $A_2$, since $\Vert\hat\theta_M(\alpha, \beta;\hat\delta_{n_0}) - \hat\theta_M(\alpha, \beta;\hat\delta_{n_0},\hat C(\Delta))\Vert_1\leq K\Delta/2$ almost surely, according to Lemma~\ref{lem:lipshitz-continuity-t}, we have that $f(\hat\theta_M(\alpha, \beta;\hat\delta_{n_0}); \hat\delta_{n_0}) - f(\hat\theta_M(\alpha, \beta; \hat\delta_{n_0},\hat C(\Delta)); \hat\delta_{n_0})\leq l^\p_f\Vert \hat\theta_M(\alpha, \beta;\hat\delta_{n_0},\hat C(\Delta)) - \hat\theta_M(\alpha, \beta;\hat\delta_{n_0})\Vert_1\leq l^\p_fK\Delta/2$ almost surely when $\Theta_M(\alpha, \beta; \hat\delta_{n_0})\neq \emptyset$.

For $A_3$, we have that $\Pr(\Theta_M(\alpha, \beta; \hat\delta_{n_0})\neq \emptyset, \text{ and }f(\hat\theta_M(\alpha, \beta; \hat\delta_{n_0},\hat C(\Delta));\hat\delta_{n_0})
- \hat f_M(\hat\theta_M(\alpha, \beta; \hat\delta_{n_0},\hat C(\Delta));\hat\delta_{n_0})\leq \epsilon_1(4\eta, M)|\{X_{-i}\}^{n_0}_{i=1})
\geq \Pr(\Theta_M(\alpha, \beta; \hat\delta_{n_0})\neq \emptyset, \text{ and }f(\theta; \hat\delta_{n_0}) - \hat f_M(\theta; \hat\delta_{n_0})\leq \epsilon_1(4\eta, M), \forall \theta\in \hat C(\Delta)|\{X_{-i}\}^{n_0}_{i=1})
\geq 1-\eta - \Pr(\Theta_M(\alpha, \beta; \hat\delta_{n_0})= \emptyset|\{X_{-i}\}^{n_0}_{i=1})$ almost surely, where the first inequality follows since $\hat\theta_M(\alpha, \beta; \hat\delta_{n_0}, \hat C(\Delta))\in \hat C(\Delta)$, and the second inequality follows from Hoeffding's inequality. 

For $A_4$, according to Lemma~\ref{lem:SAA-constraint-bound-pilot}(iii), we have that $\Pr(\Theta_M(\alpha, \beta; \hat\delta_{n_0})\neq \emptyset, \text{and }\hat f_M(\hat\theta_M(\alpha, \beta; \hat\delta_{n_0},\hat C(\Delta));\hat\delta_{n_0})-\hat f_M(\hat\theta_M(\alpha, \beta;\hat\delta_{n_0});\hat\delta_{n_0})\leq c^\p_fK(K-1)\Delta/2 + n_{2:K}\epsilon_1(4\eta/(K-1), M)|\{X_{-i}\}^{n_0}_{i=1})\geq 1-\eta - \Pr(\Theta_M(\alpha, \beta; \hat\delta_{n_0})=\emptyset|\{X_{-i}\}^{n_0}_{i=1})$ almost surely. 

For $A_5$, $\Pr(\hat f_M(\hat\theta_M(\alpha, \beta;\hat\delta_{n_0}); \hat\delta_{n_0}) -\hat f_M(\theta^*(\alpha-\epsilon_\p(\eta, M, n_0), \beta-\epsilon_\p(\eta, M, n_0); \delta);\hat\delta_{n_0})\leq 0)\geq \Pr(\Theta(\alpha - \epsilon_\p(\eta, M, n_0), \beta -\epsilon_\p(\eta, M, n_0); \delta)\subseteq \Theta_M(\alpha, \beta; \hat\delta_{n_0}))\geq 1-\eta$ according to Proposition~\ref{prop:pilot-convergence}(i).

For $A_6$, $\Pr(\hat f_M(\theta^*(\alpha-\epsilon_\p(\eta, M, n_0), \beta-\epsilon_\p(\eta, M, n_0); \delta);\hat\delta_{n_0})- f(\theta^*(\alpha-\epsilon_\p(\eta, M, n_0), \beta-\epsilon_\p(\eta, M, n_0); \delta);\hat\delta_{n_0})\leq \epsilon_1(4|\hat C(\Delta)|\eta, M)|\{X_{-i}\}^{n_0}_{i=1})\geq 1-\eta$ almost surely according to Hoeffding's inequality.

For $A_7$, $\Pr(f(\theta^*(\alpha-\epsilon_\p(\eta, M, n_0), \beta-\epsilon_\p(\eta, M, n_0); \delta);\hat\delta_{n_0}) - v^*(\alpha-\epsilon_\p(\eta, M, n_0), \beta-\epsilon_\p(\eta, M, n_0);\delta)\leq c^f_\delta|\delta|
(2\sqrt{\log(4/\eta)/(n_0 - 1)} + 2\log(4/\eta)/(n_0 - 1))\geq 1-\eta$ for some constant $c^f_\delta$ according to Lemma~\ref{lem:pilot-sample-convergence}.

Let $\epsilon(\eta, M, n_0) = L^\p_v\epsilon_\p(\eta, M, n_0) + c^f_\delta|\delta|
(2\sqrt{\log(4/\eta)/(n_0 - 1)} + 2\log(4/\eta)/(n_0 - 1)) + l^\p_fK\Delta/2 + \epsilon_1(\eta, M) + c^\p_fK(K-1)\Delta/2 + n_{2:K}\epsilon_1(4\eta/(K-1), M) + \epsilon_1(4|\hat C(\Delta)|\eta, M) + c^f_\delta|\delta|
(2\sqrt{\log(2/\eta)/(n_0 - 1)} + 2\log(2/\eta)/(n_0 - 1)$. 
Combining the above arguments, and applying the union bound, 
\begin{equation}\label{proof-equ:objective-bound-A-pilot}
\begin{split}
    &\Pr\big(\Theta_M(\alpha, \beta; \hat\delta_{n_0})\neq \emptyset, \text{ and }f(\hat\theta_M(\alpha, \beta;\hat\delta_{n_0});\delta) - v^*(\alpha, \beta;\delta) \leq \epsilon(\eta, M, n_0)\big)\\
    \geq &1 - 6\eta - 5\Pr(\Theta_M(\alpha, \beta; \hat\delta_{n_0})= \emptyset). 
\end{split}
\end{equation}

\textbf{B. Upper bound on $v^*(\alpha, \beta; \delta) - f(\hat\theta_M(\alpha, \beta; \hat\delta_{n_0}); \delta)$}. Since $\Theta(\alpha + \epsilon_\p(\eta, M, n_0), \beta+ \epsilon_\p(\eta, M, n_0); \delta)\neq \emptyset$,
we have that
\begin{equation*}
\begin{split}
    &v^*(\alpha, \beta; \delta)- f(\hat\theta_M(\alpha, \beta;\hat\delta_{n_0}); \delta) \\
    = &v^*(\alpha, \beta; \delta) - v^*(\alpha + \epsilon_\p(\eta, M, n_0), \beta+ \epsilon_\p(\eta, M, n_0); \delta)\\
    &\quad\quad
    + v^*(\alpha + \epsilon_\p(\eta, M, n_0), \beta+ \epsilon_\p(\eta, M, n_0); \delta)  - f(\hat\theta_M(\alpha, \beta;\hat\delta_{n_0}); \delta).
\end{split}
\end{equation*}

Since $M\geq \bar M_\p$ and $n_0\geq \bar n_\p$, we have that $\epsilon_\p(\eta, M, n_0)\leq \epsilon_\p$. Then, according to Lemma~\ref{lem:optimal-value-lipshitz}, 
\begin{equation*}\label{proof-equ:objective-bound-B1}
    v^*(\alpha, \beta; \delta) - v^*(\alpha + \epsilon_\p(\eta, M, n_0), \beta+ \epsilon_\p(\eta, M, n_0); \delta)
\leq L^\p_v \epsilon_\p(\eta, M, n_0).
\end{equation*}

Additionally, according to Proposition~\ref{prop:pilot-convergence}(i), we have that
\begin{equation*}\label{proof-equ:objective-bound-B2}
\begin{split}
    v^*(\alpha + \epsilon_\p(\eta, M, n_0), \beta+ \epsilon_\p(\eta, M, n_0); \delta) =& \min_{\theta\in \Theta(\alpha + \epsilon_\p(\eta, M, n_0), \beta+ \epsilon_\p(\eta, M, n_0); \delta)}f(\theta; \delta)\\
    \leq & \min_{\theta\in \Theta_M(\alpha, \beta; \hat\delta_{n_0})}f(\theta; \delta)\leq f(\hat\theta_M(\alpha , \beta; \hat\delta_{n_0}); \delta).
\end{split}
\end{equation*}
It thus follows that
\begin{equation}\label{proof-equ:objective-bound-B-pilot}
\begin{split}
    &\Pr(\Theta_M(\alpha, \beta; \hat\delta_{n_0})\neq \emptyset, \text{ and }
    v^*(\alpha, \beta) - f(\hat\theta_M(\alpha, \beta;\hat\delta_{n_0}))\leq L^\p_v\epsilon_\p(\eta, M, n_0))\\
    \geq &\Pr(\Theta_M(\alpha, \beta; \hat\delta_{n_0})\neq \emptyset, \text{ and }
    \Theta_M(\alpha, \beta; \hat\delta_{n_0})\subseteq
    \Theta(\alpha + \epsilon_\p(\eta, M, n_0), \beta+ \epsilon_\p(\eta, M, n_0);\delta)
    )\\
    \geq & 1-\eta - \Pr(\Theta_M(\alpha, \beta; \hat\delta_{n_0})= \emptyset). 
\end{split}
\end{equation}

Combining Eqs.~\eqref{proof-equ:objective-bound-A-pilot} and \eqref{proof-equ:objective-bound-B-pilot}, we have that
\begin{align*}
\Pr(\Theta_M(\alpha, \beta; \hat\delta_{n_0})\neq \emptyset, \text{ and }|f(\hat\theta_M(\alpha, \beta;\hat\delta_{n_0})) - v^*(\alpha, \beta)|\leq \epsilon(\eta, M, n_0)\\
\geq 1-7\eta-6\Pr(\Theta_M(\alpha, \beta; \hat\delta_{n_0})= \emptyset). 
\end{align*}
Finally, the argument follows by noting that for $M\geq \bar M_\p$ and $n_0\geq \bar n_\p$, $\Pr(\Theta_M(\alpha, \beta; \hat\delta_{n_0})= \emptyset)\leq \eta$.

\end{proof}

\subsection{Proof of Lemma~\ref{lem:feasible-set-pilot}}\label{sec-appendix:feasible-set-pilot}

The proof for Lemma~\ref{lem:feasible-set-pilot} proceeds similarly as that for Proposition~\ref{prop:feasible-set}. Lemma~\ref{lem:finite-feasibility-pilot} below is an analog of Lemma~\ref{lem:finite-feasibility}
\begin{lemma}\label{lem:finite-feasibility-pilot}
    $\Pr(\Theta(\alpha - \epsilon_1(\eta, M), \beta -\epsilon_1(\eta, M); \hat\delta_{n_0})\cap \hat C(\Delta)\subseteq \Theta_M(\alpha, \beta; \hat\delta_{n_0})\cap \hat C(\Delta)\subseteq \Theta(\alpha + \epsilon_1(\eta, M), \beta+ \epsilon_1(\eta, M); \hat\delta_{n_0})\cap \hat C(\Delta)|\{X_{-i}\}^{n_0}_{i=1})\geq 1-\eta$, almost surely. 
\end{lemma}
\begin{proof}[Proof of Lemma~\ref{lem:finite-feasibility-pilot}]
    The statement follows from Lemma~\ref{lem:finite-feasibility} and the fact that $\E[\hat g_{a, M}(\theta; \hat\delta_{n_0})|X_{-1}, \cdots, X_{-n_0}] = g_a(\theta; \hat\delta_{n_0})$,
    % and $\E[\hat f_M(\theta; \hat\delta_{n_0})|X_{-1}, \cdots, X_{-n_0}] = f(\theta; \hat\delta_{n_0})$
    which follows by the independence between $\{X_{-i}\}^{n_0}_{i=1}$ and $\{S^m\}^M_{m=1}$. As a result, we can  apply Hoeffding's inequality in the proof of Lemma~\ref{lem:finite-feasibility}. Specifically, in the proof of Lemma~\ref{lem:finite-feasibility}, we use the concentration inequality that for all $\epsilon > 0$, 
    \begin{align*}
    &\Pr\big(|\hat g_{a, M}(\theta; \hat\delta_{n_0}) - g_a(\theta; \hat\delta_{n_0})|\leq \epsilon|\{X_{-i}\}^{n_0}_{i=1}\big)\\
     = &\Pr\big(|\hat g_{a, M}(\theta; \hat\delta_{n_0}) - \E[\hat g_{a, M}(\theta; \hat\delta_{n_0})|X_{-1}, \cdots, X_{-n_0}]|\leq \epsilon|\{X_{-i}\}^{n_0}_{i=1}\big)\\
     \geq &1 - 2\exp(-2M\epsilon^2), \text{ almost surely}. 
    \end{align*}
\end{proof}

Lemma~\ref{lem:lipshitz-continuity-t} is an analog of Lemma~\ref{lem:lipshitz-continuity}. 
\begin{lemma}\label{lem:lipshitz-continuity-t}
For all $\delta\in\mathbb{R}$, there exist constants $l^\p_f, l^\p_{g_0}$ and $l^\p_{g_a}$ that do not depend on $\delta$ such that 
    \begin{align*}
        |f(\hat\theta; \delta) - f(\tilde\theta; \delta)|\leq l^\p_f\Vert\hat\theta - \tilde\theta\Vert_1, \forall \hat\theta, \tilde\theta\in C,\\
        |g_0(\hat\theta) - g_0(\tilde\theta)|\leq l^\p_{g_0}\Vert\hat\theta - \tilde\theta\Vert_1, \forall \hat\theta, \tilde\theta\in C,\\
        |g_{a}(\hat\theta; \delta) - g_a(\tilde\theta; \delta)|\leq l^\p_{g_a}\Vert\hat\theta - \tilde\theta\Vert_1, \forall \hat\theta, \tilde\theta\in C. 
    \end{align*}
\end{lemma}
\begin{proof}[Proof of Lemma~\ref{lem:lipshitz-continuity-t}]
    Let $p_{T(n;\delta)}$ be the density function of the $T(n;\delta)$ distribution. We first show that for all $t\in\mathbb{R}$ and all $\delta\in\mathbb{R}$, we have $p_{T(n;\delta)}(t)\leq 1/\sqrt{2\pi}$. Let $Z\sim\mathcal{N}(0,1)$ and $V\sim\chi^2_n$ be independent, and define $T=(Z+\delta)/\sqrt{V/n}$. Then $T\sim T(n;\delta)$. Let $\varphi(x)=(2\pi)^{-1/2}e^{-x^2/2}$ be the standard normal density. Let $p_{T|V=v}(t)$ be the conditional density function of $T$ conditional on $V = v > 0$, then $p_{T|V=v}(t)=\sqrt{v/n}\;\varphi(t\sqrt{v/n}-\delta)$ for all $t\in \mathbb{R}$. Let $p_V$ denote the density of $V$. Since the integrand is nonnegative, Tonelli's theorem yields $p_{T(n;\delta)}(t)
=\int_{0}^{\infty} p_{T\mid V=v}(t)\,p_V(v)\,dv
=\mathbb{E}[\sqrt{V/n}\;\varphi(t\sqrt{V/n}-\delta)]$.
Using $\sup_{x\in\mathbb{R}}\varphi(x)=\varphi(0)=(2\pi)^{-1/2}$, we obtain $p_{T(n;\delta)}(t)
\leq \E[\sqrt{V/n}]/\sqrt{2\pi}$. 
Since $x\mapsto \sqrt{x}$ is concave on $(0,\infty)$, Jensen's inequality gives $\E[\sqrt{V/n}]
\le \sqrt{\mathbb{E}[V/n]}
= \sqrt{\E[V]/n}$. 
For $V\sim\chi^2_n$, $\mathbb{E}[V]=n$, hence $\mathbb{E}[\sqrt{V/n}]\le 1$, and therefore $p_{T(n;\delta)}(t)\le 1/\sqrt{2\pi}$ for all $t\in\mathbb{R}$. 

Let $T_{-k}$ and $\theta_{-k}$ be the vectors of $T$ and $\theta$ without the $k$-th element. Now, the rest of the argument follows by noticing that $\partial g_0(\theta)/\partial \theta_k = p_{T(n_{1:k}-1, 0)}(\theta_k)\Pr(T_{-k}\leq \theta_{-k}|T_k = \theta_k, H_0) \leq p_{T(n_{1:k}-1, 0)}(\theta_k),\partial g_a(\theta; \delta)/\partial \theta_k = p_{T(n_{1:k}-1, \sqrt{n_{1:k}}\delta)}(\theta_k)\Pr(T_{-k}\leq \theta_{-k}|T_k = \theta_k, H_a)\leq p_{T(n_{1:k}-1, \sqrt{n_{1:k}}\delta)}(\theta_k)$, $\partial f(\theta; \delta)/\partial \theta_k \leq \sum^K_{j=k+1}n_jp_{T(n_{1:k}-1, \sqrt{n_{1:k}}\delta)}(\theta_k)\leq n_{1:K}/\sqrt{2\pi}$, and the mean value theorem.

\end{proof}
Lemma~\ref{lem:SAA-constraint-bound-pilot} is an analog of Lemma~\ref{lem:SAA-constraint-bound}. 
\begin{lemma}\label{lem:SAA-constraint-bound-pilot}
    There exist constants $c^\p_{g_0}, c^\p_{g_a}, c^\p_f$ such that the following arguments hold:

(i) With probability at least $1-\eta$,
\begin{align*}
    \max\Big\{\hat g_{0,M}(\tilde\theta) - \min_{\theta: \Vert\theta - \tilde\theta\Vert_\infty\leq\Delta}\hat g_{0,M}(\theta), 
    \max_{\theta: \Vert\theta - \tilde\theta\Vert_\infty\leq\Delta}\hat g_{0,M}(\theta) - \hat g_{0,M}(\tilde\theta)\Big\}\\
    \leq c^\p_{g_0}K\Delta + \epsilon_1(4\eta, M), \,\forall \tilde\theta\in \hat C(\Delta).
\end{align*}

(ii) We have the inequality below holds almost surely,
\begin{align*}
        \Pr\Big(\max\big\{\hat g_{a,M}(\tilde\theta; \hat\delta_{n_0}) - \min_{\theta: \Vert\theta - \tilde\theta\Vert_\infty\leq\Delta}\hat g_{a,M}(\theta; \hat\delta_{n_0}), \max_{\theta: \Vert\theta - \tilde\theta\Vert_\infty\leq\Delta}\hat g_{a,M}(\theta; \hat\delta_{n_0}) - \hat g_{a,M}(\tilde\theta; \hat\delta_{n_0})\big\}\\
\leq c^\p_{g_a}K\Delta + \epsilon_1(4\eta, M),\, \forall \tilde\theta\in \hat C(\Delta)| \{X_{-i}\}^{n_0}_{i=1}\Big)\geq 1-\eta, 
    \end{align*}

(iii) We have the inequality below holds almost surely,
\begin{align*}
        \Pr\Big(\max\big\{\hat f_M(\tilde\theta; \hat\delta_{n_0}) - \min_{\theta: \Vert\theta - \tilde\theta\Vert_\infty\leq\Delta}\hat f_M(\theta; \hat\delta_{n_0}), \max_{\theta: \Vert\theta - \tilde\theta\Vert_\infty\leq\Delta}\hat f_M(\theta; \hat\delta_{n_0}) - \hat f_M(\tilde\theta; \hat\delta_{n_0})\big\}\\
\leq c^\p_fK(K-1)\Delta/2 + \sum^K_{k=2}n_k\epsilon_1(4\eta/(K-1), M),\, \forall \tilde\theta\in \hat C(\Delta)| \{X_{-i}\}^{n_0}_{i=1}\Big)\geq 1-\eta.
    \end{align*}
\end{lemma}

\begin{proof}[Proof of Lemma~\ref{lem:SAA-constraint-bound-pilot}]
    (i) can be similarly proved as Lemma~\ref{lem:SAA-constraint-bound}. For (ii) and (iii), we use the concentration inequality that for any $\tilde\theta\in \hat C(\Delta)$, 
    \begin{align*}
        &\Pr\big((1/M)\sum^M_{m=1}\I(\tilde\theta_k-\Delta<T^m_k(\hat\delta_{n_0})\leq \tilde\theta_k+\Delta \text{ for some }k\in[K]) \\
    &\quad\quad\quad
    \leq \Pr(\tilde\theta_k-\Delta<T^m_k(\hat\delta_{n_0})\leq \tilde\theta_k+\Delta \text{ for some }k\in[K]|\{X_{-i}\}^{n_0}_{i=1}) + \epsilon|\{X_{-i}\}^{n_0}_{i=1}\big)\\
    \geq &1-\exp(-2M\epsilon^2), \text{ almost surely},
    \end{align*}
    and that there exists some constant $c^\p_a$ such that $\Pr(\tilde\theta_k-\Delta<T^m_k(\hat\delta_{n_0})\leq \tilde\theta_k+\Delta|\{X_{-i}\}^{n_0}_{i=1})\leq c^\p_a\Delta$ almost surely, since the density function of a noncentral Student $t$ distribution $T(n, \delta)$
    can be upper bounded by a constant that does not depend on $\delta$ (see Lemma~\ref{lem:lipshitz-continuity-t}). Then the rest of the proof follows similarly as that for Lemma~\ref{lem:SAA-constraint-bound}. We thus omit the details here.
\end{proof}

Lemma~\ref{lem:Delta-net-pilot} is an analog of Lemma~\ref{lem:Delta-net}. It can be similarly proved as Lemma~\ref{lem:Delta-net}, following from Lemmas~\ref{lem:lipshitz-continuity-t}-\ref{lem:SAA-constraint-bound-pilot}.

\begin{lemma}\label{lem:Delta-net-pilot}
    (i) $\Pr(\Theta_M(\alpha, \beta; \hat\delta_{n_0})\neq\emptyset, \text{ and }\forall \theta\in \Theta_M(\alpha, \beta; \hat\delta_{n_0}), \text{exists }\tilde\theta\in \Theta_M(\alpha + c^\p_{g_0}K\Delta + \epsilon_1(2\eta, M), \beta+ c^\p_{g_a}K\Delta + \epsilon_1(2\eta, M); \hat\delta_{n_0})\cap \hat C(\Delta)\text{ such that }\Vert \hat\theta - \tilde\theta\Vert_\infty\leq \Delta|\{X_{-i}\}^{n_0}_{i=1})\geq 1-\eta-\Pr(\Theta_M(\alpha, \beta; \hat\delta_{n_0})=\emptyset
|\{X_{-i}\}^{n_0}_{i=1})$, almost surely.

    (ii) For all $\delta$ and 
    $\hat\theta\in \Theta(\alpha, \beta; \delta)$ (given $\Theta(\alpha, \beta; \delta)\neq\emptyset$), we can find $\tilde\theta\in \Theta(\alpha + l^\p_{g_0}K\Delta, \beta+ l^\p_{g_a}K\Delta; \delta)\cap \hat C(\Delta)$ such that $\Vert \hat\theta - \tilde\theta\Vert_\infty\leq \Delta$.
\end{lemma}
\begin{proof}[Proof of Lemma~\ref{lem:Delta-net-pilot}]
    (i) Let $\tilde\theta(\hat\theta)\in\arg\min_{\theta\in \hat C(\Delta), \Vert\theta - \hat\theta\Vert_\infty\leq \Delta}\Vert\theta - \hat\theta\Vert_1$. For $\epsilon > 0$, 
\begin{align*}
&\Pr(\forall \theta\in \Theta_M(\alpha, \beta; \hat\delta_{n_0}), \text{exists }\tilde\theta\in \Theta_M(\alpha + c^\p_{g_0}K\Delta + \epsilon, \beta+ c^\p_{g_a}K\Delta + \epsilon; \hat\delta_{n_0})\cap \hat C(\Delta)\\
&\quad\quad\quad\quad\quad\text{ such that }\Vert \theta - \tilde\theta\Vert_\infty\leq \Delta, \Theta_M(\alpha, \beta; \hat\delta_{n_0})\neq\emptyset|\{X_{-i}\}^{n_0}_{i=1})\\
\geq & \Pr(\forall \theta\in \Theta_M(\alpha, \beta; \hat\delta_{n_0}), \tilde\theta(\theta)\in \Theta_M(\alpha + c^\p_{g_0}K\Delta + \epsilon, \beta+ c^\p_{g_a}K\Delta + \epsilon; \hat\delta_{n_0}), \\
&\quad\quad\quad\quad\quad
\Theta_M(\alpha, \beta; \hat\delta_{n_0})\neq\emptyset
|\{X_{-i}\}^{n_0}_{i=1}
)\\
\geq & \Pr(\forall \theta\in \Theta_M(\alpha, \beta; \hat\delta_{n_0}), 
\hat g_{0,M}(\tilde\theta(\theta)) - \hat g_{0,M}(\theta) \geq - c^\p_{g_0}K\Delta - \epsilon, \, \\
&\quad\quad\quad\quad\quad
\hat g_{a,M}(\tilde\theta(\theta); \hat\delta_{n_0}) - \hat g_{a,M}(\theta; \hat\delta_{n_0}) \leq c^\p_{g_a}K\Delta + \epsilon, \Theta_M(\alpha, \beta; \hat\delta_{n_0})\neq\emptyset
|\{X_{-i}\}^{n_0}_{i=1}
)\\
\geq & \Pr(\Theta_M(\alpha, \beta; \hat\delta_{n_0})\neq\emptyset, \text{ and }\forall \tilde\theta\in \hat C(\Delta), 
\hat g_{0,M}(\tilde\theta) - \max_{\theta: \Vert\theta - \tilde\theta\Vert_\infty\leq \Delta}\hat g_{0,M}(\theta) \geq - c^\p_{g_0}K\Delta - \epsilon, \, \\
&\quad\quad\quad\quad\quad
\hat g_{a,M}(\tilde\theta; \hat\delta_{n_0}) - \min_{\theta: \Vert\theta - \tilde\theta\Vert_\infty\leq \Delta}\hat g_{a,M}(\theta; \hat\delta_{n_0}) \leq c^\p_{g_a}K\Delta + \epsilon
|\{X_{-i}\}^{n_0}_{i=1}
)\\
\geq & 1- \Pr(\Theta_M(\alpha, \beta; \hat\delta_{n_0})\neq\emptyset, \text{ and }\exists \tilde\theta\in \hat C(\Delta), 
\hat g_{0,M}(\tilde\theta) - \max_{\theta: \Vert\theta - \tilde\theta\Vert_\infty\leq \Delta}\hat g_{0,M}(\theta) \leq - c^\p_{g_0}K\Delta - \epsilon, \, \\
&\quad\quad\quad\quad\quad
\text{ or }\exists\tilde\theta\in \hat C(\Delta), 
\hat g_{a,M}(\tilde\theta; \hat\delta_{n_0}) - \min_{\theta: \Vert\theta - \tilde\theta\Vert_\infty\leq \Delta}\hat g_{a,M}(\theta; \hat\delta_{n_0}) \geq c^\p_{g_a}K\Delta + \epsilon
|\{X_{-i}\}^{n_0}_{i=1}
) \\
&\quad\quad\quad\quad\quad\quad\quad\quad
- \Pr(\Theta_M(\alpha, \beta; \hat\delta_{n_0})=\emptyset
|\{X_{-i}\}^{n_0}_{i=1})\\
\geq & 1-2|\hat C(\Delta)|\exp(-2M\epsilon^2) - \Pr(\Theta_M(\alpha, \beta; \hat\delta_{n_0})=\emptyset
|\{X_{-i}\}^{n_0}_{i=1}),\,\text{almost surely},
\end{align*}
where the last inequality follows from Lemma~\ref{lem:SAA-constraint-bound-pilot}. 

    (ii) can be similarly proved as Lemma~\ref{lem:Delta-net}. We thus omit its proof. 
\end{proof}

\begin{proof}[Proof of Lemma~\ref{lem:feasible-set-pilot}]
    We let $\hat\epsilon_\alpha(\eta, M) = 2\epsilon_1(\eta, M) + (l^\p_{g_0} + c^\p_{g_0})K\Delta$, $\hat\epsilon_\beta(\eta, M) = 2\epsilon_1(\eta, M) + (l^\p_{g_a} + c^\p_{g_a})K\Delta$, $\varepsilon_\alpha(\eta, M) = c^\p_{g_0}K\Delta + \epsilon_1(\eta, M)$, and $\varepsilon_\beta(\eta, M) = c^\p_{g_a}K\Delta + \epsilon_1(\eta, M)$. 
\begin{align*}
&\,\Pr(\Theta_M(\alpha, \beta; \hat\delta_{n_0})\nsubseteq \Theta(\alpha + \hat\epsilon_\alpha(\eta, M), \beta+ \hat\epsilon_\beta(\eta, M); \hat\delta_{n_0})
|\{X_{-i}\}^{n_0}_{i=1}
) \\
\leq &\,1-\Pr(\Theta_M(\alpha, \beta; \hat\delta_{n_0})\neq \emptyset, \Theta_M(\alpha, \beta; \hat\delta_{n_0})\subseteq \Theta(\alpha + \hat\epsilon_\alpha(\eta, M), \beta+ \hat\epsilon_\beta(\eta, M); \hat\delta_{n_0}
),\\
&\quad\quad\quad\quad
\Theta_M(\alpha + \varepsilon_\alpha(\eta, M), \beta+ \varepsilon_\beta(\eta, M); \hat\delta_{n_0})\cap \hat C(\Delta)\\
&\quad\quad\quad\quad\quad\quad
\subseteq \Theta(\alpha + \varepsilon_\alpha(\eta, M) + \epsilon_1(\eta, M), \beta+\varepsilon_\beta(\eta, M) + \epsilon_1(\eta, M); \hat\delta_{n_0})
|\{X_{-i}\}^{n_0}_{i=1}
)\\
&\quad
-\Pr(\Theta_M(\alpha, \beta; \hat\delta_{n_0})=\emptyset
|\{X_{-i}\}^{n_0}_{i=1}), \text{ almost surely}.
    \end{align*}
Now based on similar reasoning as the proof of Proposition~\ref{prop:feasible-set}, we can show that, conditional on the realization of $\hat\delta_{n_0}$,  if for all $\hat\theta\in \Theta_M(\alpha, \beta; \hat\delta_{n_0})$, we have $\Theta_M(\alpha + \epsilon_\alpha(\eta, M), \beta + \epsilon_\beta(\eta, M); \hat\delta_{n_0})\cap \{\theta\in \hat C(\Delta): \Vert \theta - \hat\theta\Vert_\infty\leq \Delta\}\neq\emptyset$ and $\Theta_M(\alpha + \varepsilon_\alpha(\eta, M), \beta+ \varepsilon_\beta(\eta, M); \hat\delta_{n_0})\cap \hat C(\Delta)
\subseteq \Theta(\alpha + \varepsilon_\alpha(\eta, M) + \epsilon_1(\eta, M), \beta+\varepsilon_\beta(\eta, M) + \epsilon_1(\eta, M)); \hat\delta_{n_0})$, then $\Theta_M(\alpha, \beta; \hat\delta_{n_0})\subseteq \Theta(\alpha + \hat\epsilon_\alpha(\eta, M), \beta + \hat\epsilon_\beta(\eta, M); \hat\delta_{n_0})$. Thus
% Consider any $\hat\theta\in \Theta_M(\alpha, \beta; \hat\delta_{n_0})$, since $\Theta_M(\alpha + \epsilon_\alpha(\eta, M), \beta + \epsilon_\beta(\eta, M); \hat\delta_{n_0})\cap \{\theta\in \hat C(\Delta): \Vert \theta - \hat\theta\Vert_\infty\leq \Delta\}\neq\emptyset$, we can find $\tilde\theta(\hat\theta) \in \Theta_M(\alpha + \epsilon_\alpha(\eta, M), \beta + \epsilon_\beta(\eta, M); \hat\delta_{n_0})\cap \{\theta\in \hat C(\Delta): \Vert \theta - \hat\theta\Vert_\infty\leq \Delta\}$. Then $g_0(\hat\theta) \geq g_0(\tilde\theta(\hat\theta)) - l_{g_0}\Vert\hat\theta - \tilde\theta(\hat\theta)\Vert_1\geq 1-\alpha - \varepsilon_\alpha(\eta, M)  - \epsilon_1(\eta, M) - l_{g_0}K\Delta = 1-\alpha - \hat\epsilon_\alpha(\eta, M)$, where the first inequality follows from Lemma~\ref{lem:lipshitz-continuity-t} and that $\Vert\tilde\theta(\hat\theta) - \hat\theta\Vert_\infty\leq \Delta$. The second inequality follows since $\tilde\theta(\hat\theta)\in \Theta_M(\alpha + \varepsilon_\alpha(\eta, M), \beta+ \varepsilon_\beta(\eta, M); \hat\delta_{n_0})\cap \hat C(\Delta)\subseteq \Theta(\alpha + \varepsilon_\alpha(\eta, M) + \epsilon_1(\eta, M), \beta+\varepsilon_\beta(\eta, M) + \epsilon_1(\eta, M); \hat\delta_{n_0})$. Similarly, we can show that $g_a(\hat\theta; \hat\delta_{n_0}) \leq g_a(\tilde\theta(\hat\theta); \hat\delta_{n_0}) + l_{g_a}\Vert\hat\theta - \tilde\theta(\hat\theta)\Vert_1\leq \beta + \varepsilon_\beta(\eta, M) + \epsilon_1(\eta, M) + l_{g_a}K\Delta = \beta + \hat\epsilon_\beta(\eta, M)$. Based on the above arguments, we have that 
\begin{align*}
&\,\Pr(\Theta_M(\alpha, \beta; \hat\delta_{n_0})\neq \emptyset, \Theta_M(\alpha, \beta; \hat\delta_{n_0})\subseteq \Theta(\alpha + \hat\epsilon_\alpha(\eta, M), \beta+ \hat\epsilon_\beta(\eta, M); \hat\delta_{n_0}
),\\
&\quad\quad\quad\quad
\Theta_M(\alpha + \varepsilon_\alpha(\eta, M), \beta+ \varepsilon_\beta(\eta, M); \hat\delta_{n_0})\cap \hat C(\Delta)\\
&\quad\quad\quad\quad\quad\quad
\subseteq \Theta(\alpha + \varepsilon_\alpha(\eta, M) + \epsilon_1(\eta, M), \beta+\varepsilon_\beta(\eta, M) + \epsilon_1(\eta, M); \hat\delta_{n_0})
|\{X_{-i}\}^{n_0}_{i=1}
)\\
\geq & \Pr(\Theta_M(\alpha, \beta; \hat\delta_{n_0})\neq \emptyset, \text{ and }\forall \hat\theta\in \Theta_M(\alpha, \beta; \hat\delta_{n_0}), \Theta_M(\alpha + \varepsilon_\alpha(\eta, M), \beta+ \varepsilon_\beta(\eta, M); \hat\delta_{n_0})\\
&\quad\quad\quad\quad
\cap \{\theta\in \hat C(\Delta): \Vert\theta - \hat\theta\Vert_\infty\leq \Delta\}\neq\emptyset, \,
\Theta_M(\alpha + \varepsilon_\alpha(\eta, M), \beta+ \varepsilon_\beta(\eta, M); \hat\delta_{n_0})\cap \hat C(\Delta)\\
&\quad\quad\quad\quad\quad\quad
\subseteq \Theta(\alpha + \varepsilon_\alpha(\eta, M) + \epsilon_1(\eta, M), \beta+\varepsilon_\beta(\eta, M) + \epsilon_1(\eta, M); \hat\delta_{n_0})
|\{X_{-i}\}^{n_0}_{i=1}
)\\
\geq & 1-\Pr(\Theta_M(\alpha, \beta; \hat\delta_{n_0})= \emptyset, \text{ or }\Theta_M(\alpha, \beta; \hat\delta_{n_0})\neq \emptyset, 
\text{ and }\exists \hat\theta\in \Theta_M(\alpha, \beta; \hat\delta_{n_0}), \\
&\quad\quad\quad\quad
\Theta_M(\alpha + \varepsilon_\alpha(\eta, M), \beta+ \varepsilon_\beta(\eta, M); \hat\delta_{n_0})\\
&\quad\quad\quad\quad
\cap \{\theta\in \hat C(\Delta): \Vert\theta - \hat\theta\Vert_\infty\leq \Delta\} = \emptyset
|\{X_{-i}\}^{n_0}_{i=1}
) \\
&\quad
- \Pr(\Theta_M(\alpha + \varepsilon_\alpha(\eta, M), \beta+ \varepsilon_\beta(\eta, M); \hat\delta_{n_0})\cap \hat C(\Delta)\\
&\quad\quad\quad\quad
\nsubseteq \Theta(\alpha + \varepsilon_\alpha(\eta, M) + \epsilon_1(\eta, M), \beta+\varepsilon_\beta(\eta, M) + \epsilon_1(\eta, M); \hat\delta_{n_0})
|\{X_{-i}\}^{n_0}_{i=1}
)\\
\geq & 1-2\eta-
\Pr(\Theta_M(\alpha, \beta; \hat\delta_{n_0})= \emptyset
|\{X_{-i}\}^{n_0}_{i=1}), \text{almost surely},
    \end{align*}
where the last inequality follows from Lemmas~\ref{lem:finite-feasibility-pilot} and \ref{lem:Delta-net-pilot}(i). Based on the above arguments, 
$$
\Pr(\Theta_M(\alpha, \beta; \hat\delta_{n_0})\nsubseteq \Theta(\alpha + \hat\epsilon_\alpha(\eta, M), \beta+ \hat\epsilon_\beta(\eta, M); \hat\delta_{n_0})
|\{X_{-i}\}^{n_0}_{i=1}
)\leq 2\eta, \text{ almost surely}.
$$

On the other hand, we have that 
    \begin{equation*}
        \begin{split}
&\,\Pr(\Theta(\alpha - \hat\epsilon_\alpha(\eta, M), \beta -\hat\epsilon_\beta(\eta, M); \hat\delta_{n_0})\subseteq \Theta_M(\alpha, \beta; \hat\delta_{n_0})
|\{X_{-i}\}^{n_0}_{i=1}
) \\
\geq &\,\Pr(\Theta(\alpha - \hat\epsilon_\alpha(\eta, M), \beta -\hat\epsilon_\beta(\eta, M); \hat\delta_{n_0})\subseteq \Theta_M(\alpha, \beta; \hat\delta_{n_0}),\\
&\quad\quad\quad\quad
\Theta(\alpha - \hat\epsilon_\alpha(\eta, M) + l^\p_{g_0}K\Delta, \beta -\hat\epsilon_\beta(\eta, M) + l^\p_{g_a}K\Delta; \hat\delta_{n_0})\cap \hat C(\Delta)\\
&\quad\quad\quad\quad\quad\quad
\subseteq \Theta_M(\alpha - \varepsilon_\alpha(\eta, M), \beta-\varepsilon_\beta(\eta, M); \hat\delta_{n_0})
|\{X_{-i}\}^{n_0}_{i=1}
), \text{ almost surely}.
        \end{split}
    \end{equation*}
Now based on similar reasoning as Proposition~\ref{prop:feasible-set}, we have that given the realization of $\hat\delta_{n_0}$, if for all $\tilde\theta\in \hat C(\Delta)$, we have $\min_{\theta: \Vert\theta-\tilde\theta\Vert_\infty\leq \Delta}\hat g_{0,M}(\theta) - \hat g_{0,M}(\tilde\theta)\geq -\epsilon_1(\eta, M) - c^\p_{g_0}K\Delta$ and $\max_{\theta: \Vert\theta-\tilde\theta\Vert_\infty\leq \Delta}\hat g_{a,M}(\theta; \hat\delta_{n_0}) - \hat g_{a,M}(\tilde\theta; \hat\delta_{n_0})\leq \epsilon_1(\eta, M) + c^\p_{g_a}K\Delta$, and $\Theta(\alpha - \hat\epsilon_\alpha(\eta, M) + l^\p_{g_0}K\Delta, \beta -\hat\epsilon_\beta(\eta, M) + l^\p_{g_a}K\Delta; \hat\delta_{n_0})
\cap \hat C(\Delta)\subseteq \Theta_M(\alpha - \varepsilon_\alpha(\eta, M), \beta-\varepsilon_\beta(\eta, M); \hat\delta_{n_0})$, then $\Theta(\alpha - \hat\epsilon_\alpha(\eta, M), \beta -\hat\epsilon_\beta(\eta, M); \hat\delta_{n_0})\subseteq \Theta_M(\alpha, \beta; \hat\delta_{n_0})$. Thus
\begin{align*}
&\,\Pr(\Theta(\alpha - \hat\epsilon_\alpha(\eta, M), \beta -\hat\epsilon_\beta(\eta, M); \hat\delta_{n_0}))\subseteq \Theta_M(\alpha, \beta; \hat\delta_{n_0})),\\
&\quad\quad\quad\quad
\Theta(\alpha - \hat\epsilon_\alpha(\eta, M) + l^\p_{g_0}K\Delta, \beta -\hat\epsilon_\beta(\eta, M) + l^\p_{g_a}K\Delta; \hat\delta_{n_0}))\cap \hat C(\Delta)\\
&\quad\quad\quad\quad\quad\quad
\subseteq \Theta_M(\alpha - \varepsilon_\alpha(\eta, M), \beta-\varepsilon_\beta(\eta, M); \hat\delta_{n_0}))
|\{X_{-i}\}^{n_0}_{i=1}
)\\
\geq & \Pr(\min_{\theta: \Vert\theta-\tilde\theta\Vert_\infty\leq \Delta}\hat g_{0,M}(\theta) - \hat g_{0,M}(\tilde\theta)\geq -\epsilon_1(\eta, M) - c^\p_{g_0}K\Delta,\,\forall \tilde\theta\in\hat C(\Delta),\\
& 
\quad\quad\quad\quad\max_{\theta: \Vert\theta-\tilde\theta\Vert_\infty\leq \Delta}\hat g_{a,M}(\theta; \hat\delta_{n_0}) - \hat g_{a,M}(\tilde\theta; \hat\delta_{n_0})\leq \epsilon_1(\eta, M) + c^\p_{g_a}K\Delta,\, \forall \tilde\theta\in\hat C(\Delta),\\
&\quad\quad\quad\quad
\Theta(\alpha - \hat\epsilon_\alpha(\eta, M) + l^\p_{g_0}K\Delta, \beta-\hat\varepsilon_\beta(\eta, M) + l^\p_{g_a}K\Delta; \hat\delta_{n_0})\cap \hat C(\Delta)\\
&\quad\quad\quad\quad\quad\quad
\subseteq \Theta_M(\alpha -\varepsilon_\alpha(\eta, M), \beta-\varepsilon_\beta(\eta, M); \hat\delta_{n_0}))
|\{X_{-i}\}^{n_0}_{i=1}
)\\
\geq &1-\Pr(\exists \tilde\theta\in\hat C(\Delta), \min_{\theta: \Vert\theta-\tilde\theta\Vert_\infty\leq \Delta}\hat g_{0,M}(\theta) - \hat g_{0,M}(\tilde\theta)< -\epsilon_1(\eta, M) - c^\p_{g_0}K\Delta
|\{X_{-i}\}^{n_0}_{i=1}
)\\
& 
-\Pr(\exists \tilde\theta\in\hat C(\Delta), \max_{\theta: \Vert\theta-\tilde\theta\Vert_\infty\leq \Delta}\hat g_{a,M}(\theta; \hat\delta_{n_0}) - \hat g_{a,M}(\tilde\theta; \hat\delta_{n_0}))> \epsilon_1(\eta, M) + c^\p_{g_a}K\Delta
|\{X_{-i}\}^{n_0}_{i=1}
)\\
&-\Pr(\Theta(\alpha - \hat\epsilon_\alpha(\eta, M) + l^\p_{g_0}K\Delta, \beta -\hat\epsilon_\beta(\eta, M) + l^\p_{g_a}K\Delta; \hat\delta_{n_0})\cap \hat C(\Delta)\\
&\quad\quad\quad\quad
\nsubseteq \Theta_M(\alpha - \varepsilon_\alpha(\eta, M), \beta-\varepsilon_\beta(\eta, M); \hat\delta_{n_0})
|\{X_{-i}\}^{n_0}_{i=1}
)\\
\geq &1- 2\eta, \text{ almost surely},
\end{align*}
where the last inequality follows from Lemmas~\ref{lem:finite-feasibility-pilot} and \ref{lem:SAA-constraint-bound-pilot}(i)-(ii). 

Combining the above arguments, 
\begin{align*}
\Pr(\Theta(\alpha - \hat\epsilon_\alpha(\eta, M), \beta -\hat\epsilon_\beta(\eta, M); \hat\delta_{n_0})\subseteq \Theta_M(\alpha, \beta;\hat\delta_{n_0})\\
\subseteq \Theta(\alpha + \hat\epsilon_\alpha(\eta, M), \beta+ \hat\epsilon_\beta(\eta, M); \hat\delta_{n_0})
|\{X_{-i}\}^{n_0}_{i=1}
)
\geq 1-4\eta, \text{ almost surely}, 
\end{align*}

Then by law of iterated expectation, we have that 
\begin{align*}
    &\Pr\big(\Theta(\alpha - \hat\epsilon_\alpha(\eta, M), \beta - \hat\epsilon_\beta(\eta, M); \hat\delta_{n_0})\subseteq \Theta_M(\alpha, \beta; \hat\delta_{n_0})
    \subseteq \Theta(\alpha + \hat\epsilon_\alpha(\eta, M), \beta+ \hat\epsilon_\beta(\eta, M); \hat\delta_{n_0})\big)\\
     = &\E\Big[\Pr\big(\Theta(\alpha - \hat\epsilon_\alpha(\eta, M), \beta - \hat\epsilon_\beta(\eta, M); \hat\delta_{n_0})\subseteq \Theta_M(\alpha, \beta; \hat\delta_{n_0})\\
    &\quad\quad\quad\quad\quad\quad
    \subseteq \Theta(\alpha + \hat\epsilon_\alpha(\eta, M), \beta+ \hat\epsilon_\beta(\eta, M); \hat\delta_{n_0})|X_{-1}, \cdots, X_{-n_0}\big)\Big]\\
    \geq &1-4\eta. 
\end{align*}
Then the argument follows by letting $\Delta = (\bar\theta - \underline{\theta})/\sqrt{M}$.

\end{proof}

\subsection{Proof of Lemma~\ref{lem:pilot-feasible-set-2}}\label{sec-appendix:pilot-feasible-set-2}

\begin{lemma}\label{lem:pilot-sample-convergence}
Let $\delta\neq 0$ and $\eta\in(0,1)$. For any $\theta$, and $n_0\geq 64\log(4/\eta) + 1$, 
    $$
    \Pr(|g_a(\theta, \hat\delta_{n_0}) - g_a(\theta; \delta)|\leq c^g_\delta|\delta|
(2\sqrt{\log(4/\eta)/(n_0 - 1)} + 2\log(4/\eta)/(n_0 - 1)))\geq 1-\eta
$$
and
$$
\Pr(|f(\theta; \hat\delta_{n_0}) - f(\theta; \delta)|\leq c^f_\delta|\delta|
(2\sqrt{\log(4/\eta)/(n_0 - 1)} + 2\log(4/\eta)/(n_0 - 1)))\geq 1-\eta,
    $$
    where $0\leq c^g_\delta, c^f_\delta <\infty$ are constants that do not depend on $n_0$ nor $M$. 
\end{lemma}
\begin{proof}[Proof of Lemma~\ref{lem:pilot-sample-convergence}]
We prove the first inequality, and the second can be proved similarly. 
    We first notice that for all $\theta\in C$, we have that $\partial g_a(\theta; \delta)/\partial \delta < \infty$. It then follows from the mean value theorem that there exists constant $0\leq \tilde c^g_\delta< \infty$ such that $|g_a(\theta, \delta_1) - g_a(\theta, \delta_2)|\leq \tilde c^g_\delta|\delta_1 - \delta_2|$. Letting $U = (n_0-1)\hat\sigma^2_0/\sigma^2$, then $U\sim\chi^2_{n_0-1}$, where we use $\chi^2_n$ to denote the $\chi^2$ distribution with $n$ degrees of freedom. 
According to Laurent-Massart bound for $U\sim \chi^2_{n_0-1}$ (Lemma 1 of \citet{laurent2000adaptive}), with probability at least $1-2\exp(-x)$, 
$$
|U/(n_0 - 1) - 1|\leq 2\sqrt{x/(n_0 - 1)} + 2x/(n_0 - 1).
$$
Let $\eta \in(0,1)$. Let $\tilde n_\p = 64\log(4/\eta) + 1$.
Then for $n_0\geq \tilde n_\p$, with probability at least $1-\eta/2$, we have that $U/(n_0 - 1)\geq 1/2$. Additionally, by letting $2\exp(-x) = \eta/2$, we have that with probability at least $1-\eta$, $|U/(n_0 - 1) - 1|\leq 2\sqrt{\log(4/\eta)/(n_0 - 1)} + 2\log(4/\eta)/(n_0 - 1)$. 

Since
$$
|\delta - \hat\delta_{n_0}| = |\mu_a/\sigma - \mu_a/\hat\sigma_0| = |\mu_a/\sigma|\cdot |1-\sqrt{(n_0 - 1)/U}|\leq |\delta|\cdot|1-U/(n_0 - 1)|/\sqrt{U/(n_0 - 1)},
$$
it follows that for $n_0\geq 64\log(4/\eta) + 1$, with probability at least $1-\eta$, 
$$
|g_a(\theta; \hat\delta_{n_0}) - g_a(\theta; \delta)|\leq \tilde c^g_\delta|\delta - \hat\delta_{n_0}|\leq 2\tilde c^g_\delta|\delta|
(2\sqrt{\log(4/\eta)/(n_0 - 1)} + 2\log(4/\eta)/(n_0 - 1)). 
$$
Letting $c^g_\delta = 2\tilde c^g_\delta$ we have the desired result.

\end{proof}

\begin{proof}[Proof of Lemma~\ref{lem:pilot-feasible-set-2}]
Recall that $\hat C(\Delta) = \{\underline{\theta}, \underline{\theta} + \Delta, \underline{\theta} + 2\Delta, \cdots, \bar\theta\}^K$ is a discretization of $C$.
Let $\epsilon(\eta, n_0) = c^g_\delta|\delta|
(2\sqrt{\log(4|\hat C|/\eta)/(n_0 - 1)} + 2\log(4|\hat C|/\eta)/(n_0 - 1)) + 2l^\p_{g_a}(\bar\theta - \underline{\theta})/\sqrt{n_0}$. We have 
    \begin{align*}
    &\Pr(\Theta(\alpha, \beta - \epsilon(\eta, n_0); \delta)\subseteq \Theta(\alpha, \beta; \hat\delta_{n_0})\subseteq \Theta(\alpha, \beta + \epsilon(\eta, n_0); \delta))\\
    \geq& 1-\Pr(\Theta(\alpha, \beta - \epsilon(\eta, n_0); \delta)\nsubseteq \Theta(\alpha, \beta; \hat\delta_{n_0})) 
    - \Pr(\Theta(\alpha, \beta; \hat\delta_{n_0})\nsubseteq \Theta(\alpha, \beta + \epsilon(\eta, n_0); \delta))
    \end{align*}
    Now we show that $\Pr(\Theta(\alpha, \beta - \epsilon(\eta, n_0); \delta)\nsubseteq \Theta(\alpha, \beta; \hat\delta_{n_0}))\leq \eta/2$, and that $\Pr(\Theta(\alpha, \beta; \hat\delta_{n_0})\nsubseteq \Theta(\alpha, \beta + \epsilon(\eta, n_0); \delta)) \leq \eta/2$ can be proved similarly. We have that
    \begin{align*}
        &\Pr(\Theta(\alpha, \beta - \epsilon(\eta, n_0); \delta)\nsubseteq \Theta(\alpha, \beta; \hat\delta_{n_0}))\\
        =&\Pr(\exists \theta\in \Theta(\alpha, \beta - \epsilon(\eta, n_0); \delta) \text{ such that }\theta\notin \Theta(\alpha, \beta; \hat\delta_{n_0}))\\
        \leq & \Pr(\exists \theta\in C, \text{ such that }g_a(\theta; \hat\delta_{n_0})\geq g_a(\theta; \delta) + \epsilon(\eta, n_0))\\
        \leq & \Pr(\exists \theta\in \hat C(\Delta), \text{ such that }g_a(\theta; \hat\delta_{n_0})\geq g_a(\theta; \delta) + \epsilon(\eta, n_0) - 2l^\p_{g_a}K\Delta)\\
        \leq & \eta/2. 
    \end{align*}

In the above equations, the first inequality follows since if $\theta\in \Theta(\alpha, \beta - \epsilon(\eta, n_0); \delta)$ and $\theta\notin\Theta(\alpha, \beta; \hat\delta_{n_0})$, then $g_a(\theta; \hat\delta_{n_0})\geq \beta \geq g_a(\theta; \delta) + \epsilon(\eta, n_0)$. The second inequality follows since for all $\theta\in C$ 
we can find $\hat\theta\in \hat C(\Delta)$ with $\Vert \hat\theta - \theta\Vert_\infty\leq \Delta$. Then according to Lemma~\ref{lem:lipshitz-continuity-t}, we must also have $g_a(\theta, \hat\delta_{n_0}) - g_a(\hat\theta; \hat\delta_{n_0})\leq l^\p_{g_a}K\Delta$ and $g_a(\hat\theta, \delta) - g_a(\theta; \delta)\leq l^\p_{g_a}K\Delta$ almost surely. As a result $g_a(\hat\theta; \hat\delta_{n_0})\geq g_a(\theta, \hat\delta_{n_0}) - l^\p_{g_a}K\Delta\geq g_a(\theta; \delta) + \epsilon(\eta, n_0) - l^\p_{g_a}K\Delta \geq g_a(\hat\theta, \delta) + \epsilon(\eta, n_0) - 2l^\p_{g_a}K\Delta$. Then the last inequality follows from Lemma~\ref{lem:pilot-sample-convergence} and by letting $\Delta = (\bar\theta - \underline{\theta})/\sqrt{n_0}$. 

\end{proof}

\section{Additional Materials for Section~\ref{sec:convergence-gen}}

\subsection{Proof of Proposition~\ref{prop:convergence-generalization}}
Proposition~\ref{prop:convergence-generalization} can be proved similarly as Propositions~\ref{prop:feasible-set} and \ref{prop:objective-bound}. Specifically, 
we notice that the analog of Lemma~\ref{lem:finite-feasibility} holds regardless of the distribution of the summary statistics. Analogs of Lemmas~\ref{lem:lipshitz-continuity} and \ref{lem:SAA-constraint-bound} hold when Assumption~\ref{ass:convergence-generalization}(i) holds. An analogue of Lemma~\ref{lem:Delta-net} holds as it follows from Analogs of Lemmas~\ref{lem:lipshitz-continuity} and \ref{lem:SAA-constraint-bound}.  An analogue of Lemma~\ref{lem:optimal-value-lipshitz} holds when Assumption~\ref{ass:convergence-generalization}(iii) holds. Then Proposition~\ref{prop:convergence-generalization} can be proved based on preceding lemmas.

%%%%

\subsection{Sequential Analysis with Symmetric Designs}\label{sec:futility}
As mentioned at the end of Section~\ref{sec:convergence-gen}, we return to the symmetric design problem. All proofs for this section are presented in Appendix~\ref{sec-appendix:futility}.

Formally, in each stage $k\in[K-1]$, we choose a pair of cutoff values $\theta_k = (\theta^0_k, \theta^1_k)$ with $\theta^0_k < \theta^1_k$. We end the experiment early in favor of $H_0$ if the summary statistics $S_k \leq \theta^0_k$, we end the experiment early in favor of $H_a$ if $S_k\geq \theta^1_k$, and we continue the experiment if $S_k \in (\theta^0_k, \theta^1_k)$. In the last stage $K$, we have $\theta^0_K = \theta^1_K$, as we reach the conclusion of whether rejecting $H_0$ or not. Let $\theta = (\theta^0_1, \cdots, \theta^0_K, \theta^1_1, \cdots, \theta^1_K)$. Let $C = [\underline{\theta},\bar\theta]^{2K}$ be the feasible region of $\theta$. 
Then we can formulate the optimization problem for the sequential testing as:
\begin{subequations}\label{equ:symmetric-design}
        \begin{align}
            \min_{\theta} &\,n_1 + \sum^K_{k=2}n_k\Pr(S_i\in (\theta^0_i, \theta^1_i), \forall i\in[k-1]|\tilde{H})\label{equ:symmetric-obj}\\
            \st&\,\Pr(S_1 \geq \theta^1_1|H_0) + \sum^K_{k=2}\Pr(S_k\geq \theta^1_k, \text{ and } S_i \in (\theta^0_i, \theta^1_i)\text{ for all }i\in[k-1]|H_0)\leq \alpha,\label{equ:symmetric-c1}\\
            &\,\Pr(S_1\geq \theta^1_1|H_a) + \sum^K_{k=2}\Pr(S_k\geq \theta^1_k, \text{ and } S_i \in (\theta^0_i, \theta^1_i)\text{ for all }i\in[k-1]|H_a)\geq 1-\beta,\label{equ:symmetric-c2}\\
            &\, \theta\in C, \theta^0_K = \theta^1_K, \theta^0_k\leq \theta^1_k, \forall k\in[K-1],
        \end{align}
\end{subequations}
where $\tilde{H}$ can be $H_0$ (optimizing accepting the null) or $H_a$ (optimizing rejecting the null) or a hybrid strategy $\tilde{H}: \mu = \lambda \mu_a$ for some $\lambda\in[0,1]$ (optimizing a combination of both the null and alternative). 
Equation~\eqref{equ:symmetric-obj} computes the expected sample size under symmetric designs, as each $\Pr(S_i\in (\theta^0_i, \theta^1_i), \forall i\in[k-1]|\tilde H)$ represents the probability to continue the experiment to stage $k$. Equation~\eqref{equ:symmetric-c1} says the type I error is no larger than $\alpha$, as each $\Pr(S_k\geq \theta^1_k, \text{ and } S_i \in (\theta^0_i, \theta^1_i)\text{ for all }i\in[k-1]|H_0)$ calculates the probability to reject $H_0$ in stage $k$. Equation~\eqref{equ:symmetric-c2} ensures the power is no less than $1-\beta$. Similar to problem~\eqref{equ:main-1},  
problem~\eqref{equ:symmetric-design} generalizes the formulations for one-sided tests, two-sided tests, one-sample, and two-sample tests. 

For convenience, we reformulate problem~\eqref{equ:symmetric-design} as:
\begin{subequations}\label{equ:symmetric-design-1}
        \begin{align}
            \min_{\theta} &\,n_1 + \sum^K_{k=2}n_k\Pr(S_i\in (\theta^0_i, \theta^1_i), \forall i\in[k-1]|\tilde H)\label{equ:symmetric-1-obj}\\
            \st&\,\Pr(S_1\leq\theta^0_1|H_0) + \sum^K_{k=2}\Pr(S_k\leq \theta^0_k, \text{ and } S_i \in [\theta^0_i, \theta^1_i]\text{ for all }i\in[k-1]|H_0)\geq 1- \alpha,\label{equ:symmetric-1-c1}\\
            &\,\Pr(S_1<\theta^0_1|H_a) + \sum^K_{k=2}\Pr(S_k< \theta^0_k, \text{ and } S_i \in (\theta^0_i, \theta^1_i)\text{ for all }i\in[k-1]|H_a)\leq \beta,\label{equ:symmetric-1-c2}\\
            &\, \theta\in C, \theta^0_K = \theta^1_K, \theta^0_k\leq \theta^1_k, \forall k\in[K-1].
        \end{align}
\end{subequations}
\begin{lemma}\label{lem:reformulation-symmetric design}
    Problem~\eqref{equ:symmetric-design} and \eqref{equ:symmetric-design-1} are equivalent, i.e., they have the same optimal value and (potentially non-unique) solution(s).
\end{lemma}
In the remainder of this section, we thus focus on problem~\eqref{equ:symmetric-design-1}. 

\subsubsection{S-MILP Reformulation}\label{sec:futility-milp}
Let $S^m = (S^m_k)_{k\in[K]}$ for each $m\in[M]$ be a sample path of $(S_k)_{k\in[K]}$ under $H_0$, let $S^m_a = (S^m_{a, k})_{k\in[K]}$ for each $m\in[M]$ be a sample path of $(S_k)_{k\in[K]}$ under $H_a$, and let $S^m_{a'} = (S^m_{a', k})_{k\in[K]}$ for each $m\in[M]$ be a sample path of $(S_k)_{k\in[K]}$ under $\tilde H$. Then the SAA analog of \eqref{equ:symmetric-design-1} is:
\begin{subequations}\label{equ:SAA-symmetric-design}
        \begin{align}
            \min_{\theta} &\,n_1 + (1/M)\sum^M_{m=1}\sum^K_{k=2}n_k\I(S^m_{a', i}\in (\theta^0_i, \theta^1_i), \forall i\in[k-1])\\
            \st&\,(1/M)\sum^M_{m=1}
            \big(\I(S^m_1\leq \theta^0_1) + \sum^K_{k=2}\I(S^m_k\leq \theta^0_k, \text{ and } S^m_i \in [\theta^0_i, \theta^1_i]\text{ for all }i\in[k-1])\big)
            \geq 1-\alpha,\label{equ:saa-symmetric-c1}\\
            &\,(1/M)\sum^M_{m=1}
            \big(\I(S^m_{a,1}< \theta^0_1) + \sum^K_{k=2}\I(S^m_{a, k}< \theta^0_k, \text{ and } S^m_{a,i} \in (\theta^0_i, \theta^1_i)\text{ for all }i\in[k-1])\big)
            \leq \beta,\label{equ:saa-symmetric-c2}\\
            &\, \theta\in C, \theta^0_K = \theta^1_K, \theta^0_k\leq \theta^1_k, \forall k\in[K-1].
        \end{align}
    \end{subequations}

To reformulate \eqref{equ:SAA-symmetric-design}, we follow similar procedures as those in Section~\ref{sec:milp-reformulation} and introduce the following binary variables. 
For each sample path $S^m$, we let $w^m_k\in\{0,1\}$ serve as a surrogate for $\I(S^m_k\leq \theta^1_k)$, so that if $S^m_k > \theta^1_k$, then $w^m_k = 0$. This is enforced by constraint \eqref{equ:milp-symmetric-c1}. Similarly, $\tau^m_k$ and $\lambda^m_k$ are surrogates for $\I(S^m_k\geq \theta^0_k)$ and $\I(S^m_k\leq \theta^0_k)$ respectively, so that $\tau^m_k = 0$ if $\theta^0_k > S^m_k$ and $\lambda^m_k = 0$ if $\theta^0_k < S^m_k$. These are enforced by constraints \eqref{equ:milp-symmetric-c2} and \eqref{equ:milp-symmetric-c2-hat}, respectively. The empirical type-1 constraint in \eqref{equ:saa-symmetric-c1} involves the conjunction of events $\{S^m_k\leq \theta^0_k\}$, $\{S^m_i\geq \theta^0_i\}$ and $\{S^m_i\leq \theta^1_i\}$ for $i\in[k-1]$. We thus introduce $\kappa^m_k$ to encode $\I(S^m_k\leq \theta^0_k, \text{ and } S^m_i \in [\theta^0_i, \theta^1_i]\text{ for all }i\in[k-1])$. Specifically, we require that $\kappa^m_k = 0$ if $S^m_k > \theta^0_k$, or $S^m_i < \theta^0_i$, or $S^m_i > \theta^1_i$ for some $i\in[k-1]$. In other words, $\kappa^m_k \leq \lambda^m_k$, $\kappa^m_k \leq \tau^m_i$ and $\kappa^m_k \leq w^m_i$ for all $i\in[k-1]$ hold simultaneously, which correspond to constraints \eqref{equ:milp-symmetric-c3}-\eqref{equ:milp-symmetric-c5}. Aggregating $\kappa^m_k$ across sample paths in constraint \eqref{equ:milp-symmetric-c6} then yields a linear expression for
the empirical probability of not rejecting $H_0$ under $H_0$, which is constrained to be at least $1-\alpha$.
The other sets of binary variables and enforcing constraints for sample paths $S^m_{a}$ and $S^m_{a'}$ can be constructed similarly and we omit further details. 
Then the resulting MILP reformulation for \eqref{equ:SAA-symmetric-design} is:
\begin{subequations}\label{equ:milp-symmetric-design}
        \begin{align}
            \min&\,n_1 + (1/M)\sum^M_{m=1}\sum^K_{k=2}n_k\xi^m_{k-1}
            \label{equ:milp-symmetric-ob}\\
            \st
            &\, \theta^1_k - \underline{\theta}\geq w^m_k(S^m_k- \underline{\theta}), \forall k\in[K], m\in[M], 
            \label{equ:milp-symmetric-c1}\\
            &\, \theta^0_k - S^m_k \leq (1-\tau^m_k)(\bar\theta - S^m_k), \forall k\in[K],m\in[M], \label{equ:milp-symmetric-c2}\\
            &\, \theta^0_k - \underline{\theta} \geq \lambda^m_k(S^m_k - \underline{\theta}), \forall k\in[K],m\in[M], \label{equ:milp-symmetric-c2-hat}\\
            &\,\kappa^m_k\leq w^m_i, \,\forall i\in[k-1], 2\leq k\leq K, m\in[M],
            \label{equ:milp-symmetric-c3}\\
            &\,\kappa^m_k\leq \tau^m_i, \,\forall i\in[k-1], 2\leq k\leq K, m\in[M],
            \label{equ:milp-symmetric-c4}\\
            &\,\kappa^m_k\leq \lambda^m_k,\,\forall k\in[K],m\in[M],\label{equ:milp-symmetric-c5}\\
            &\, (1/M)\sum^M_{m=1}\sum^K_{k=1}\kappa^m_k \geq 1-\alpha, \label{equ:milp-symmetric-c6}\\
            &\,(\bar\theta - S^m_{a,k})\rho^m_k\geq \theta^1_k - S^m_{a,k}, \,\forall k\in[K], m\in[M],
            \label{equ:milp-symmetric-c7}\\
            &\, (S^m_{a,k} - \underline{\theta})\gamma^m_k\geq S^m_{a,k} - \theta^0_k,\,\forall k\in[K], m\in[M],\label{equ:milp-symmetric-c8}\\
            &\, (\bar\theta - S^m_{a,k})\pi^m_k\geq  \theta^0_k - S^m_{a,k},\,\forall k\in[K], m\in[M],\label{equ:milp-symmetric-c8-hat}\\
            &\, \zeta^m_k \geq \sum^{k-1}_{i=1}\rho^m_i + \sum^{k-1}_{i=1}\gamma^m_i + \pi^m_k -2k+2, \,\forall 2\leq k\leq K, m\in[M], \label{equ:milp-symmetric-c9}\\
            &\, \zeta^m_1 \geq \pi^m_1, \,\forall m\in[M]\label{equ:milp-symmetric-c10}\\
            &\, (1/M)\sum^M_{m=1}\sum^K_{k=1}\zeta^m_k\leq \beta,\label{equ:milp-symmetric-c11}\\
            &\,(\bar\theta - S^m_{a',k})\varphi^m_k\geq \theta^1_k - S^m_{a',k}, \,\forall k\in[K], m\in[M],
            \label{equ:milp-symmetric-c12}\\
            &\, (S^m_{a',k} - \underline{\theta})\phi^m_k\geq S^m_{a',k} - \theta^0_k,\,\forall k\in[K], m\in[M],\label{equ:milp-symmetric-c13}\\
            &\, \xi^m_k \geq \sum^{k}_{i=1}\varphi^m_i + \sum^{k}_{i=1}\phi^m_i -2k+1, \,\forall 1\leq k\leq K, m\in[M], \label{equ:milp-symmetric-c14}\\
            &\, w^m_k\in\{0,1\}, \tau^m_k\in\{0,1\},\rho^m_k\in\{0,1\},\gamma^m_k\in\{0,1\},\kappa^m_k\in\{0,1\},\forall k\in[K], m\in[M],\\
            &\,\lambda^m_k\in\{0,1\}, \pi^m_k\in\{0,1\},\forall k\in[K], m\in[M],\\
            &\,
            \zeta^m_k\in\{0,1\},\xi^m_k\in\{0,1\},\varphi^m_k\in\{0,1\}, \phi^m_k\in\{0,1\}, \forall k\in[K], m\in[M],\\
            &\, \theta\in C, \theta^0_K = \theta^1_K, \theta^0_k\leq \theta^1_k, \forall k\in[K-1].
        \end{align}
    \end{subequations}
Proposition~\ref{prop:milp-equivalence-symmetric} below connects the MILP formulation in \eqref{equ:milp-symmetric-design} to the SAA formulation \eqref{equ:SAA-symmetric-design}. 
\begin{proposition}\label{prop:milp-equivalence-symmetric}
    \sloppy \eqref{equ:milp-symmetric-design}, whenever feasible, has the same optimal value as \eqref{equ:SAA-symmetric-design}. Furthermore, if $(\hat\theta^0_k, \hat\theta^1_k), \hat w_k^m, \hat\tau^m_k, \hat\rho^m_k, \hat\gamma^m_k, \hat\lambda^m_k, \hat\pi^m_k, \hat\kappa^m_k, \hat\zeta^m_k,$ $ \hat\xi^m_k, \hat\varphi^m_k, \hat\phi^m_k$ for $k\in[K], m\in[M]$ is an optimal solution to \eqref{equ:milp-symmetric-design}, then $(\hat\theta^0_k, \hat\theta^1_k)^K_{k=1}$ is an optimal solution to \eqref{equ:SAA-symmetric-design}.
\end{proposition}

\subsubsection{Finite Sample Convergence Analysis}\label{sec:convergence-symmetric}
In this section, we study the finite sample convergence analysis between the SAA problem~\eqref{equ:SAA-symmetric-design} and the original problem~\eqref{equ:symmetric-design-1} for $z$-tests. 

Let $\Theta_\f(\alpha, \beta)$ be the feasible set of problem~\eqref{equ:symmetric-design-1} and $\Theta^\f_M(\alpha, \beta)$ be the feasible set of problem~\eqref{equ:SAA-symmetric-design}, respectively. Let $f_\f(\theta) = n_1 + \sum^K_{k=2}n_k\Pr(S_i\in (\theta^0_i, \theta^1_i), \forall i\in[k-1]|\tilde H)$ and $v^*_\f(\alpha, \beta)$ be the optimal value of \eqref{equ:symmetric-design-1}. Let $\hat\theta^\f_M(\alpha, \beta)$ be any optimal solution to SAA problem~\eqref{equ:SAA-symmetric-design}.

\begin{proposition}\label{prop:convergence-symmetric}

(i) With probability at least $1-\eta$, we have 
\[
\Theta_\f(\alpha - \epsilon_\f(\eta, M), \beta - \epsilon_\f(\eta, M))\subseteq \Theta^\f_M(\alpha, \beta)\subseteq \Theta_\f(\alpha + \epsilon_\f(\eta, M), \beta+ \epsilon_\f(\eta, M)),
\]
with $\epsilon_\f(\eta, M) = \sqrt{\big(c^\f_1(K)\log(1/\eta) + c^\f_2(K)\log(M) + c^\f_3(K)\big)/M}$, where $c^\f_1(K), c^\f_2(K)$ and $c^\f_3(K)$ are some constants that depend on $K$ and not $M$. 

% Alternatively, for all $\epsilon > 0$, we have $
% \Pr(\Theta_\f(\alpha - \epsilon, \beta - \epsilon)\subseteq \Theta^\f_M(\alpha, \beta)\subseteq \Theta_\f(\alpha + \epsilon, \beta+ \epsilon))\geq 1-\tilde c^\f_1(K)M^{\tilde c^\f_2(K)}\exp{(-\tilde c^\f_3(K)M\epsilon^2)} 
% $,
% where $\tilde c^\f_1(K), \tilde c^\f_2(K)$ and $\tilde c^\f_3(K)$ are some constants that depend on $K$ and not $M$. 

(ii) Suppose MFCQ holds at any optimal solution of problem~\eqref{equ:symmetric-design-1}. There exists some threshold $\bar M_\f> 0$ such that for all $M\geq \bar M_\f$, with probability at least $1-\eta$, we have $\Theta^\f_M(\alpha, \beta)\neq \emptyset$ and 
\[
|f_\f(\hat\theta^\f_M(\alpha, \beta)) - v^*_\f(\alpha, \beta)|\leq \sqrt{\big(c^\f_{\mv,1}(K) + c^\f_{\mv,2}(K)\log(1/\eta) + c^\f_{\mv,3}(K)\log(M)\big)/M}, 
\]
where $c^\f_{\mathrm{v,1}}(K), c^\f_{\mathrm{v,2}}(K)$ and $c^\f_{\mathrm{v,3}}(K)$ are positive constants that depend on $K$ and not $M$. 

% Alternatively, for all $M\geq \bar M_\f$ and any $\epsilon>0$,
% $
% \Pr(|f_\f(\hat\theta^\f_M(\alpha, \beta)) - v_\f^*(\alpha, \beta)|\leq \epsilon) \geq 1-\tilde c^\f_{\mathrm{v},1}(K,n)M^{\tilde c^\f_{\mathrm{v},2}(K,n)}\exp(-\tilde c_{\mathrm{v},3}(K,n)M\epsilon^2)
% $
% where $\tilde c^\f_{\mathrm{v,1}}(K, n), \tilde c^\f_{\mathrm{v,2}}(K, n)$ and $\tilde c^\f_{\mathrm{v,3}}(K, n)$ are positive constants that depend on $K$ and $n$ and not $M$.

\end{proposition}

The above propositions mirror the guarantees in Propositions~\ref{prop:convergence-t},~\ref{prop:pilot-convergence}, and \ref{prop:convergence-generalization}, i.e., our S-MILP approach for the symmetric design problem is also ``equivalent'' to the original problem. For completeness, we also show a similar Lemma as Lemma~\ref{lem:MFCQ} that shows MFCQ for this problem is also a mild assumption, namely that it fails only on a set with zero Lebesgue measure: 
\begin{lemma}\label{lem:MFCQ-symmetric}
    $\{\theta\in\mathbb{R}^{2K}: \text{MFCQ fails at } \theta \text{ for \eqref{equ:symmetric-design-1}}\}$ has a zero Lebesgue measure. 
\end{lemma}
We remark again that Lemma~\ref{lem:MFCQ-symmetric} holds for \eqref{equ:symmetric-design-1} applied to one-sample, two-sample, one-sided, and two-sided z-tests. 

%%%

\section{Additional Materials for Appendix~\ref{sec:futility}}\label{sec-appendix:futility}
\subsection{Proof of Lemma~\ref{lem:reformulation-symmetric design}}

\begin{proof}[Proof of Lemma~\ref{lem:reformulation-symmetric design}]
    It suffices to show that for $H = H_0$ and $H = H_a$, we have $1 - \sum^K_{k=1}\Pr(S_k\leq \theta^0_k, \text{ and } S_i \in (\theta^0_i, \theta^1_i)\text{ for all }i\in[k-1]|H) = \sum^K_{k=1}\Pr(S_k\geq \theta^1_k, \text{ and } S_i \in (\theta^0_i, \theta^1_i)\text{ for all }i\in[k-1]|H)$, where we let $\theta^0_0 = -\infty$ and $\theta^1_0 = \infty$ by default. Below we prove the argument
    \begin{equation}\label{equ:symmetric-design-induction}
        \begin{split}
    &1 - \sum^K_{k=1}\Pr(S_k\leq \theta^0_k, \text{ and } S_i \in (\theta^0_i, \theta^1_i)\text{ for all }i\in[k-1]|H)\\
    = &\sum^{K-1}_{k=1}\Pr(S_k\geq \theta^1_k, \text{ and } S_i \in (\theta^0_i, \theta^1_i)\text{ for all }i\in[k-1]|H)\\
    &\quad\quad\quad\quad+\Pr(S_K\geq \theta^0_K, \text{ and } S_i \in (\theta^0_i, \theta^1_i)\text{ for all }i\in[K-1]|H).
        \end{split}
    \end{equation}
    using induction. Then, by noticing $\theta^0_K = \theta^1_K$ we have the desired result. For $K = 1$, we have that 
    $1 - \Pr(S_1\leq \theta^0_1) = \Pr(S_1 \geq \theta^0_1)$. 
    Thus the argument holds for $K = 1$. Now, suppose the argument \eqref{equ:symmetric-design-induction} holds for $K = K'\geq 2$.
    Then, for $K = K' + 1$, we have
    \begin{align*}
    &1 - \sum^{K'+1}_{k=1}\Pr(S_k\leq \theta^0_k, \text{ and } S_i \in (\theta^0_i, \theta^1_i)\text{ for all }i\in[k-1]|H) \\
    = &1 - \sum^{K'}_{k=1}\Pr(S_k\leq \theta^0_k, \text{ and } S_i \in (\theta^0_i, \theta^1_i)\text{ for all }i\in[k-1]|H) \\
    &\quad\quad\quad\quad- \Pr(S_{K'+1}\leq \theta^0_{K'+1}, \text{ and } S_i \in (\theta^0_i, \theta^1_i)\text{ for all }i\in[K']|H)\\
    =& \sum^{K'-1}_{k=1}\Pr(S_k\geq \theta^1_k, \text{ and } S_i \in (\theta^0_i, \theta^1_i)\text{ for all }i\in[k-1]|H) \\
    &\quad\quad\quad\quad+ \Pr(S_{K'}\geq \theta^0_{K'}, \text{ and } S_i \in (\theta^0_i, \theta^1_i)\text{ for all }i\in[K'-1]|H)\\
    &\quad\quad\quad\quad-  \Pr(S_{K'+1}\leq \theta^0_{K'+1}, \text{ and } S_i \in (\theta^0_i, \theta^1_i)\text{ for all }i\in[K']|H)\\
    =&\sum^{K'-1}_{k=1}\Pr(S_k\geq \theta^1_k, \text{ and } S_i \in (\theta^0_i, \theta^1_i)\text{ for all }i\in[k-1]|H) \\
    &\quad\quad\quad\quad+ \Pr(S_{K'}\geq \theta^1_{K'}, \text{ and } S_i \in (\theta^0_i, \theta^1_i)\text{ for all }i\in[K'-1]|H) \\
    &\quad\quad\quad\quad+ \Pr(S_{K'+1}\geq \theta^0_{K'+1}, \text{ and } S_i \in (\theta^0_i, \theta^1_i)\text{ for all }i\in[K']|H)\\
    = & \sum^{K'}_{k=1}\Pr(S_k\geq \theta^1_k, \text{ and } S_i \in (\theta^0_i, \theta^1_i)\text{ for all }i\in[k-1]|H) \\
    &\quad\quad\quad\quad+ \Pr(S_{K'+1}\geq \theta^0_{K'+1}, \text{ and } S_i \in (\theta^0_i, \theta^1_i)\text{ for all }i\in[K']|H),
    \end{align*}
    where the second equality follows from \eqref{equ:symmetric-design-induction}, and 
    the third equality follows since $\Pr(S_{K'}\geq \theta^0_{K'}, \text{ and } S_i \in (\theta^0_i, \theta^1_i)\text{ for all }i\in[K'-1]|H) = \Pr(S_{K'}\geq \theta^1_{K'}, \text{ and } S_i \in (\theta^0_i, \theta^1_i)\text{ for all }i\in[K'-1]|H) + \Pr(S_{K'}\in(\theta^0_{K'}, \theta^1_{K'}), \text{ and } S_i \in (\theta^0_i, \theta^1_i)\text{ for all }i\in[K'-1]|H)$, and $\Pr(S_{K'}\in(\theta^0_{K'}, \theta^1_{K'}), \text{ and } S_i \in (\theta^0_i, \theta^1_i)\text{ for all }i\in[K'-1]|H) = \Pr(S_{K'+1}\leq \theta^0_{K'+1}, \text{ and } S_i \in (\theta^0_i, \theta^1_i)\text{ for all }i\in[K'-1]|H) + \Pr(S_{K'+1}\geq \theta^0_{K'+1}, \text{ and } S_i \in (\theta^0_i, \theta^1_i)\text{ for all }i\in[K']|H)$.
    Thus the argument holds for $K = K' + 1$. 
    
    Finally, since $S = (S_k)^K_{k=1}$ has a multivariate Normal distribution, we also have that $\Pr(S_1\leq \theta^0_1|H) = \Pr(S_1< \theta^0_1|H)$, and that $\Pr(S_k\leq \theta^0_k, \text{ and } S_i \in (\theta^0_i, \theta^1_i)\text{ for all }i\in[k-1]|H) = \Pr(S_k<\theta^0_k, \text{ and } S_i \in (\theta^0_i, \theta^1_i)\text{ for all }i\in[k-1]|H) = \Pr(S_k\leq\theta^0_k, \text{ and } S_i \in [\theta^0_i, \theta^1_i]\text{ for all }i\in[k-1]|H)$ for $2\leq k\leq K$. 
\end{proof}

\subsection{Proof of Proposition~\ref{prop:milp-equivalence-symmetric}}
\begin{proof}[Proof of Proposition~\ref{prop:milp-equivalence-symmetric}]
    Let $v_{\mathrm{ip}}$ be the optimal value of \eqref{equ:milp-symmetric-design} and let $v_{\mathrm{saa}}$ be the optimal value of \eqref{equ:SAA-symmetric-design}. We first show $v_{\mathrm{ip}}\geq v_{\mathrm{saa}}$ and then $v_{\mathrm{saa}}\geq v_{\mathrm{ip}}$.

    Let $\hat\theta, \hat w_k^m, \hat\tau^m_k, \hat\rho^m_k, \hat\gamma^m_k, \hat\lambda^m_k, \hat\pi^m_k, \hat\kappa^m_k, \hat\zeta^m_k, \hat\xi^m_k, \hat\varphi^m_k, \hat\phi^m_k$ be an optimal solution to \eqref{equ:milp-symmetric-design}. We show that $\hat\theta = (\hat\theta^0_1, \cdots, \hat\theta^0_K, \hat\theta^1_1, \cdots, \hat\theta^1_K)$ is a feasible solution to \eqref{equ:SAA-symmetric-design} and $v_{\mathrm{ip}}\geq v_{\mathrm{saa}}$.
    According to \eqref{equ:milp-symmetric-c1}, we have $\hat w^m_k\leq \I(S^m_k\leq \hat\theta^1_k)$; According to \eqref{equ:milp-symmetric-c2}, $\hat\tau^m_k\leq \I(S^m_k\geq \hat\theta^0_k)$; 
    According to \eqref{equ:milp-symmetric-c2-hat}, we have $\hat \lambda^m_k\leq \I(S^m_k\leq \hat\theta^0_k)$. 
    According to \eqref{equ:milp-symmetric-c3}-\eqref{equ:milp-symmetric-c5}, 
    $\hat\kappa^m_1\leq \hat \lambda^m_1\leq \I(S^m_1\leq \hat\theta^0_1)$ and $\hat\kappa^m_k\leq \hat\lambda^m_k\prod^{k-1}_{i=1}\hat w^m_i\hat \tau^m_i\leq \I(S^m_k\leq \hat\theta^0_k, \text{ and }\hat\theta^0_i\leq S^m_i\leq\hat\theta^1_i,\,\forall i\in[k-1])$ for $2\leq k\leq K$. Then, according to \eqref{equ:milp-symmetric-c6}, $1-\alpha\leq (1/M)\sum^M_{m=1}\sum^K_{k=1}\hat\kappa^m_k\leq (1/M)\sum_{m\in[M]}(\I(S^m_1\leq \hat\theta^0_1) + \sum^K_{k=2}\I(S^m_k\leq \hat\theta^0_k, \text{ and }\hat\theta^0_i\leq S^m_i\leq\hat\theta^1_i,\,\forall i\in[k-1]))$. Thus $\hat\theta$ is feasible to \eqref{equ:saa-symmetric-c1}. 
    
    According to \eqref{equ:milp-symmetric-c7}, $\hat\rho^m_k \geq \I(\hat\theta^1_k > S^m_{a, k})$; According to \eqref{equ:milp-symmetric-c8}, $\hat\gamma^m_k\geq \I(\hat\theta^0_k < S^m_{a,k})$; 
    According to \eqref{equ:milp-symmetric-c8-hat}, $\hat\pi^m_k \geq \I(\hat\theta^0_k > S^m_{a, k})$;
    According to \eqref{equ:milp-symmetric-c9} and \eqref{equ:milp-symmetric-c10}, 
    $\hat\zeta^m_1\geq \hat\pi^m_1 \geq \I(\hat\theta^0_1 > S^m_{a, 1})$ and $\hat\zeta^m_k\geq \I(S^m_{a,k} < \hat\theta^0_k, \text{ and }\hat\theta^0_i< S^m_{a, i} < \hat\theta^1_i, \, \forall i\in[k-1])$ for $2\leq k\leq K$. Then, according to \eqref{equ:milp-symmetric-c11}, $\beta\geq (1/M)\sum^M_{m=1}(\I(\hat\theta^0_1 > S^m_{a, 1}) + \sum^K_{k=2}\I(S^m_{a,k} < \hat\theta^0_k, \text{ and }\hat\theta^0_i< S^m_{a, i} < \hat\theta^1_i, \, \forall i\in[k-1]))$. Thus $\hat\theta$ is feasible to \eqref{equ:saa-symmetric-c2}. 

    Similar to above arguments, according to \eqref{equ:milp-symmetric-c12}-\eqref{equ:milp-symmetric-c14}, we have that $\hat\varphi^m_k \geq \I(S^m_{a', k}<\hat\theta^1_k)$, $\hat\phi^m_k \geq \I(S^m_{a', k} \geq \hat\theta^0_k)$, and $\hat\xi^m_k \geq \I(\hat\theta^0_i < S^m_{a', i} < \hat\theta^1_i,\, \forall i\in[k])$. It thus follows that $v_{\mathrm{ip}} = 
    n_1 + (1/M)\sum^M_{m=1}\sum^K_{k=2}n_k\hat\xi^m_{k-1}\geq
    n_1 + (1/M)\sum^M_{m=1}\sum^K_{k=2}n_k\I(\hat\theta^0_i < S^m_{a', i} < \hat\theta^1_i,\, \forall i\in[k-1])\geq v_\saa$, where the second inequality follows since $\hat\theta$ is feasible to \eqref{equ:SAA-symmetric-design}. 

    \sloppy 
    On the other hand, let $\tilde\theta = (\tilde\theta^0_1, \cdots, \tilde\theta^0_K, \tilde\theta^1_1, \cdots, \tilde\theta^1_K)$ be an optimal solution to \eqref{equ:SAA-symmetric-design}. Let $\tilde w^m_k = \I(S^m_k \leq \tilde\theta^1_k)$, $\tilde\tau^m_k = \I(S^m_k\geq\tilde\theta^0_k)$, $\tilde \lambda^m_k = \I(S^m_k \leq \tilde\theta^0_k)$,
    and $\tilde\kappa^m_k = \tilde \lambda^m_k\prod^{k-1}_{i=1}\tilde w^m_i\tilde\tau^m_i$. Then \eqref{equ:milp-symmetric-c1}-\eqref{equ:milp-symmetric-c5} are feasible. \eqref{equ:milp-symmetric-c6} is feasible since by \eqref{equ:saa-symmetric-c1}, $1-\alpha \leq (1/M)\sum^M_{m=1}\sum^K_{k=1}\I(S^m_k\leq \tilde\theta^0_k, \text{ and } S^m_i \in [\tilde\theta^0_i, \tilde\theta^1_i]\text{ for all }i\in[k-1]) =  (1/M)\sum^M_{m=1}\tilde\kappa^m_k$. 
    Similarly, let $\tilde\rho^m_k = \I(S^m_{a, k} < \tilde\theta^1_k)$, $\tilde\gamma^m_k = \I(S^m_{a,k}>\tilde\theta^0_k)$, 
    $\tilde\pi^m_k = \I(S^m_{a, k} < \tilde\theta^0_k)$,
    and $\tilde\zeta^m_k = \tilde\pi^m_k\prod^{k-1}_{i=1}\tilde\rho^m_i\tilde\gamma^m_i$. Then \eqref{equ:milp-symmetric-c7}-\eqref{equ:milp-symmetric-c10} are feasible. \eqref{equ:milp-symmetric-c11} is feasible since by \eqref{equ:saa-symmetric-c2}, $\beta\geq (1/M)\sum^M_{m=1}\sum^K_{k=1}\I(S^m_{a, k}< \tilde\theta^0_k, \text{ and } S^m_{a,i} \in (\tilde\theta^0_i, \tilde\theta^1_i)\text{ for all }i\in[k-1]) = (1/M)\sum^M_{m=1}\sum^K_{k=1}\tilde\zeta^m_k$. Furthermore let $\tilde\varphi^m_k = \I(S^m_{a', k}< \tilde\theta^1_k)$, $\tilde\phi^m_k = \I(S^m_{a',k}>\tilde\theta^0_k)$, and $\tilde\xi^m_k = \prod^k_{i=1}\tilde\varphi^m_i\tilde\phi^m_i$. Then \eqref{equ:milp-symmetric-c12}-\eqref{equ:milp-symmetric-c14} are feasible. Additionally, $v_\saa = n_1 + (1/M)\sum^M_{m=1}\sum^K_{k=2}n_k\I(S^m_{a', i}\in (\tilde\theta^0_i, \tilde\theta^1_i), \forall i\in[k-1]) = n_1 + (1/M)\sum^M_{m=1}\sum^K_{k=2}n_k\tilde\xi^m_{k-1}\geq v_{\mathrm{ip}}$, where the last inequality follows since $\tilde\theta, \tilde w_k^m, \tilde\tau^m_k, \tilde\rho^m_k, \tilde\gamma^m_k, \tilde\kappa^m_k, \tilde\zeta^m_k, \tilde\lambda^m_k, \tilde\pi^m_k, \tilde\xi^m_k, \tilde\varphi^m_k, \tilde\phi^m_k$ for $k\in[K], m\in[M]$ is a feasible solution to \eqref{equ:milp-symmetric-design}. 
\end{proof}

\subsection{Proof of Proposition~\ref{prop:convergence-symmetric}}

Recall that $v_\F^*(\alpha, \beta)$ and $\theta^*_\F(\alpha, \beta)$ are the optimal value and any optimal solution of problem~\eqref{equ:symmetric-design-1}, respectively. Let $\hat v^\f_M(\alpha, \beta)$ and $ \hat\theta^\f_M(\alpha, \beta)$ be the optimal value and any optimal solution of the SAA problem~\eqref{equ:SAA-symmetric-design}, respectively. $\Theta_\f(\alpha, \beta)$ is the feasible region of true problem \eqref{equ:symmetric-design-1}, and $\Theta^\f_M(\alpha, \beta)$ is the feasible region of the SAA problem \eqref{equ:SAA-symmetric-design}. $S_\f(\alpha, \beta)$ and $S^\f_M(\alpha, \beta)$ are the sets of optimal solutions to the true problem \eqref{equ:symmetric-design-1} and the SAA problem \eqref{equ:SAA-symmetric-design} respectively. 

In this subsection, with some abuse of notation, we let $f(\theta) = n_1 + \sum^K_{k=2}n_k\Pr(S_i\in (\theta^0_i, \theta^1_i), \forall i\in[k-1]|\tilde H)$ be the objective, $g_0(\theta) = \Pr(S_1\leq\theta^0_1|H_0) + \sum^K_{k=2}\Pr(S_k\leq \theta^0_k, \text{ and } S_i \in (\theta^0_i, \theta^1_i)\text{ for all }i\in[k-1]|H_0)$ be the LHS of the first constraint and $g_a(\theta) = \Pr(S_1\leq\theta^0_1|H_a) + \sum^K_{k=2}\Pr(S_k\leq \theta^0_k, \text{ and } S_i \in (\theta^0_i, \theta^1_i)\text{ for all }i\in[k-1]|H_a)$ be the LHS of the second constraint of \eqref{equ:symmetric-design-1}, respectively. Furthermore, we let $\hat f_M(\theta) = n_1 + (1/M)\sum^M_{m=1}\sum^K_{k=2}n_k\I(S^m_{a', i}\in (\theta^0_i, \theta^1_i), \forall i\in[k-1])$ be the objective of the SAA problem \eqref{equ:SAA-symmetric-design}, $\hat g_{0,M}(\theta) = (1/M)\sum^M_{m=1}
\big(\I(S^m_1\leq \theta^0_1) + \sum^K_{k=2}\I(S^m_k\leq \theta^0_k, \text{ and } S^m_i \in [\theta^0_i, \theta^1_i]\text{ for all }i\in[k-1])\big)$ and $\hat g_{a,M}(\theta) = (1/M)\sum^M_{m=1}
\big(\I(S^m_{a,1}< \theta^0_1) + \sum^K_{k=2}\I(S^m_{a, k}< \theta^0_k, \text{ and } S^m_{a,i} \in (\theta^0_i, \theta^1_i)\text{ for all }i\in[k-1])\big)$ be the LHS of the first constraint and second constraint of SAA problem \eqref{equ:SAA-symmetric-design}, respectively.

We continue to use $\hat C(\Delta)$ as a discretization of $C$ so that for any $\hat\theta\in C$, there exists some $\tilde\theta\in \hat C(\Delta)$ such that $\Vert\hat\theta - \tilde\theta\Vert_\infty\leq \Delta$. Then $|\hat C(\Delta)|\leq ((\bar\theta - \underline{\theta})/\Delta)^{2K}$. Let $\epsilon_{1,\F}(\eta, M) = K\sqrt{\log(4|\hat C(\Delta)|/\eta)/(2M)}$. 

We first present Lemmas~\ref{lem:SAA-feasibility-futility} (as an analog of Lemma~\ref{lem:SAA-feasibility}), \ref{lem:finite-feasibility-futility} (as an analog of Lemma~\ref{lem:finite-feasibility}), \ref{lem:lipshitz-continuity-futility} (as an analog of Lemma~\ref{lem:lipshitz-continuity}), \ref{lem:SAA-constraint-bound-futility} (as an analog of Lemma~\ref{lem:SAA-constraint-bound}), and \ref{lem:Delta-net-futility} (as an analog of Lemma~\ref{lem:Delta-net}). 

\begin{lemma}\label{lem:SAA-feasibility-futility}
    Suppose there exists some $\bar\epsilon> 0$ such that  $\Theta_\f(\alpha - \bar\epsilon, \beta -\bar\epsilon)\neq \emptyset$. Then for $M\geq K^2\log{(2/\eta)}/(2\bar\epsilon^2)$, with probability at least $1-\eta$,  $\Theta^\f_M(\alpha, \beta)\neq \emptyset$. 
\end{lemma}
\begin{proof}[Proof of Lemma~\ref{lem:SAA-feasibility-futility}]
    As $\Theta_\f(\alpha - \bar\epsilon, \beta -\bar\epsilon)\neq\emptyset$, we let $\hat\theta$ be arbitrary such that
    $\hat\theta\in \Theta_\f(\alpha - \bar\epsilon, \beta -\bar\epsilon)$. Then,
    \begin{equation*}
        \begin{split}
    &\Pr\big(\Theta^\f_M(\alpha, \beta)\neq\emptyset\big)\\
    \geq &\Pr\big(\hat\theta\in \Theta^\f_M(\alpha, \beta)\big)\\
    = & \Pr\big(\hat g_{0,M}(\hat\theta)\geq 1-\alpha \text{ and }\hat g_{a, M}(\hat\theta)\leq \beta\big)\\
    \geq & \Pr\big(\hat g_{0,M}(\hat\theta)\geq g_0(\hat\theta)-\bar\epsilon \text{ and }\hat g_{a,M}(\hat\theta)\leq g_a(\hat\theta)+\bar\epsilon\big)\\
    \geq & 1- \Pr\big(\hat g_{0,M}(\hat\theta)\leq g_0(\hat\theta)-\bar\epsilon\big) - \Pr\big(\hat g_{a, M}(\hat\theta)\geq g_a(\hat\theta)+\bar\epsilon\big)\\
    \geq & 1-2\exp(-2M\bar\epsilon^2/K^2), 
        \end{split}
    \end{equation*}
\sloppy 
where the last inequality follows from Hoeffding's inequality, and that $g_0(\theta), g_a(\theta)\in [0,K]$. The argument then follows by letting $2\exp(-2M\bar\epsilon^2/K^2) = \eta$. 
\end{proof}

Lemmas~\ref{lem:finite-feasibility-futility} and \ref{lem:lipshitz-continuity-futility} follow similarly as Lemmas~\ref{lem:finite-feasibility} and \ref{lem:lipshitz-continuity} respectively. We thus omit their proofs. 

\begin{lemma}\label{lem:finite-feasibility-futility}
    With probability at least $1 - \eta$, we have that $\Theta_\f(\alpha - \epsilon_{1,\F}(\eta, M), \beta -\epsilon_{1,\F}(\eta, M))\cap \hat C(\Delta)\subseteq \Theta^\f_M(\alpha, \beta)\cap \hat C(\Delta)\subseteq \Theta_\f(\alpha + \epsilon_{1,\F}(\eta, M), \beta+ \epsilon_{1,\F}(\eta, M))\cap \hat C(\Delta)$. 
\end{lemma}
\begin{lemma}\label{lem:lipshitz-continuity-futility}
There exist positive constants $l^\F_f, l^\F_{g_0}$ and $l^\F_{g_a}$ such that the following holds:
    \begin{align*}
        |f(\hat\theta) - f(\tilde\theta)|\leq l^\F_f\Vert\hat\theta - \tilde\theta\Vert_1, \forall \hat\theta, \tilde\theta\in C,\\
        |g_0(\hat\theta) - g_0(\tilde\theta)|\leq l^\F_{g_0}\Vert\hat\theta - \tilde\theta\Vert_1, \forall \hat\theta, \tilde\theta\in C,\\
        |g_{a}(\hat\theta) - g_a(\tilde\theta)|\leq l^\F_{g_a}\Vert\hat\theta - \tilde\theta\Vert_1, \forall \hat\theta, \tilde\theta\in C. 
    \end{align*}
\end{lemma}
\begin{lemma}\label{lem:SAA-constraint-bound-futility}
There exist constants $c^\F_{g_0}, c^\F_{g_a}, c^\F_f$ such that the following arguments hold:

(i) With probability at least $1-\eta$, 
\begin{equation*}
    \begin{split}
        \max\Big\{\hat g_{0,M}(\tilde\theta) - \min_{\theta: \Vert\theta - \tilde\theta\Vert_\infty\leq\Delta}\hat g_{0,M}(\theta), 
    \max_{\theta: \Vert\theta - \tilde\theta\Vert_\infty\leq\Delta}\hat g_{0,M}(\theta) - \hat g_{0,M}(\tilde\theta)\Big\}\\
    \leq c^\F_{g_0}K^2\Delta + 2K^2\epsilon_{1,\F}(4\eta/K, M), \,\forall \tilde\theta\in \hat C(\Delta).
    \end{split}
\end{equation*}

(ii) With probability at least $1-\eta$, 
\begin{equation*}
    \begin{split}
        \max\Big\{\hat g_{a,M}(\tilde\theta) - \min_{\theta: \Vert\theta - \tilde\theta\Vert_\infty\leq\Delta}\hat g_{a,M}(\theta), \max_{\theta: \Vert\theta - \tilde\theta\Vert_\infty\leq\Delta}\hat g_{a,M}(\theta) - \hat g_{a,M}(\tilde\theta)\Big\}\\
\leq c^\F_{g_a}K^2\Delta + 2K^2\epsilon_{1,\F}(4\eta/K, M),\, \forall \tilde\theta\in \hat C(\Delta). 
    \end{split}
\end{equation*}

(iii) With probability at least $1-\eta$,  
\begin{align*}
    \max\{\hat f_M(\tilde\theta) - \min_{\theta: \Vert\theta - \tilde\theta\Vert_\infty\leq\Delta}\hat f_M(\theta), \max_{\theta: \Vert\theta - \tilde\theta\Vert_\infty\leq\Delta}\hat f_M(\theta) - \hat f_M(\tilde\theta)\}\\
    \leq c^\F_fK^2\Delta + \sum^K_{k=2}n_k\epsilon_{1, \F}(4\eta/(K-1), M)\, \forall \tilde\theta\in \hat C(\Delta).
\end{align*}
\end{lemma}
\begin{proof}[Proof of Lemma~\ref{lem:SAA-constraint-bound-futility}]
(i) For convenience we let $\tilde\theta^0_0 = \theta^0_0 = -\infty$ and $\tilde\theta^1_0 = \theta^1_0 = \infty$.
For all $\theta, \tilde\theta\in C$ and $\Vert\theta - \tilde\theta\Vert_\infty\leq \Delta$, we have that 
\begin{equation*}
    \begin{split}
        &\hat g_{0,M}(\tilde\theta) - \hat g_{0,M}(\theta)\\
    \leq &\frac{1}{M}\sum^M_{m=1}\sum^K_{k=1}
    \big(\I(S^m_k\leq \max\{\tilde\theta^0_k, \theta^0_k\})\I(S^m_i \in [\min\{\tilde\theta^0_i, \theta^0_i\}, \max\{\tilde\theta^1_i, \theta^1_i\}], \forall i\in[k-1])\\
    &-\I(S^m_k\leq \min\{\tilde\theta^0_k, \theta^0_k\})\I(S^m_i \in [\max\{\tilde\theta^0_i, \theta^0_i\}, \min\{\tilde\theta^1_i, \theta^1_i\}], \forall i\in[k-1])
    \big)
    \\
    \leq &\frac{1}{M}\sum^M_{m=1}\sum^K_{k=1}
    \big(\I(\min\{\tilde\theta^0_k, \theta^0_k\}<S^m_k\leq \max\{\tilde\theta^0_k, \theta^0_k\}) \\
    &+ \I(S^m_i\in[\min\{\tilde\theta^0_i, \theta^0_i\}, \max\{\tilde\theta^0_i, \theta^0_i\}]\cup [\min\{\tilde\theta^1_i, \theta^1_i\}, \max\{\tilde\theta^1_i, \theta^1_i\}] \text{ for some }i\in[k-1])\big)\\
    \leq &
    \frac{1}{M}\sum^M_{m=1}\sum^K_{k=1}\sum^k_{i=1}
    \big(\I(\tilde\theta^1_i - \Delta<S^m_i\leq \tilde\theta^1_i + \Delta) + \I(\tilde\theta^0_i - \Delta<S^m_i\leq \tilde\theta^0_i + \Delta)\big).
    \end{split}
\end{equation*}
Since the above inequality holds for all $\theta$ with $\Vert\theta-\tilde\theta\Vert_\infty\leq\Delta$, it holds for $\hat g_{0,M}(\tilde\theta) - \min_{\theta: \Vert\theta-\tilde\theta\Vert_\infty\leq \Delta}\hat g_{0,M}(\theta)$ and $\max_{\theta: \Vert\theta-\tilde\theta\Vert_\infty\leq \Delta}\hat g_{0,M}(\theta) - \hat g_{0,M}(\tilde\theta)$ as well. Thus we have
\begin{equation*}
    \begin{split}
\max\{\hat g_{0,M}(\tilde\theta) - \min_{\theta: \Vert\theta-\tilde\theta\Vert_\infty\leq \Delta}\hat g_{0,M}(\theta), \max_{\theta: \Vert\theta-\tilde\theta\Vert_\infty\leq \Delta}\hat g_{0,M}(\theta) - \hat g_{0,M}(\tilde\theta)\}\\
\leq \frac{1}{M}\sum^M_{m=1}\sum^K_{k=1}\sum^k_{i=1}
    \big(\I(\tilde\theta^1_i - \Delta<S^m_i\leq \tilde\theta^1_i + \Delta) + \I(\tilde\theta^0_i - \Delta<S^m_i\leq \tilde\theta^0_i + \Delta)\big).
    \end{split}
\end{equation*}
Let $p_{0,k}(x)$ and $p_{a',k}(x)$ be the densities of the marginal distributions of $S_k$ at $x$ under $H_0$ and $\tilde H$ respectively, and let $c^\f_{g_0} = 4\max_{k\in[K]}\max_{x\in[\underline{\theta}, \overline{\theta}]}p_{0,k}(x)$. 
Then, for all $\epsilon > 0$, 
    \begin{align*}
    &\Pr(\max\{\hat g_{0,M}(\tilde\theta) - \min_{\theta: \Vert\theta - \tilde\theta\Vert_\infty\leq\Delta}\hat g_{0,M}(\theta), 
    \max_{\theta: \Vert\theta - \tilde\theta\Vert_\infty\leq\Delta}\hat g_{0,M}(\theta) - \hat g_{0,M}(\tilde\theta)\} \\
    &\quad
    \leq c^\F_{g_0}K^2\Delta + \epsilon, \forall \tilde\theta\in \hat C(\Delta))\\
    \geq &\Pr((1/M)\sum^M_{m=1}\sum^K_{k=1}\sum^k_{i=1}
    \big(\I(\tilde\theta^1_i - \Delta<S^m_i\leq \tilde\theta^1_i + \Delta) + \I(\tilde\theta^0_i - \Delta<S^m_i\leq \tilde\theta^0_i + \Delta)\big)\\
    &\quad\leq c^\F_{g_0}K^2\Delta + \epsilon, \forall \tilde\theta\in \hat C(\Delta))\\
    \geq& \Pr((1/M)\sum^M_{m=1}\sum^K_{k=1}\sum^k_{i=1}
    \big(\I(\tilde\theta^1_i - \Delta<S^m_i\leq \tilde\theta^1_i + \Delta) + \I(\tilde\theta^0_i - \Delta<S^m_i\leq \tilde\theta^0_i + \Delta)\big) \\
    &\quad
    \leq \sum^K_{k=1}\sum^k_{i=1}
    \big(\Pr(\tilde\theta^1_i - \Delta<S^m_i\leq \tilde\theta^1_i + \Delta) + \Pr(\tilde\theta^0_i - \Delta<S^m_i\leq \tilde\theta^0_i + \Delta)\big) + \epsilon,\,\forall\tilde\theta\in\hat C(\Delta))\\
    \geq& \Pr((1/M)\sum^M_{m=1}
    \big(\I(\tilde\theta^1_k - \Delta<S^m_k\leq \tilde\theta^1_k + \Delta) + \I(\tilde\theta^0_k - \Delta<S^m_k\leq \tilde\theta^0_k + \Delta)\big) \\
    &\quad
    \leq 
    \Pr(\tilde\theta^1_k - \Delta<S^m_k\leq \tilde\theta^1_k + \Delta) + \Pr(\tilde\theta^0_k - \Delta<S^m_k\leq \tilde\theta^0_k + \Delta) + \epsilon/K^2,\,\forall\tilde\theta\in\hat C(\Delta), k\in[K])\\
    \geq &1-K|\hat C(\Delta)|\exp(-M\epsilon^2/(2K^4)),
    \end{align*}
    where the second inequality follows since $\sum^K_{k=1}\sum^k_{i=1}
    \big(\Pr(\tilde\theta^1_i - \Delta<S^m_i\leq \tilde\theta^1_i + \Delta) + \Pr(\tilde\theta^0_i - \Delta<S^m_i\leq \tilde\theta^0_i + \Delta)\big)\leq  c^\F_{g_0}K^2\Delta$. The last inequality follows from Hoeffding's inequality. Letting $\epsilon = 2K^2\epsilon_{1,\F}(4\eta/K, M)$, we have the desired result. 

(ii) can be proved similarly as (i). We thus omit its proof.

(iii) We have that $\hat f_M(\theta) = n_1 + (1/M)\sum^M_{m=1}\sum^K_{k=2}n_k\I(S^m_{a', i}\in(\theta^0_i, \theta^1_i), \forall i\in[k-1])$. Then, for all $\theta, \tilde\theta\in C$ and $\Vert\theta - \tilde\theta\Vert_\infty\leq \Delta$, we have that
\begin{align*}
&\hat f_M(\tilde\theta) - \hat f_M(\theta)\\
\leq &(1/M)\sum^M_{m=1}\sum^K_{k=2}n_k
\big(\I(S^m_{a',i} \in(\min(\theta^0_i, \tilde\theta^0_i), \max(\tilde\theta^1_i, \theta^1_i)), \forall i\in[k-1]) \\
&\quad\quad\quad\quad\quad\quad
- \I(S^m_{a',i} \in(\max(\theta^0_i, \tilde\theta^0_i), \min(\tilde\theta^1_i, \theta^1_i)), \forall i\in[k-1])
\big)\\
\leq &
(1/M)\sum^M_{m=1}\sum^K_{k=2}n_k
\big(\I(S^m_{a',i} \in [\min(\theta^0_i, \tilde\theta^0_i), \max(\theta^0_i, \tilde\theta^0_i)] \text{ or }S^m_{a',i} \in[\min(\theta^1_i, \tilde\theta^1_i),\max(\theta^1_i, \tilde\theta^1_i)] \\
&\quad\quad\quad\quad\quad\quad\text{ for some } i\in[k-1])
\big)\\
\leq & (1/M)\sum^M_{m=1}\sum^K_{k=2}n_k\I(\tilde\theta^0_i - \Delta \leq S^m_{a',i}\leq \tilde\theta^0_i + \Delta \text{ or }\tilde\theta^1_i - \Delta \leq S^m_{a',i}\leq \tilde\theta^1_i + \Delta\text{ for some }i\in[k-1]). 
\end{align*}
It thus follows that $\hat f_M(\tilde\theta) - \min_{\theta: \Vert\theta - \tilde\theta\Vert_\infty\leq \Delta}\hat f_M(\theta)\leq (1/M)\sum^M_{m=1}\sum^K_{k=2}n_k\I(\tilde\theta^0_i - \Delta \leq S^m_{a',i}\leq \tilde\theta^0_i + \Delta \text{ or }\tilde\theta^1_i - \Delta \leq S^m_{a',i}\leq \tilde\theta^1_i + \Delta\text{ for some }i\in[k-1])$ and $\max_{\theta: \Vert\theta - \tilde\theta\Vert_\infty\leq \Delta}\hat f_M(\theta) - \hat f_M(\tilde\theta)\leq (1/M)\sum^M_{m=1}\sum^K_{k=2}n_k\I(\tilde\theta^0_i - \Delta \leq S^m_{a',i}\leq \tilde\theta^0_i + \Delta \text{ or }\tilde\theta^1_i - \Delta \leq S^m_{a',i}\leq \tilde\theta^1_i + \Delta\text{ for some }i\in[k-1])$. Let $c^\f_f = 4(\max_{k\in[K]}n_k)\max_{k\in[K]}\max_{x\in[\underline{\theta}, \overline{\theta}], k\in[K]}p_{a',k}(x)$.
Then, for any $\epsilon_2, \cdots, \epsilon_K>0$, 
 \begin{align*}
    &\Pr(\max\{\hat f_M(\tilde\theta) - \min_{\theta: \Vert\theta - \tilde\theta\Vert_\infty\leq\Delta}\hat f_M(\theta), 
    \max_{\theta: \Vert\theta - \tilde\theta\Vert_\infty\leq\Delta}\hat f_M(\theta) - \hat f_M(\tilde\theta)\} \\
    &\quad\quad
    \leq c^\F_fK^2\Delta + \sum^K_{k=2}\epsilon_k, \forall \tilde\theta\in \hat C(\Delta))\\
    \geq 
    &\Pr((1/M)\sum^M_{m=1}\sum^K_{k=2}n_k\I(\tilde\theta^0_i - \Delta \leq S^m_{a',i}\leq \tilde\theta^0_i + \Delta \text{ or }\tilde\theta^1_i - \Delta \leq S^m_{a',i}\leq \tilde\theta^1_i + \Delta\text{ for some }i\in[k-1]) \\
    &\quad\quad
    \leq c^\f_f K^2\Delta + \sum^K_{k=2}\epsilon_k, \forall \tilde\theta\in \hat C(\Delta))\\
    \geq 
    &\Pr((1/M)\sum^M_{m=1}n_k\I(\tilde\theta^0_i - \Delta \leq S^m_{a',i}\leq \tilde\theta^0_i + \Delta \text{ or }\tilde\theta^1_i - \Delta \leq S^m_{a',i}\leq \tilde\theta^1_i + \Delta\text{ for some }i\in[k-1]) \\
    &\quad\quad
    \leq c^\f_f (k-1)\Delta + \epsilon_k, \forall \tilde\theta\in \hat C(\Delta), 2\leq k\leq K)\\
    \geq&
    \Pr((1/M)\sum^M_{m=1}n_k\I(\tilde\theta^0_i - \Delta \leq S^m_{a',i}\leq \tilde\theta^0_i + \Delta \text{ or }\tilde\theta^1_i - \Delta \leq S^m_{a',i}\leq \tilde\theta^1_i + \Delta\text{ for some }i\in[k-1]) \\
    &\quad\quad
    \leq n_k\Pr(\tilde\theta^0_i - \Delta \leq S^m_{a',i}\leq \tilde\theta^0_i + \Delta \text{ or }\tilde\theta^1_i - \Delta \leq S^m_{a',i}\leq \tilde\theta^1_i + \Delta\text{ for some }i\in[k-1]) + \epsilon_k,\\
    &\quad\quad\quad\quad\quad
    \forall\tilde\theta\in\hat C(\Delta), 2\leq k\leq K)\\
    \geq &1-(K-1)|\hat C(\Delta)|\exp(-2M\epsilon_k^2/n^2_k), 
    \end{align*}
    where the third inequality follows since $n_k\Pr(\tilde\theta^0_i - \Delta \leq S^m_{a',i}\leq \tilde\theta^0_i + \Delta \text{ or }\tilde\theta^1_i - \Delta \leq S^m_{a',i}\leq \tilde\theta^1_i + \Delta\text{ for some }i\in[k-1])\leq n_k\sum^{k-1}_{i=1}\Pr(\tilde\theta^0_i - \Delta \leq S^m_{a',i}\leq \tilde\theta^0_i + \Delta \text{ or }\tilde\theta^1_i - \Delta \leq S^m_{a',i}\leq \tilde\theta^1_i + \Delta)\leq c^\f_f (k-1)\Delta$. The last inequality follows from Hoeffding's inequality. Letting $\epsilon_k = n_k\epsilon_{1,\f}(4\eta/(K-1), M)$, we have the desired result. 

\end{proof}
Lemma~\ref{lem:Delta-net-futility} directly follows from Lemma~\ref{lem:SAA-constraint-bound-futility} and the proof of Lemma~\ref{lem:Delta-net}. We thus state it without proof. 
\begin{lemma}\label{lem:Delta-net-futility}
    (i) With probability at least $1-\eta - \Pr(\Theta^\f_M(\alpha, \beta)=\emptyset)$, $\Theta^\f_M(\alpha, \beta)\neq\emptyset$, and for all $\hat\theta\in \Theta^\f_M(\alpha, \beta)$, we can find $\tilde\theta\in \Theta^\f_M(\alpha + c^\F_{g_0}K^2\Delta + 2K^2\epsilon_{1,\f}(2\eta/K, M), \beta+ c^\F_{g_a}K^2\Delta + 2K^2\epsilon_{1,\f}(2\eta/K, M))\cap \hat C(\Delta)$ such that $\Vert \hat\theta - \tilde\theta\Vert_\infty\leq \Delta$. 

    (ii) Suppose $\Theta_\f(\alpha, \beta)\neq\emptyset$. Then for all $\hat\theta\in \Theta_\f(\alpha, \beta)$, we can find $\tilde\theta\in \Theta_\f(\alpha + l^\F_{g_0}K\Delta, \beta+ l^\F_{g_a}K\Delta)\cap \hat C(\Delta)$ such that $\Vert \hat\theta - \tilde\theta\Vert_\infty\leq \Delta$.
\end{lemma}
Then, Proposition~\ref{prop:convergence-symmetric} can be proved similarly as Propositions~\ref{prop:feasible-set} and \ref{prop:objective-bound} based on Lemmas~\ref{lem:finite-feasibility-futility}-\ref{lem:Delta-net-futility}. We omit its proof. 

% \begin{proposition}\label{prop:feasible-set-futility}

% (i) With probability at least $1-\eta$, 
% $$
% \Theta(\alpha - \epsilon_\F(\eta, M), \beta -\epsilon_\F(\eta, M)\subseteq \Theta_M(\alpha, \beta)\subseteq \Theta(\alpha + \epsilon_\F(\eta, M), \beta+ \epsilon_\F(\eta, M)),
% $$
% where $\epsilon_\F(\eta, M) = 2K\sqrt{(2\log(1/\eta) + 2K\log(M) + c_{\mathrm{F}}(K))/M}$, where $c_{\mathrm{F}}(K)$ is some constant that depends on $K$ but not $M$. 

% (ii) For all $\epsilon > 0$, we have $$
% \Pr(\Theta(\alpha - \epsilon, \beta -\epsilon)\subseteq \Theta_M(\alpha, \beta)\subseteq \Theta(\alpha + \epsilon, \beta+ \epsilon))\geq 1-\tilde c_{\mathrm{F}}(K)M^K\exp{(-M\epsilon^2/8)}. 
% $$
% where $\tilde c_{\mathrm{F}}(K)$ is some constant that depends on $K$ and not $M$. 

% \end{proposition}

% \begin{proposition}\label{prop:confidence-interval-futility}
% With probability at least $1-\eta$, 
% \begin{equation}
%     v^*(\alpha, \beta)\in\bigg[f\Big(\hat\theta_M(\alpha + \epsilon_\F(\eta, M), \beta+ \epsilon_\F(\eta, M))\Big) - \epsilon_f(\eta, M), f\Big(\hat\theta_M(\alpha - \epsilon(\eta, M), \beta -\epsilon(\eta, M))\Big)\bigg],
% \end{equation}
% where $\epsilon_f(\eta, M) = (\tilde c_1(K, n) + \tilde c_2(K, n)\log(1/\eta) + \tilde c_3(K, n)\log(M))/\sqrt{M}$ and $\tilde c_1(K, n), \tilde c_2(K, n)$ and $\tilde c_3(K, n)$ are constants that depend on $K$ and $n$ but not $M$. 
% \end{proposition}

\subsection{Proof of Lemma~\ref{lem:MFCQ-symmetric}}
We notice that the set $\{\theta\in\mathbb{R}^{2K}: \theta^0_k = \underline{\theta}, \text{ or }\theta^0_k = \bar\theta, \text{ or }\theta^1_k = \underline{\theta}, \text{ or }\theta^1_k = \bar\theta, \text{ or }\theta^0_k = \theta^1_k \text{ for some }k\in[K]\}$ has a zero Lebesgue measure. 
It thus suffices to show that $\{\theta\in\mathbb{R}^{2K}: \nabla g_0(\theta)\text{ and }\nabla 
g_a(\theta) \text{ are linearly dependent}\}$ has a zero Lebesgue measure. 

We first consider the one-sided $z$ tests (either one-sample or two-sample), and show that there exists some $\theta\in\mathbb{R}^{2K}$ such that $\nabla g_0(\theta)$ and $\nabla g_a(\theta)$ are linearly independent. Under $H_0$, $S = (S_k)^K_{k=1}$ has a multivariate normal distribution with mean $0$ and a covariance $\Sigma$ whose $(k,k')$-th element $\Sigma_{k, k'}< 1$ for $k\neq k'$, and $\Sigma_{k, k} = 1$. Under $H_a$, $S$ has a multivariate normal distribution with mean $\mu = (\mu_k)^K_{k=1}$ and covariance $\Sigma$. 
Let $p_0$ and $p_{a,k}$ be the density functions of the marginal distributions of $S_k$ under $H_0$ and $H_a$. Let $S_{1:k} = (S_i)^k_{i=1}$ and $\theta_{1:k} = (\theta_i)^k_{i=1}$. Then we have, 
\begin{equation*}
    \begin{split}
    \partial g_0(\theta)/\partial \theta^0_K = p_0(\theta^0_K)\Pr(S_{1:K-1}\in(\theta^0_{1:K-1}, \theta^1_{1:K-1})|S_K = \theta^0_K, H_0),
    \\
    \partial g_0(\theta)/\partial \theta^1_{K-1} = p_0(\theta^1_{K-1})\Pr(S_K\leq \theta^0_K, \text{ and }S_{1:K-2}\in(\theta^0_{1:K-2}, \theta^1_{1:K-2})|S_{K-1} = \theta^1_{K-1}, H_0), 
    \\
    \partial g_a(\theta)/\partial \theta^0_K = p_{a, K}(\theta^0_K)\Pr(S_{1:K-1}\in(\theta^0_{1:K-1}, \theta^1_{1:K-1})|S_K = \theta^0_K, H_a),
    \\
    \partial g_a(\theta)/\partial \theta^1_{K-1} = p_{a,K-1}(\theta^1_{K-1})\Pr(S_K\leq \theta^0_K, \text{ and }S_{1:K-2}\in(\theta^0_{1:K-2}, \theta^1_{1:K-2})|S_{K-1} = \theta^1_{K-1}, H_a), 
    \end{split}
\end{equation*}

Let $\theta(x) = (\theta^0_1(x), \cdots, \theta^0_K(x), \theta^1_1(x), \cdots, \theta^1_K(x))$ with $\theta^0_k(x) = -x$ for all $k\in[K-1]$ and $\theta^0_K(x) = \theta^1_k(x) = x$ for all $k\in[K]$. We let $l_0(\theta) = (\partial g_0(\theta)/\partial \theta^1_{K-1}, \partial g_0(\theta)/\partial \theta^0_K)^\top$ and $l_a(\theta) = (\partial g_a(\theta)/\partial \theta^1_{K-1}, \partial g_a(\theta)/\partial \theta^0_K)^\top$. It is then sufficient to show that there exists some $x$ such that $l_0(\theta(x))$ and $l_a(\theta(x))$ are linearly independent. Then, since $\theta^0_K(x) = \theta^1_{K-1}(x) = x$, 
\begin{align*}
    &\lim_{x\to\infty}\frac{\partial g_0(\theta(x))/\partial \theta^0_K}{\partial g_0(\theta(x))/\partial \theta^1_{K-1}} \\
= &\lim_{x\to\infty}
\frac{\Pr(S_{1:K-1}\in(\theta^0_{1:K-1}(x), \theta^1_{1:K-1}(x))|S_K = \theta^0_K(x), H_0)}
{\Pr(S_K\leq \theta^0_K(x), \text{ and }S_{1:K-2}\in(\theta^0_{1:K-2}(x), \theta^1_{1:K-2}(x))|S_{K-1} = \theta^1_{K-1}(x), H_0)} \\
= &1,
\end{align*}
where the second equality follows since for $S\sim N(\mu, \Sigma)$ where $\mu = (\mu_k)^K_{k=1}$, $\Sigma_{k,k} = 1$ and $\Sigma_{k,k'}< 1$ for $k\neq k'$,  
$S_{-k}|S_k = x$ is normally distributed with mean $\mu_{-k} + \Sigma_{-k, k}\Sigma^{-1}_{k,k} (x-\mu_k)< \mu_{-k} + x - \mu_k$. 
On the other hand, 
$$
\frac{p_{a,K-1}(x)}{p_{a,K}(x)} = \exp(- (x - \mu_K)^2/(2\sigma^2) + (x - \mu_{K-1})^2/(2\sigma^2))\neq 1, 
$$
and
$$
\lim_{x\to\infty}\frac{\Pr(S_{1:K-1}\in(\theta^0_{1:K-1}(x), \theta^1_{1:K-1}(x))|S_K = \theta^0_K(x), H_a)}
{\Pr(S_k\leq \theta^0_K(x), \text{ and }S_{1:K-2}\in(\theta^0_{1:K-2}(x), \theta^1_{1:K-2}(x))|S_{K-1} = \theta^1_{K-1}(x), H_a)}  = 1, 
$$
and thus 
$$
\lim_{x\to\infty}\frac{\partial g_a(\theta(x))/\partial \theta^0_K}{\partial g_a(\theta(x))/\partial \theta^1_{K-1}} = \lim_{x\to\infty}\frac{p_{a,K}(x)}{p_{a,K-1}(x)} \neq 1.
$$
Since both $\frac{\partial g_0(\theta(x))/\partial \theta^0_K}{\partial g_0(\theta(x))/\partial \theta^1_{K-1}}$ and  $\frac{\partial g_a(\theta(x))/\partial \theta^0_K}{\partial g_a(\theta(x))/\partial \theta^1_{K-1}}$ are continuous in $x$, it follows that there exists some $\hat x\in\mathbb{R}$ such that $\frac{\partial g_0(\theta(\hat x))/\partial \theta^0_K}{\partial g_0(\theta(\hat x))/\partial \theta^1_{K-1}}\neq \frac{\partial g_a(\theta(\hat x))/\partial \theta^0_K}{\partial g_a(\theta(\hat x))/\partial \theta^1_{K-1}}$. So $\det((l_0(\theta(\hat x)), l_a(\theta(\hat x))))\neq 0$.  

Now, consider the two sided tests. For convenience, we let $S_k = |\tilde S_k|$ for each $k\in[K]$ where $\tilde S = (\tilde S_k)^K_{k=1}$ has multivariate normal distribution with mean $0$ and covariance $\Sigma$ under $H_0$, and mean $\mu = (\mu_k)^K_{k=1}$ and covariance $\Sigma$ under $H_a$. Let $p_0$ and $p_{a,k}$ be the density functions of the marginal distributions of $\tilde S_k$ under $H_0$ and $H_a$. Let $S_{1:k} = (S_i)^k_{i=1}$ and $\theta_{1:k} = (\theta_i)^k_{i=1}$. Then we have, 
\begin{align*}
    &\partial g_0(\theta)/\partial \theta^0_K = p_0(\theta^0_K)\Pr(S_{1:K-1}\in(\theta^0_{1:K-1}, \theta^1_{1:K-1})|\tilde S_K = \theta^0_K, H_0) \\
    &\quad\quad\quad+ p_0(-\theta^0_K)\Pr(S_{1:K-1}\in(\theta^0_{1:K-1}, \theta^1_{1:K-1})|\tilde S_K = -\theta^0_K, H_0),
    \\
    &\partial g_0(\theta)/\partial \theta^1_{K-1} = p_0(\theta^1_{K-1})\Pr(S_K\leq \theta^0_K, \text{ and }S_{1:K-2}\in(\theta^0_{1:K-2}, \theta^1_{1:K-2})|\tilde S_{K-1} = \theta^1_{K-1}, H_0)\\
    &\quad\quad\quad+
    p_0(-\theta^1_{K-1})\Pr(S_K\leq \theta^0_K, \text{ and }S_{1:K-2}\in(\theta^0_{1:K-2}, \theta^1_{1:K-2})|\tilde S_{K-1} = -\theta^1_{K-1}, H_0)
    \\
    &\partial g_a(\theta)/\partial \theta^0_K = p_{a,K}(\theta^0_K)\Pr(S_{1:K-1}\in(\theta^0_{1:K-1}, \theta^1_{1:K-1})|\tilde S_K = \theta^0_K, H_a) \\
    &\quad\quad\quad
    + p_{a,K}(-\theta^0_K)\Pr(S_{1:K-1}\in(\theta^0_{1:K-1}, \theta^1_{1:K-1})|\tilde S_K = -\theta^0_K, H_a),
    \\
    &\partial g_a(\theta)/\partial \theta^1_{K-1} = p_{a,K-1}(\theta^1_{K-1})\Pr(S_K\leq \theta^0_K, \text{ and }S_{1:K-2}\in(\theta^0_{1:K-2}, \theta^1_{1:K-2})|\tilde S_{K-1} = \theta^1_{K-1}, H_a)\\
    &\quad\quad\quad
    +
    p_{a,K-1}(-\theta^1_{K-1})\Pr(S_K\leq \theta^0_K, \text{ and }S_{1:K-2}\in(\theta^0_{1:K-2}, \theta^1_{1:K-2})|\tilde S_{K-1} = -\theta^1_{K-1}, H_a).
\end{align*}

Let $\theta(x) = (\theta^0_1(x), \cdots, \theta^0_K(x), \theta^1_1(x), \cdots, \theta^1_K(x))$ with $\theta^0_k(x) = 0$ for all $k\in[K-1]$ and $\theta^0_K(x) = \theta^1_k(x) = x$ for all $k\in[K]$. It is then sufficient to show that there exists some $x$ such that $l_0(\theta(x))$ and $l_a(\theta(x))$ are linearly independent. Then, since $\theta^0_K(x) = \theta^1_{K-1}(x) = x$ and $p_0(x) = p_0(-x)$, 
\begin{align*}
&\lim_{x\to\infty}\frac{\partial g_0(\theta(x))/\partial \theta^0_K}{\partial g_0(\theta(x))/\partial \theta^1_{K-1}} \\
= &\lim_{x\to\infty}
\big(\Pr(S_{1:K-1}\in(\theta^0_{1:K-1}(x), \theta^1_{1:K-1}(x))|\tilde S_K = \theta^0_K(x), H_0) \\
&+ \Pr(S_{1:K-1}\in(\theta^0_{1:K-1}(x), \theta^1_{1:K-1}(x))|\tilde S_K = -\theta^0_K(x), H_0)\big)\\
&/\big(\Pr(S_K\leq \theta^0_K(x), \text{ and }S_{1:K-2}\in(\theta^0_{1:K-2}(x), \theta^1_{1:K-2}(x))|\tilde S_{K-1} = \theta^1_{K-1}(x), H_0) \\
&+ 
\Pr(S_K\leq \theta^0_K(x), \text{ and }S_{1:K-2}\in(\theta^0_{1:K-2}(x), \theta^1_{1:K-2}(x))|\tilde S_{K-1} = -\theta^1_{K-1}(x), H_0)\big) \\
=& 1.
\end{align*}
On the other hand, 
$
\frac{p_{a, k}(-x)}{p_{a, k}(x)} = \exp((x-\mu_k)^2/(2\sigma^2) - (-x - \mu_k)^2/(2\sigma^2)), \,\frac{p_{a, K-1}(x)}{p_{a, K}(x)} = \exp(-(x-\mu_{K-1})^2/(2\sigma^2) + (x - \mu_{K})^2/(2\sigma^2))$, 
and 
$
\frac{p_{a, K-1}(-x)}{p_{a, K}(x)} = \exp(-(-x-\mu_{K-1})^2/(2\sigma^2) - (x - \mu_K)^2/(2\sigma^2))$. 
It thus follows that 
\begin{align*}
&\lim_{x\to\infty}\frac{\partial g_a(\theta(x))/\partial \theta^0_K}{\partial g_a(\theta(x))/\partial \theta^1_{K-1}}\\
= &\lim_{x\to\infty}
\big(\Pr(S_{1:K-1}\in(\theta^0_{1:K-1}, \theta^1_{1:K-1})|\tilde S_K = \theta^0_K, H_a) \\
&+ \frac{p_{a,K}(-x)}{p_{a,K}(x)}\Pr(S_{1:K-1}\in(\theta^0_{1:K-1}, \theta^1_{1:K-1})|\tilde S_K = -\theta^0_K, H_a)\big)\\
&/\big(\frac{p_{a,K-1}(x)}{p_{a,K}(x)}\Pr(S_K\leq \theta^0_K, \text{ and }S_{1:K-2}\in(\theta^0_{1:K-2}, \theta^1_{1:K-2})|\tilde S_{K-1} = \theta^1_{K-1}, H_a) \\
&+ 
\frac{p_{a,K-1}(-x)}{p_{a,K}(x)}\Pr(S_K\leq \theta^0_K, \text{ and }S_{1:K-2}\in(\theta^0_{1:K-2}, \theta^1_{1:K-2})|\tilde S_{K-1} = -\theta^1_{K-1}, H_a)\big) \\
\neq& 1.
\end{align*}
Thus we have shown that there exists some $\theta\in\mathbb{R}^{2K}$ such that $l_0(\theta)$ and $l_a(\theta)$ are linearly independent for both one-sided and two-sided tests. 
In the meanwhile, the determinant $\det((l_0(\theta), l_a(\theta)))$ is a real analytic function. Then, according to Proposition~0 of \citet{mityagin2015zero}, the set $\{\theta\in\mathbb{R}^{2K}: \det((l_0(\theta), l_a(\theta))) = 0\}$ has a zero measure. As $\det((l_0(\theta), l_a(\theta))) = 0$ is a necessary condition for $\nabla g_0(\theta)$ to be linearly dependent with $\nabla g_a(\theta)$, we have that $\{\theta\in\mathbb{R}^{2K}: \text{MFCQ fails at }\theta\}$ has a zero Lebesgue measure.

\end{appendices}

% \end{titlepage}
% aketitle

\end{document}